\documentclass[11pt,letter]{article}
\usepackage{fullpage}

\usepackage{amsmath}
\usepackage{xspace}
\usepackage{amsthm}

\usepackage[T2A,T1]{fontenc}
\usepackage[utf8]{inputenc}
\usepackage{babel}
\usepackage{indentfirst,verbatim} 
\usepackage{listings} 
\usepackage[normalem]{ulem}
\usepackage{soul}
\usepackage{amsfonts}
\usepackage{amsthm}
\usepackage{hyperref}

\usepackage{amsmath,amsfonts,amsthm}
\usepackage{mathbbol}
\usepackage{amsthm}
\usepackage{graphicx,color}
\usepackage{boxedminipage}
\usepackage[linesnumbered, boxed, noend]{algorithm2e}
\usepackage{framed}
\usepackage{thmtools}
\usepackage{thm-restate}
\usepackage{xspace}
\usepackage[disable]{todonotes}
\usepackage{pgfplots}
\usetikzlibrary{calc}
\usetikzlibrary{shapes}
\usetikzlibrary{arrows.meta}
\usepackage{colortbl}
\usepackage{tcolorbox} 

\usepackage{xifthen}
\usepackage{tabularx}
\usepackage{capt-of}
\usepackage{thmtools}
\usepackage{thm-restate}
\usepackage{subcaption}
\usepackage{booktabs}
\usepackage{calc}
\usepackage{cleveref}
\usepackage{xcolor}

\theoremstyle{plain}

\newtheorem{observation}{Observation} 
 
\newtheorem{lemma}{Lemma} 
\newtheorem{claim}{Claim} 
\newtheorem{proposition}{Proposition} 

\theoremstyle{definition}
\newtheorem{definition}{Definition}

\newenvironment{claimproof}{\medskip\noindent \emph{Proof of Claim~\theclaim.}  }{\hfill\cqed\medskip}

\newcommand{\cqed}{\renewcommand{\qedsymbol}{$\lrcorner$}\qed}
\newcommand{\pname}{\textsc}
\newcommand{\polyn}{n^{\Oh(1)}}
\newcommand{\polyI}{|\mathcal{I}|^{\Oh(1)}}
\renewcommand{\subset}{\subseteq}

\newlength{\RoundedBoxWidth}
\newsavebox{\GrayRoundedBox}
\newenvironment{GrayBox}[1]%
{\setlength{\RoundedBoxWidth}{.93\textwidth}
	\def\boxheading{#1}
	\begin{lrbox}{\GrayRoundedBox}
		\begin{minipage}{\RoundedBoxWidth}}%
		{   \end{minipage}
	\end{lrbox}
	\begin{center}
		\begin{tikzpicture}%
			\node(Text)[draw=black!20,fill=white,rounded corners,%
			inner sep=2ex,text width=\RoundedBoxWidth]%
			{\usebox{\GrayRoundedBox}};
			\coordinate(x) at (current bounding box.north west);
			\node [draw=white,rectangle,inner sep=3pt,anchor=north west,fill=white] 
			at ($(x)+(6pt,.75em)$) {\boxheading};
		\end{tikzpicture}
\end{center}}     

\newenvironment{defproblemx}[2][]{\noindent\ignorespaces%
	\FrameSep=6pt%
	\parindent=0pt%
	\vspace*{-1.5em}
	\ifthenelse{\isempty{#1}}{%
		\begin{GrayBox}{#2}%
		}{%
			\begin{GrayBox}{#2 parameterized by~{#1}}%
			}
			\begin{tabular*}{\textwidth}{@{\hspace{.1em}} >{\itshape} p{1.8cm} p{0.8\textwidth} @{}}%
			}{
			\end{tabular*}%
		\end{GrayBox}%
		\ignorespacesafterend
	}  
	
	\newcommand{\defparproblema}[4]{
		\begin{defproblemx}[#3]{#1}
			Input:  & #2 \\
			Task: & #4
		\end{defproblemx}
	}%

	\newcommand{\Oh}{\mathcal{O}}

	\newcommand{\probMASMST}{\pname{Maximum Weight Acyclic Subgraph Above MaxST}\xspace}
	\newcommand{\probMASMSTshort}{\pname{MAS/MaxST}\xspace}
	\newcommand{\probWDFAS}{\pname{Weighted Directed Feedback Arc Set}\xspace}
	\newcommand{\probWDFASshort}{\pname{WDFAS}\xspace}
	
	\newcommand{\probMAS}{\pname{Maximum Acyclic Subgraph}\xspace}		
	
	\DeclareMathOperator{\operatorClassNP}{{\sf NP}}
	\newcommand{\classNP}{\ensuremath{\operatorClassNP}}

	\DeclareMathOperator{\operatorClassFPT}{{\sf FPT}\xspace}
	\newcommand{\classFPT}{\ensuremath{\operatorClassFPT}\xspace}
	\DeclareMathOperator{\operatorClassW}{{\sf W}}
	\newcommand{\classW}[1]{\ensuremath{\operatorClassW[#1]}}

	\DeclareMathOperator{\operatorClassXP}{{\sf X}P\xspace}
	\newcommand{\classXP}{\ensuremath{\operatorClassXP}\xspace}
	\lstnewenvironment{code}{\lstset{language=Haskell,basicstyle=\small}}{}
	\DeclareMathOperator{\MinST}{MinST}
	\DeclareMathOperator{\MaxST}{MaxST}
	\DeclareMathOperator{\MSF}{MaxST}
	
	\DeclareMathOperator{\Inv}{Inv}
	\DeclareMathOperator{\Profit}{\bf p}
	\DeclareMathOperator{\RemainingProfit}{\bf rp}
	
	\DeclareMathOperator{\wt}{{\bf w}}
	
	\newcommand{\intwtdef}{\wt:E(G)\to \mathbb{Z}_{\ge 1}}

	\theoremstyle{plain}
	\newtheorem{rrule}{Reduction Rule}
	\newtheorem{problem}{Open Question}

\title{\textsc{Exploiting Spanning Trees for Directed Acyclicity}\thanks{Sergei Khargeliia was supported by PJSC <<Gazprom Neft>>, aggr. \# \foreignlanguage{russian}{ГПН-26/03000/00502/Р}.} \thanks{Danil Sagunov was supported by the Ministry of Science and Higher Education of the Russian
		Federation (agreement 075-15-2025-344 dated 29/04/2025 for Saint Petersburg Leonhard Euler
		International Mathematical Institute at PDMI RAS).}
		 }
\author{
	Sergei Khargeliia\thanks{
		ITMO University,
		Saint Petersburg State University,
		St. Petersburg,
		Russia.}\\sergey.khargelia@gmail.com
	\and
	Danil Sagunov\thanks{
		Markov Lab,
		Saint Petersburg State University,
		St. Petersburg,
		Russia.
	}\\danilka.pro@gmail.com
}
\date{}
\usetikzlibrary{arrows}

\begin{document}

\maketitle


 \begin{abstract}
 	
 	We study the weighted case of the  \textsc{Maximum Acyclic Subgraph (MAS)} problem, where each edge of a given directed graph has a positive weight assigned, and the task is to find a maximum-weight acyclic edge set.
 	The famous and well-studied random ordering lower bound guarantees the existence of an acyclic set that gives at least the half of the total edge weight.
  
 	The maximum spanning tree (MaxST) guarantee, which is the weight of a maximum-weight acyclic subgraph of the underlying undirected graph of $G$, is another natural lower bound for the weight of an acyclic subgraph.
 	A solution of this weight dominates the random ordering solution on instances where MaxST spans the most of the total edge weight.
 	
 	Our main contribution are two parameterized algorithms that find acyclic subgraphs of total weight larger than the weight of the MaxST of $G$.
 	Both our algorithms find a solution of total weight at least $\MaxST(G)+k$, for a given integer $k\ge 0$, or report that it does not exist, and
 	\begin{itemize}
 		\item First of our algorithms runs in time $2^{k^{\Oh(1)}}\cdot\polyI$ and works when all weights are integers;
 		\item Our second algorithm handles rational weights not less than $1$, and its running time is upper-bounded by $n^{k^{\Oh(1)}}\cdot \polyI$.
 	This positive result is rather surprising since solving \textsc{MAS} above the random ordering lower bound is \classNP-hard in the same rational weights scenario, when $k=1$.
 	 	\end{itemize} 	
 	 	
 	 	Our findings unravel intricate connections between structure of MaxSTs and directed cycles, use perfect graph theorem to tackle rational weights, and raise graph-theoretic questions that are interesting on their own.
 	Of another importance, this is one of the few examples of positive ``above guarantee'' results for a weighted problem on directed graphs, especially for rational weights.
\end{abstract}		

\newpage
\tableofcontents

\newpage

 \section{Introduction}
 
 In the \textsc{Maximum Acyclic Subgraph} problem, \textsc{MAS} for short, we are given a directed graph (digraph) $G$ with $n$ vertices and $m$ edges, and an integer $k\ge 0$, and the task is to find an acyclic subgraph of $G$ that contains at least $k$ edges.
 In the dual version of the problem, alternatively, one can ask to delete at most $k'=m-k$ edges from $G$.
 This formulation is well-known as \textsc{Feedback Arc Set} and comes from the seminal work of Karp \cite{karp1975computational}.\todo{speak about motivation, applications, recent research}
 In this paper, we study the weighted version of \textsc{MAS}, here we discuss the unweighted case first.
 
  One can easily see that any digraph contains an acyclic subgraph with at least $m/2$ edges.
  To construct such a subgraph in polynomial time, take an arbitrary ordering of $V(G)$, and take all edges that go from the left to the right into solution.
  If there are less than $m/2$ such edges, take all edges that go from the right to the left in the solution instead.
  
  This is known as the \emph{random ordering} lower bound and is quite significant to \textsc{MAS} from the perspective of approximation algorithms.
  Obviously, it provides a polynomial-time $2$-approximation algorithm for \textsc{MAS}.
   In \cite{mas-above-half-hardness}, Guruswami, Manokaran and Raghavendra proved that this is tight and no better approximation ratios for \textsc{MAS} are possible, a \emph{``first tight inapproximability result for an ordering problem''}, under the Unique Games Conjecture~\cite{UGC}. 
  One year before that, Charikar, Makarychev and Makarychev \cite{mas-above-half-approx} designed a polynomial-time algorithm that finds an acyclic subgraph with at least $(\frac{1}{2}+\Omega(\delta/\log n))\cdot m$ edges in any given $n$-vertex graph that admits an acyclic subgraph with at least $(\frac{1}{2}+\delta)\cdot m$ edges, for any  $\delta>0$.

  The same lower bound was the original starting point of a quite successful concept in parameterized complexity.
   It was initiated from the search for an efficient algorithm that finds an acyclic subgraph with at least $m/2+k$ edges.
  The question of an existence of such an algorithm was posed as an open problem by Raman and Saurabh in \cite{RamanSaurabhMAS}, who obtained partial results.
  Later, Mahajan, Raman and Sikdar \cite{MahajanAboveGuarantee} addressed this question again among other similar ones.
  Importantly, \cite{MahajanAboveGuarantee} is the first paper that addresses parameterizations above (or below) tight lower (or upper) bounds as a general concept, now known as \emph{above-guarantee} or \emph{below-guarantee} parameterizations (see, e.g.\ a survey by Gutin and Mnich \cite{gutin2024surveygraphproblemsparameterized} on this topic).
  Finally, in \cite{GutinMASAboveHalfWeighted}, Gutin, Kim, Szeider and Yeo gave a positive answer to the question.
  They showed that the problem of finding an acyclic set of at least $m/2+k$ edges admits a $2^{\Oh(k
  	\log k)}\cdot\polyn$-running-time fixed-parameter tractable (\classFPT) algorithm.\todo{improve AG discussion?}

Another tight lower bound for \textsc{MAS} was considered independently by Mnich, Philip, Saurabh and Such{\`y} \cite{MnichPoljakTurzik} and by Crowston, Gutin and Jones \cite{PoljakTurzikCrowston}.
  It comes from the work of Poljak and Turz\'{i}k from 1986 \cite{PoljakTurzikOriginal}, and guarantees that if $G$ is an oriented\footnote{$G$ is directed, and each pair of vertices is connected by at most one directed edge, going in either of the two directions, but not both simultaneously.} weakly-connected\footnote{$G$ is connected if we remove edge orientations.}  graph, then $G$ has an acyclic subgraph with at least $m/2+(n-1)/4$ edges.
  In \cite{MnichPoljakTurzik, PoljakTurzikCrowston} it was shown that the problem of finding an acyclic subgraph with at least $m/2+(n-1)/4+k$ edges in a given oriented graph is fixed-parameter tractable.
  
Later, Etscheid and Mnich \cite{etscheid2018linear} showed that a variety of problems (with \textsc{MAS} and \textsc{Max Cut} among them) admit linear kernels, when parameterized above the Poljak--Turz\'{i}k bound.
  In particular, they showed a kernel with $\mathcal{O}(k)$ vertices for \textsc{MAS} above Poljak--Turz\'{i}k, and improved the running time to $2^{\mathcal{O}(k)}\cdot n^{\Oh(1)}$.
  Whereas the Poljak--Turz\'{i}k bound is applicable only to oriented graphs, this parameterization is superior to the parameterization above $m/2$.
  Finding an acyclic subgraph with at least $m/2+k$ edges in an arbitrary digraph can be simply reduced to the same problem on an oriented graph.
  
  
  There is one particular lower bound for \textsc{MAS} that has been overlooked in this line of research.
  This lower bound is the \emph{spanning tree} guarantee, and it speaks for itself: any weakly-connected digraph admits an acyclic graph with at least $n-1$ edges.
The spanning tree bound is also (algorithmically) inferior to the Poljak--Turz\'{i}k bound.
    
    The situation with these three guarantees changes slightly when we introduce edge weights in $G$.   \todo{ref to the weighted case studies}
    The algorithm of \cite{GutinMASAboveHalfWeighted} for \textsc{MAS} above random ordering works in the multigraph setting when we are allowed to have several copies of the same edge in $G$.
    This can alternatively be seen as a weighted scenario.
    Instead of having several copies of the same edge, we keep $G$ a simple digraph, and we assign a positive integer weight $\wt(e)$ to each edge $e\in E(G)$.
    Then the random ordering guarantee is at least the half of the total edge weight, $\wt(E(G))/2$.
     Within this context, the mentioned algorithm of \cite{GutinMASAboveHalfWeighted} finds an acyclic subgraph of total edge weight at least $\wt(E(G))/2+k$ in $2^{\Oh(k\log k)}\cdot\polyn$ running time.
    
The Poljak--Turz\'{i}k lower bound was originally formulated in  \cite{PoljakTurzikOriginal} for the weighted case, and it essentially guarantees that an edge-weighted \emph{oriented} graph $G$ contains an acyclic subgraph of weight at least $\wt(E(G))/2+\wt(\MinST(G,\wt))/4$, where $\MinST(G,\wt)$ is the \emph{minimum} spanning tree of the underlying undirected graph of $G$ (or its minimum spanning forest, if $G$ is not weakly-connected).
    In  \cite{etscheid2018linear}, the authors posed an open question on whether the efficient \classFPT-algorithms above Poljak--Turz\'{i}k can be generalized to the weighted case.\todo{note that PT generalizes RO?}
    
    But are these two lower bounds that good when the edge weights are not evenly distributed?
    If, for example, the graph has just one heavy edge, whose weight is greater than the sum of all other weights, then this edge alone is better than any of the two guarantees can provide.
    
	The spanning tree guarantee becomes the \emph{maximum} spanning tree (MaxST) guarantee in the weighted setting.
	Clearly, $G$ admits an acyclic subgraph of weight at least $\MaxST(G, \wt)$, a weight of the maximum spanning tree of the underlying undirected graph of $G$.
	This lower bound captures the case in the paragraph above and is better on instances where edge weights are concentrated around a sparse structure.
	For instance, it dominates both the random ordering and the Poljak--Turz\'{i}k guarantees on instances where $\MaxST(G,\wt)$ is at least $\frac{3}{4}\cdot \wt(E(G))$.	
	The MaxST lower bound is the central guarantee for \textsc{MAS} in our work.
We give a formal definition to the corresponding parameterized problem below.

\defparproblema{\probMASMST (\probMASMSTshort)}{
	Weakly-connected $n$-vertex digraph $G$ with integer weight function $\intwtdef$ assigned to its edges, and integer $k\ge 0$.
}{}{	Find an acyclic subgraph $G'$ of $G$ such that $\wt(E(G'))$ is at least $\wt(\MaxST(G,\wt))+k$, or report that $G'$ does not exist.}

This parameterization differs significantly from the Poljak--Turz\'{i}k one.
First, we do not have to enforce that $G$ is oriented for our guarantee to work.
The restriction on the weak connectivity of $G$ can also be lifted, because $G$ can always be made weakly-connected by adding bridges between connected components of $G$.
Second, these bounds are generally incomparable since they behave differently on sparse and dense graphs.

\subparagraph{Our contribution.}
Our first main result is an \classFPT-algorithm for \probMASMSTshort.
It is obtained by a novel approach that combines several constructive ideas and precise analysis of the edges outside the maximum spanning tree.
Integrality of the weight function is crucial, since the edges of MaxST are classified according to their \emph{profits}---the amount of extra weight we could possibly get from edges outside of MaxST if we remove this edge of MaxST from the solution.
Since all profits are integral, and large profits guarantee the solution, we are able to achieve only $k^{\Oh(1)}$ equivalence classes, and we are interested only in the ``most suitable'' edges in each class.
We then perform branching on which edges should be deleted from the MaxST.

\begin{restatable}{theorem}{thmFPTint}\label{thm:main-result}\label{thm:integral-fpt}
	$\probMASMSTshort(\mathbb{Z}_{\ge 1})$ admits an algorithm with $2^{k^{\Oh(1)}}\cdot\polyI$ running time.
\end{restatable}


As an example, we can use \Cref{thm:main-result} to find an (unweighted) acyclic subgraph that contains two prescribed edges $e_1, e_2$ (that are not opposites) and at least $(n-3)+k$ other edges.
This can be extended to more than two edges, but only if they form an undirected forest.
Note that the random ordering or Poljak--Turz\'{i}k guarantees cannot provide us such use since they could only ensure that the half of the prescribed edges belong to the solution.


\Cref{thm:main-result} works only in a restricted setting of integer weights.
What can we do for the general case when weights are allowed to be rational numbers?
In this setting, the choice of the parameter $k$ is not very obvious.
If we allow arbitrary rational edge weights in the definition of \probMASMSTshort, then we can reduce an arbitrary weakly-connected instance of (unweighted) \textsc{MAS} to an instance of \probMASMSTshort with $k=1$.
Hence it is \classNP-hard for $k=1$.

To overcome this trivial obstacle and capture the actual parameterized complexity, we don't allow $k$ to be too large compared to edge weights.
We keep the original definition of \probMASMSTshort and restrict weights to be rationals \emph{not less than one}, and denote this version as $\probMASMSTshort(\mathbb{Q}_{\ge 1})$.
Our second main contribution is an \classXP-algorithm for this problem.\todo{don't forget to mention NP-hardness for RO}

\begin{restatable}{theorem}{thmXPRat}\label{thm:main-result-rational}\label{thm:rational-xp}
	$\probMASMSTshort(\mathbb{Q}_{\ge 1})$ admits an algorithm with $n^{k^{\Oh(1)}}\cdot\polyI$ running time.
\end{restatable}

To prove \Cref{thm:main-result-rational}, we use some of the structural profit-related insights that are used to prove \Cref{thm:main-result}.
Main obstacle when transitioning from integer to rational weights is that edge weights differences are not integers anymore, and can be arbitrarily close to $0$.
This raises two new challenges that we have to deal with, even to obtain an \classXP-algorithm.
First, the profits of the edges of MaxST could be very small and even to obtain just ``$+1$'' excess weight, we could have to remove a lot of them.
Second, the profits are not integers and hence cannot be simply split into $f(k)$ equivalence classes by their profit.
We resolve these challenges by providing two new structural insights (alternative to ``profit classification'' used for the integral case) that focus on the structure of directed paths of the MaxST instead.

\subparagraph{Related work.}
\textsc{MAS} and \textsc{Feedback Arc Set} are known for notorious hardness from the exact exponential algorithms perspective.
For both weighted and unweighted versions of these problems, the best known exact algorithms are still the ones that run in $2^n\cdot \polyn$ time.
On the positive side, Kim, Kratsch, Pilipczuk and Wahlstr\"{o}m in their recent work \cite{KimWFAS} introduced the \emph{flow augmentation} technique.
They showed that, if $k$ is the number of edges that we are allowed to remove, then  \textsc{Weighted Feedback Arc Set} is solvable in $2^{k^{\Oh(1)}}\cdot\polyI$ time. 
We use this algorithm in the last step of our proof of \Cref{thm:main-result}.

While the above-guarantee approach is now well developed in parameterized complexity, relatively few works study weighted guarantees specifically. One of the few is the work by Gutin and Patel  \cite{gutin-tsp-above-average}, which investigates \classFPT-algorithms for finding weighted Hamiltonian cycles below the average cycle weight in complete undirected graphs with integer edge weights. The authors leave the complexity for complete digraphs as an open question. There is also a related line of works \cite{Crowston2010, crowston_et_al:LIPIcs.FSTTCS.2011.229, Crowston2012} on weighted CSPs for linear systems over $\mathbb{F}_2$, parameterized above or below guarantees, where the parameter $k$ is also the excess weight.













\subparagraph{Organization of the paper.}
In the next section, we introduce the notation used throughout the paper and provide some preliminary results. 
In \Cref{sec:classification}, we introduce definitions and properties of edges around $\MaxST(G,\wt)$ that are fundamental to our approach.
In \Cref{sec:overview}, we provide our proof of \Cref{thm:main-result}, supporting it with necessary intuition and discussion.
In \Cref{sec:xp-rational-weights}, we give the proof of \Cref{thm:main-result-rational} in the same way.

\section{Preliminaries}

\subparagraph{Notation.}
We use standard graph terminology and notation (we refer the reader to the book of Diestel~\cite{Diestel}).
We also use several additional notions.
We use $G-S$ to denote a graph obtained by deleting all edges of $S$ from $G$.
We write $G+S$ to denote a graph obtained by adding all edges of $S$ to $G$.
We refer to both undirected and directed edges as just \emph{edges}, the actual variant is always clear from the context.

For a graph $(G, \wt)$ with edge weights, we always assume that $\wt$ is additively extended to all subsets of edges. 
That is, for each $S \subseteq E(G)$, we have $\wt(S) = \sum_{e \in S} \wt(e)$. 
For a subgraph $X \subseteq G$, we also shorten the notation $\wt(E(X))$ to just $\wt(X)$. 

For the comprehensive introduction to parameterized algorithms and terminology of parameterized complexity, we refer the reader to the Parameterized Algorithms book~\cite{CyganFKLMPPS15}.

\subparagraph{Partially ordered sets.}

A \emph{partially ordered set} (poset for short) is an ordered pair $\mathcal{P} = (X, \preceq)$ where $X$ is a set and $\preceq$ is a partial order on $X$. 
A relation $\preceq^*$ is called a \emph{linear extension} of $\mathcal{P}$ if it is a total order on $X$ that is compatible with $\preceq$, that is, for all $x, y \in X$,  $x \preceq y$ implies $x \preceq^* y$. 

The \emph{order dimension} of a poset $\mathcal{P} = (X, \preceq)$ is the least integer $t$ for which there exists a family $(\preceq_1, \preceq_2, \ldots, \preceq_t)$ of linear extensions of $\mathcal{P}$ with the following property: for each $x, y \in X$, it holds that $x \preceq y$ if and only if $x \preceq_i y$ for all $i$.
This concept was introduced by Dushnik and Miller in~\cite{Dushnik1941PartiallyOS}. 
We will use the following result on the order dimension. 

\begin{proposition}[\cite{TrotterPlanarPosets}, \cite{abram2026dimensionunicycleposets}]\label{prop:tree-order-dimension}
	Let $T$ be an oriented tree, and let $\mathcal{P} = (V(T), \preceq)$ be a poset such that for each $u, v \in V(T)$, $u \preceq v$ if and only if there exists a directed path from $u$ to $v$ in $T$.
	Then the order dimension of $\mathcal{P}$ is at most three, and the corresponding linear extensions $(\preceq_1, \preceq_2, \preceq_3)$ of $\mathcal{P}$ can be constructed in polynomial time. 
\end{proposition}

Trotter and Moore~\cite{TrotterPlanarPosets} proved that the order dimension of such posets is at most three, but they did not explicitly state the existence of a polynomial-time algorithm that constructs the corresponding linear extensions.
Although such an algorithm follows from their proof, we refer to a recent work of Abram and Segovia for a more direct construction (see Corollary~6.3 in~\cite{abram2026dimensionunicycleposets}).

\section{Classification of directed edges and their basic properties}\label{sec:classification}

In this section, we introduce the basic toolbox that we use extensively throughout our proof.
Our tools are mostly related to classification of edges of $G$ relative to its MaxST and several useful properties surrounding this classification for \emph{arbitrary positive edge weights}.




\subsection{Maximum spanning tree, allowed and blocked edges}
We start with properly defining the ``guaranteed'' subgraph of $G$.
We informally refer to it as a maximum spanning tree, or MaxST for short. 

\begin{definition}[Maximum spanning tree]
	Let $G$ be a weakly-connected $n$-vertex digraph with an edge set $\{e_1, \ldots, e_m\}$ and edge weights $\wt$. 
	By $\MaxST(G,\wt)$ we denote a subgraph $T$ of $G$, such that:
	\begin{itemize}
		\item $|V(T)| = n$ and $|E(T)| = n - 1$;
		\item $T$ is weakly-connected;
		\item $\wt(E(T))$ is maximum possible;
		\item a sorted sequence of indices of edges from the set $E(T)$ is lexicographically smallest possible among all subgraphs satisfying the previous constraints.
	\end{itemize} 
	Throughout the paper, we use $T$ as a shortcut for $\MaxST(G,\wt)$. 
\end{definition}

Note that in the definition above, we put the last property in order that $T$ is defined unambiguously. 
Moreover, $T$ can still be found in polynomial time via the standard argument: we first find the maximum weight $w$ of a spanning tree of $G$. 
Then we construct $T$ iteratively, starting with the empty subgraph. 
In each iteration, we find an edge $e_i$ of minimum index such that $T + e_i$ can be extended to a spanning tree of $G$ of weight $w$, and add this $e_i$ to $T$. 

Naturally, some edges outside $E(T)$ do not form a cycle, if added to $T$.
For example, if $T+e$ is acyclic for $e \in E(G)\setminus E(T)$, then $(G,\wt,k)$ is a yes-instance of \probMASMSTshort for any $k\le \wt(e)$.
We proceed with a proper formal definition that includes the concept of such edges.

\begin{definition}[Classification of graph edges relative to MaxST] 
	For a weakly-connected digraph $G$ with edge weights $\wt$, we say that
	\begin{itemize}
		\item Each edge $e \in E(T)$ is a \textit{tree} edge.
		\item Each edge $e \in E(G) \setminus E(T)$ such that $T + e$ has no directed cycles is an \textbf{\textit{allowed}} edge.
		\item Each edge $e \in E(G) \setminus E(T)$ such that $T + e$ has a directed cycle is a \textbf{\textit{blocked}} edge.
	\end{itemize}
	
	By $A(G, \wt)$ we denote the set of all \textit{allowed} edges of $G$ and by $B(G, \wt)$ we denote the set of all \textit{blocked} edges of $G$.
	We shorten this notation to just $A$ and $B$, unless the arguments are different from $(G,\wt)$.
\end{definition}

\begin{figure}
	\centering
	\begin{tabular}{rccrc}
		(a)&\includegraphics[scale=2]{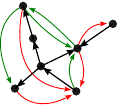} & \hspace{1cm} & (b)&
		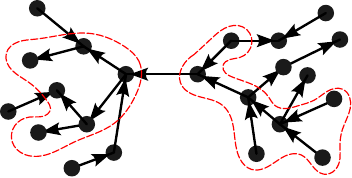
	\end{tabular}
	
	\caption{(a) An instance with six tree edges (thick black), four allowed edges (green) and four blocked edges (red).
		(b) $\Inv(e)$ is the set of all edges that start in $X$ (vertices reachable from $e$ in $T$) and end in $Y$ (vertices $e$ is reachable from in $T$).}\label{fig:inv-edges}
\end{figure}

The main challenge in our two algorithms is to handle both edge types simultaneously.
In the rest of this section, we focus on blocked edges.


\subsection{Inverse edges and their properties}

Note that removing an edge from $T$ can ``unblock'' some edges in $B$.
That is, if $b$ is blocked, and $e\in E(T)$ belongs to a directed cycle in $T+b$, then $(T-e)+b$ is acyclic.
In this situation, we say that $b$ is an \emph{inverse edge} of $e$.
We give a formal definition below.

\begin{definition}[Inverse edges]
	For a weakly-connected digraph with positive edge weights $(G,\wt)$ and a tree edge $e\in E(T)$, by
	 $\Inv(G,\wt,e) = \{b \in B: T - e + b \text{ is acyclic}\}$
	 we denote the set of all \emph{inverse} edges of $e$ in $G$.
	 
	 For a tree edge set $S\subset E(T)$, by
	 $\Inv(G,\wt,S)= \bigcup_{e\in S} \Inv(G,\wt,e)$
	 we denote the union of all inverse edges among all edges in $S$.
	 We shorten the notation to just $\Inv(e)$ and $\Inv(S)$.
\end{definition}

We start by observing basic properties of inverse edges that follow from their definition. 

\begin{observation}\label{obs:basic-inverse-edge-property}
	Let $b \in B$ be a blocked edge that goes from $u$ to $v$.
	Then $T$ contains a directed path $P$ from $v$ to $u$.
	Moreover, for a tree edge $e$, an edge $b \in \Inv(e)$ iff $e$ lies on $P$.
\end{observation}
\begin{proof}
	Since $b$ is a \textit{blocked} edge, the subgraph $T + b$ contains a directed cycle $C$, and $b$ lies on it (because $T$ is acyclic).
	Then there is a path $P$ in $T$ from $v$ to $u$.
	
	Since $T$ is a tree, $C$ is the only cycle in $T + b$.
	Hence for each $e$ on $P$, we have that $T - e + b$ is acyclic.
	Conversely, if $e \notin P$, then $P \subseteq T - e + b$ and therefore the subgraph still contains $C$.
\end{proof}

We also note that if $X \subseteq G$ is acyclic and $S$ is the set of all tree edges of $G$ outside $E(X)$, then $X$ can contain only blocked edges from the set $\Inv(S)$.

\begin{observation}\label{obs:blocked-edges-in-acyclic-subgraph}
	If $X \subseteq G$ is acyclic and $S = E(T) \setminus E(X)$, then $X \subseteq T + A - S + \Inv(S)$.
\end{observation}
\begin{proof}
	Assume the contrary. 
	Then $X$ contains some blocked edge $b \notin \Inv(S)$.
	Suppose that $b$ goes from $u$ to $v$.
	By \Cref{obs:basic-inverse-edge-property}, we know that $T$ contains a directed path $P$ from $v$ to $u$.
	Together with $b$, this path yields a cycle $C$ in $T + b$.
	Since $b \notin \Inv(s)$ for every $s \in S$, the same observation guarantees that $S \cap P = \emptyset$ and hence $P \subseteq T - S \subseteq X$.
	Therefore, $X$ still contains $C$ which contradicts its acyclicity.
\end{proof}

Finally, we observe that blocked and inverse edges have to follow specific directions derived from a cut of $T$ (see \Cref{fig:inv-edges} (b)).

\begin{observation}\label{claim:inverse-edges-structure}\label{claim:banned-edges-structure}
 	Let $e\in E(T)$ and let $X, Y$ be the vertex sets of two weakly-connected components of $T-e$ labeled so that $e$ goes from $X$ to $Y$ in $G$.
	Then 
	\begin{enumerate}
		\item For each $b\in B$, $b$ does not go from $X$ to $Y$ in $G$;
		\item For each $b\in \Inv(e)$, $b$ goes from $Y$ to $X$ in $G$.
	\end{enumerate}
\end{observation}
\begin{proof}
	We prove the first point of the observation.
	Assume the contrary, and let $b\in B$ be a blocked edge that goes from $X$ to $Y$.
	Then $T+b$ has a cycle that goes through $b$.
	This cycle should contain an edge that goes in the opposite direction, from $Y$ to $X$.
	This is a contradiction, since $T$ does not have such edges.
	
	We now prove the second point of the observation.
	Since $b\in \Inv(e)\subset B$, we have that $T+b$ contains a cycle $C$, and $e$ belongs $C$, due to $b$ is an inverse edge of $e$.
	That is, $C$ contains an edge from $X$ to $Y$.
	Consequently, $C$ contains an edge from $Y$ to $X$.
	By definition of $X$ and $Y$, there is exactly one edge connecting $X$ and $Y$ in $T$, that is, $e$.
	But $E(C)\subseteq E(T+b)$, and no edge can go from $X$ to $Y$ in $E(T)$.
	Then $b$ goes from $Y$ to $X$.
	The proof is complete.
\end{proof}

The most crucial property of inverse edges is that, once unblocked by deleting some edges in $T$, we can add them to the remaining graph together without forming any directed cycle.

\begin{lemma}\label{lemma:inverse-edges-independence}
	Let $(G,\wt)$ be a weakly-connected digraph with positive edge weights.
	Let $S\subset E(T)$ be its tree edge set.
	Then, the graph $T-S+\Inv(S)$ is acyclic.
\end{lemma}
\begin{proof}
	Assume the contrary, and $T-S+\Inv(S)$ contains a cycle $C$.
	While $T-S$ is acyclic, $C$ contains an edge $b\in \Inv(S)$.
	Then, there is $e\in E(T)$ such that $b\in \Inv(e)$.
	
	Consider $X,Y$ as in \Cref{claim:inverse-edges-structure} for $T$ and $e$.
	By the second point of \Cref{claim:inverse-edges-structure}, $b$ goes from $Y$ to $X$.
	Then $C$ contains an edge that goes from $X$ to $Y$.
	We know that $T$ has only one such edge, that is, $e$, and $C$ does not contain $e$ since $e\notin S$.
	
	Then, $\Inv(S)$ contains an edge that goes from $X$ to $Y$.
	We know that $\Inv(S)\subset B$, so this edge is blocked.
	This contradicts the first point of \Cref{claim:inverse-edges-structure} and finishes the proof of the lemma.
\end{proof}

\subsection{Tree edge profits}


The strategy that comes naturally from \Cref{lemma:inverse-edges-independence} can informally be described as trading tree edges for the union of all their inverse edges.
The profit of this trade is a difference between total weights of tree edges and their inverse edge union.
We formalize this intuition.

\begin{definition}[Tree edge profits]
	Let $(G,\wt)$ be a weakly-connected digraph with positive edge weights.
	For a tree edge $e \in E(T)$, by
	 \[\Profit(G, \wt, e)=\wt(\Inv(e)) - \wt(e)\]
	 we denote the \emph{profit} of $e$ in $(G,\wt)$. For a tree edge set $S\subset E(T)$, by
	 \[\Profit(G,\wt,S)=\wt(\Inv(S))-\wt(S)\] we denote the \emph{profit} of $S$ in $(G,\wt)$. 
	 As usual, we use shortened notation $\Profit(e)$ and $\Profit(S)$.
\end{definition}

We observe some straightforward properties of edge profits below.

\begin{observation}\label{obs:profits}
	For a weakly-connected digraph with positive edge weights $(G,\wt)$, and tree edge set $S\subset E(T)$, and an integer $k\ge 0$, we have that
	\begin{itemize}
		\item If $\Profit(S)\ge k$, then $(G,\wt,k)$ is a yes-instance of \probMASMSTshort.
		\item $\Profit(S) \le \sum_{s \in S} \Profit(s)$.
		\item If for each  distinct $e_1, e_2 \in S$ we have  $\Inv(e_1) \cap \Inv(e_2) = \emptyset$, then $\Profit(S) = \sum_{e \in S} \Profit(e)$.
	\end{itemize}
\end{observation}

Throughout our algorithms for \probMASMSTshort, we  assume that the profit of each single tree edge is at most $k$.
Otherwise, we have a yes-instance that is recognizable in polynomial time. 

\subsection{Remaining profits}

While the basic concept of profits clearly suggests to look for a profitable enough tree edge set, it does not really allow us to construct a such a tree set iteratively.

Assume, for example, that $A=\emptyset$ and that we guess correctly that an edge $e_1\in E(T)$ is not a part of the solution.
We know that if we remove $e_1$ from $T$ and add all its inverse edges, we gain $\Profit(e_1)$ profit from this.
Afterwards, we move on to guessing the second edge $e_2$ that should be removed from $T$.
But the profits are not ``up-to-date'' anymore, since it might be the case that $\Inv(e_1)\cap \Inv(e_2)\neq\emptyset$, and we gain less profit than $\Profit(e_2)$ from removing $e_2$.

We introduce the concept of \emph{remaining profit}, which is aimed to support the additive edge removal and represent the actual profit when some edges are already removed from $T$.

\begin{definition}[Remaining profits]
	Let $(G,\wt)$ be a weakly-connected digraph with positive edge weights.
	For a tree edge $e \in E(T)$, and for a tree edge set $S\subset E(T)$, we denote by
	\[\RemainingProfit(G, \wt, S, e)=\wt(\Inv(e)\setminus \Inv(S)) - \wt(e)\]
	the \emph{remaining profit} of $e$ in $(G,\wt)$ with respect to $S$.
	As usual, we shorten this to $\RemainingProfit(S,e)$.
\end{definition}

The following properties of remaining profits are straightforward from the definition.

\begin{observation}\label{obs:rem-profit-props}
	Let $(G,\wt)$ be a weakly-connected digraph with positive edge weights, let $e\in E(T)$, let $S\subset E(T)$ and let $e_1, e_2, \ldots, e_q\in E(T)$ such that $e_i\neq e_j$ for $i\neq j$.
	Then
	\begin{itemize}
		\item $\RemainingProfit(S, e) \le \Profit(e)$.
		
		\item $\Profit( \{e_1, \ldots, e_q\}) = \sum_{i = 1}^q \RemainingProfit(\bigcup_{j=1}^{i-1} \{e_j\}, e_i)$.
	\end{itemize}
	
\end{observation}

In the following observation, we highlight that picking $k$ edges from $E(T)$ consecutively, each time with positive remaining profit, is a winning strategy for \probMASMSTshort.

\begin{observation}\label{obs:win-by-positive-rp}
	Let $\mathcal{I}=(G,\wt,k)$ be an instance of \probMASMSTshort.
	If there exist $e_1, e_2, \ldots, e_k\in E(T)$ with $\RemainingProfit(\bigcup_{j=1}^{i-1} \{e_j\}, e_i)\ge 1$ for each $i\in [k]$, then $\mathcal{I}$ is a yes-instance.
\end{observation}

\begin{proof}
	Straightforward from the second point of \Cref{obs:rem-profit-props} and the first point of \Cref{obs:profits}.
\end{proof}

For $A=\emptyset$, solving an instance is equivalent to finding a set $S\subset E(T)$ with $\Profit(S)\ge k$.

\section{\classFPT-algorithm for integral edge weights}\label{sec:overview}

In this section, we provide four major blocks of the proof of \Cref{thm:main-result}, each block corresponds to a lemma. 

\subsection{Limiting allowed edges}

In \Cref{sec:classification} we gave a lot of insight on blocked edges.
Allowed edges, if given in instance without inverse edges, are somewhat simpler to deal with, even for rational weights.
We are able to show an algorithm that constructs a solution if their total weight is at least $3k$.

\begin{lemma}\label{lemma:allowed-edges}
	Let $\mathcal{I}=(G, \wt, k)$ be an instance of $\probMASMSTshort(\mathbb{Q}_{\ge 1})$. If $\wt(A) \ge 3k$, then $\mathcal{I}$ is a yes-instance.
	The solution to $\mathcal{I}$ can be constructed in polynomial time.
\end{lemma}
\begin{proof}
	We first apply \Cref{prop:tree-order-dimension} to the maximum spanning tree $T$ of $(G, \wt)$.
	Let $\mathcal{P} = (V(T), \preceq)$ be the poset defined as in the proposition statement, that is, $u \preceq v$ if and only if there is a directed path in $T$ from $u$ to $v$.  
	Then, in polynomial time, we can construct linear extensions $\preceq_1$, $\preceq_2$ and $\preceq_3$ of $\mathcal{P}$ such that for each $x, y \in V(T)$, $x \preceq y$ if and only if $x \preceq_i y$ for each $i \in \{1, 2, 3\}$.
	
	We now provide a partition $A_1 \sqcup A_2 \sqcup A_3$ of $A$ such that $T + A_i$ is acyclic for each $i \in \{1, 2, 3\}$.
	For each edge $uv \in A$, we determine its part as follows. 
	By the definition of allowed edges, $T + uv$ does not contain directed cycles. 
	In particular, this means that there is no directed path from $v$ to $u$ in $T$.
	Hence, $v \preceq u$ does not hold. 
	Then, by the definition of $(\preceq_1, \preceq_2, \preceq_3)$, we know that $u$ must precede $v$ in one of these total orders. 
	We choose an arbitrary $\preceq_i$  with this property and add $uv$ to the set $A_i$. 
	Note that the construction of $A_1$, $A_2$ and $A_3$ also takes polynomial time. 
	
	We claim that for each $i \in \{1, 2, 3\}$, all edges from $E(T) \cup A_i$ go from smaller to larger elements with respect to $\preceq_i$. For the edges from $A_i$, this property is ensured by the definition of $A_i$. 
	Now consider an edge $uv \in E(T)$. 
	This edge gives a trivial path from $u$ to $v$ in $T$, and thus we have $u \preceq v$, implying $u \preceq_i v$. 
	Because $uv$ is a tree edge, it cannot be a loop, and hence $u$ strictly precedes $v$ in $\preceq_i$. 
	
	It follows that for each $i \in \{1, 2, 3\}$, the graph $T + A_i$ is acyclic. 
	Let $A_j$ be a set of maximum weight among $A_1$, $A_2$ and $A_3$. 
	If $\wt(A) \ge 3k$, then $\wt(A_j) \ge k$, and hence $\wt(T + A_j) \ge \wt(T) + k$. 
	Therefore, $T + A_j$ is a solution to $\mathcal{I}$, and it can be constructed in polynomial time.  
\end{proof}

It is an interesting open problem to determine whether the bound in \Cref{lemma:allowed-edges} is tight or can be improved.

\subsection{Decomposing MST into paths}

Our next block is a polynomial-time algorithm that either constructs a solution, or obtains an equivalent instance of $\probMASMSTshort(\mathbb{Z}_{\ge 1})$ via edge contractions.
In the resulting instance, all but at most $\Oh(k+|A|)$ vertices have exactly one ingoing and exactly one outgoing edge.

\begin{lemma}\label{lemma:directed-paths-cover}
	There an algorithm that, given an instance $\mathcal{I}=(G,\wt,k)$ of $\probMASMSTshort(\mathbb{Z}_{\ge 1})$, works in polynomial time and either
	\begin{itemize}
		\item outputs a solution to $\mathcal{I}$, or
		\item outputs an equivalent instance $\mathcal{I}'=(G',\wt',k)$ of $\probMASMSTshort(\mathbb{Z}_{\ge 1})$ such that in $\MSF(G',\wt')$ all but at most $\Oh(k+|A'|)$ vertices have both in- and out-degree one, where $A'=A(G',\wt')$.
	\end{itemize}
	Additionally, a solution to $\mathcal{I}'$  can be transformed into a solution to $\mathcal{I}$ in polynomial time.
\end{lemma}
\begin{proof}
	We call a tree edge $e \in E(T)$ \textit{profitable} if $\Profit(e) > 0$. 
	Otherwise we say that it's \textit{unprofitable}. 
	A leaf vertex in $T$ inherits this classification from its unique incident \textit{tree} edge.
	We provide two reduction rules which contract unprofitable edges. 
	These rules are designed in such a way that if they are not applicable, but $T$ has too many vertices whose in-degree or out-degree differs from one, then $\mathcal{I}$ is a yes-instance.
	We start with a couple of claims that are used in both reduction rules. 
	\begin{claim}\label{claim:contracted-instance-properties}
		Consider the graph $G/e$, where $e \in E(T)$ is a tree edge. Then:
		\begin{enumerate}
			\item Weight of $\MSF$ in $G/e$ is $\wt(T) - \wt(e)$.
			\item If $X' \subseteq G/e$ is acyclic and has weight at least $\wt(T) - \wt(e) + k$, then $X + e$ is also acyclic and $\wt(X + e) \ge \wt(T) + k$, where $X$ is a subgraph of $G$ whose edge set corresponds to $E(X')$.
		\end{enumerate}
	\end{claim}
	\begin{claimproof}
		Note that edge contraction preserves cycles (both directed and undirected). Then, if there are no cycles after contracting the edge, there were none before the contraction. In particular, if $T'$ is a spanning tree in $G/e$, then the same edge set together with $e$ forms a spanning tree in $G$. The same argument applies if $X'$ is an acyclic subgraph in $G/e$. In both cases, the weight of the resulting subgraph in $G$ is increased by $\wt(e)$. This proves the second point of the claim.
		
		It remains to show that $\MSF$ in $G/e$ has weight $\wt(T) - \wt(e)$. Indeed, $T/e$ is a spanning tree of exactly this weight, and any spanning tree in $G/e$ with larger weight would, together with $e$, yield a spanning tree in $G$ heavier than $T$, contradicting the maximality of $T$.
	\end{claimproof}
	
	This claim guarantees that a solution for the reduced instance $(G/e, \wt', k)$ always yields a solution for $\mathcal{I}$, where $\wt'$ is the weight function induced by $\wt$ on $G/e$.
	Consequently, to ensure the safeness of our reduction rules, we only need to prove the converse implication. 
	
	We also make use of the following fact about sets of unprofitable edges:
	\begin{claim}\label{claim:trade-unprofitable-edge-set}
		Let $X$ be a subgraph of $G$ and let $E' \subseteq (E(T) \setminus E(X))$ be a subset of tree edges, each of which is unprofitable. 
		Then $\wt(X + E' - \Inv(E')) \ge \wt(X)$.
	\end{claim}
	\begin{claimproof}
		Since $E' \cap E(X) = \emptyset$, it follows that 
		\begin{equation*}
			\wt(X + E' - \Inv(E')) \ge \wt(X) + \wt(E') - \wt(\Inv(E')) = \wt(X) - \Profit(E')
		\end{equation*}
		By \Cref{obs:profits}, we have 
		$\wt(X) - \Profit(E') \ge \wt(X) - \sum_{e' \in E'} \Profit(e')$. 
		Finally, given that all edges in $E'$ are unprofitable, we conclude that 
		$\wt(X + E' - \Inv(E')) \ge \wt(X)$. 
	\end{claimproof}
	
	We proceed to the first reduction rule, which helps us limit the number of unprofitable leaves in $T$.
	\begin{rrule}[Unprofitable leaf contraction]\label{rrule:unprofitable-leaf-contraction}
		Consider a leaf vertex $\ell$ in $T$ that is not incident to any \textit{allowed} edge in $G$, and denote by $e$ the unique tree edge incident to $\ell$. If $\Profit(e) \le 0$, contract $e$.
	\end{rrule}
	\begin{claim}\label{claim:leaf-contraction-is-safe}
		\Cref{rrule:unprofitable-leaf-contraction} is safe.
	\end{claim}
	\begin{claimproof}
		By \Cref{claim:contracted-instance-properties}, it suffices to show that if $\mathcal{I}$ is a yes-instance, then $\mathcal{I}' := (G/e, \wt', k)$ is also a yes-instance. 
		Suppose we have an acyclic subgraph $X \subseteq G$ such that $\wt(X) \ge \wt(T) + k$. 
		If $e \notin E(X)$, consider a subgraph $\tilde{X} := X + e - \Inv(e)$. Since $\Profit(e) \le 0$, by \Cref{claim:trade-unprofitable-edge-set}, we have $\wt(\tilde{X}) \ge \wt(X) \ge \wt(T) + k$.
		Also, we claim that $\tilde{X}$ is acyclic. 
		Assume the contrary. 
		Since $X$ is acyclic, each cycle in $\tilde{X}$ must contain the edge $e$ and the vertex $\ell$. 
		By the definition of $\ell$, it is not incident to \textit{allowed} edges and to \textit{tree} edges except $e$. 
		Consider a \textit{blocked} edge $e'$ that is incident to $\ell$.
		By the definition of a \textit{blocked} edge, $T + e'$ contains a directed cycle. 
		Since $e$ is the only edge that is incident to $\ell$ in $T$, this cycle must contain $e$.
		Then $e' \in \Inv(e)$ and hence $e' \notin E(\tilde{X})$.
		This means that $e$ is the only edge incident to $\ell$ in $\tilde{X}$ and hence $\ell$ cannot lie on any cycle. 
		Then $\tilde{X}$ is also a valid solution.
		Therefore, if $(G, \wt, k)$ is a yes-instance, we always have a solution that contains $e$. Furthermore, we denote such a solution by $X$.
		
		We aim to show that $X/e$ is a valid solution for $\mathcal{I}'$.
		This subgraph has weight $\wt(X) - \wt(e) \ge (\wt(T) - \wt(e))	+ k$. 
		Thus, it remains to prove that $X/e$ is acyclic. 
		Without loss of generality, assume that $\ell$ is a source leaf. 
		Let $v$ be the ending vertex of $e$, and let $w$ be the vertex in $X/e$ resulting from the merge of $\ell$ and $v$. 
		Recall that $X$ is acyclic.
		This means that if $X/e$ contains a cycle $C$, then $w$ must lie on it.
		Consider an edge set $C'$ that corresponds to $C$ in $X$. 
		Note that it forms either a cycle, a path from $\ell$ to $v$, or a path from $v$ to $\ell$. 
		If it is a path from $v$ to $\ell$, then $C' \cup \{e\}$ is a directed cycle in $X$.
		Thus, taking into account the acyclicity of $X$, the only possible case is that $C'$ is the path from $\ell$ to $v$.
		By construction, $e \notin C'$. 
		Then there must be another edge outgoing from $\ell$ in $X$. 
		But note that this is impossible: by the definition of $\ell$, it is not incident to \textit{allowed} edges and to other \textit{tree} edges except $e$;
		also, applying \Cref{claim:inverse-edges-structure} to the edge $e$, we can conclude that there are no \textit{blocked} edges, going from $\ell$.
		Therefore, $X/e$ is acyclic, which completes the proof.
	\end{claimproof}
	
	In fact, an exhaustive application of \Cref{rrule:unprofitable-leaf-contraction} leads us to a structure of $\MSF$ that is very similar to the required one.
	Let $S \subseteq V$ be a set of vertices that are either incident to $\textit{allowed}$ edges or have degree different from two in $T$.
	We argue that if the size of $S$ is not bounded by a linear function of $k + |A|$, then we can immediately conclude that $\mathcal{I}$ is a yes-instance.
	
	\begin{claim}\label{claim:S-is-bounded}
		If \Cref{rrule:unprofitable-leaf-contraction} is not applicable and $|S| \ge 4(k + |A|)$, then $\mathcal{I}$ is a yes-instance of \probMASMSTshort.
	\end{claim}
	\begin{claimproof}
		First, we divide the vertices from $S$ into three groups, depending on their degree in $T$: $S_1$, $S_2$, and $S_{\ge 3}$ are the sets of vertices of degree one, two, and at least three, respectively.
		Note that $S = S_1 \sqcup S_2 \sqcup S_{\ge 3}$.
		Also, from the definition of $S$, we know that $S_1$ and $S_{\ge 3}$ contain all vertices of $T$ with the corresponding degrees, while $S_2$ consists only of those vertices of degree two that are incident to the \textit{allowed} edges. 
		It is known that in every tree $|S_1| > |S_{\ge 3}|$. Then
		\begin{equation*}
			|S_1| + |S_2| \ge \frac{1}{2}(|S_1| + |S_{\ge 3}|) + |S_2| \ge \frac{1}{2}|S| = 2(k + |A|)
		\end{equation*}
		Among the vertices from the sets $S_1$ and $S_2$, at most $2|A|$ vertices are incident to the \textit{allowed} edges, and the rest are leaves in $T$. 
		Since \Cref{rrule:unprofitable-leaf-contraction} is not applicable, we can conclude that $T$ contains at least $|S_1| + |S_2| - 2|A| = 2k$ profitable leaves. 
		We divide them into sources and sinks, depending on the direction of the only incident edge. 
		Without loss of generality, we assume that at least half of them are sources.
		Let $D$ be the set of tree edges incident to such source leaves.
		
		We aim to show that the subgraph $X := T - D + \Inv(D)$ is a valid solution for the instance $\mathcal{I}$.
		By \Cref{lemma:inverse-edges-independence}, we know that it is acyclic.
		Thus, it remains to show that $\wt(X) \ge \wt(T) + k$.
		Consider two distinct edges $d_1$ and $d_2$ from $D$. 
		Let $\ell_1$ and $\ell_2$ be the corresponding source leaves.
		Applying \Cref{claim:inverse-edges-structure} to the edge $d_1$, we obtain that all its inverse edges end at vertex $\ell_1$. 
		By the same argument, all edges from $\Inv(d_2)$ end at $\ell_2$, and hence $\Inv(d_1) \cap \Inv(d_2) = \emptyset$. 
		Then the last point of \Cref{obs:profits} guarantees that $\Profit(D) = \sum_{d \in D} \Profit(d)$.
		Recall that all edges in $D$ are profitable.
		It follows that $\wt(X) = \wt(T) + \sum_{d \in D} \Profit(d) \ge \wt(T) + |D| \ge \wt(T) + k$.
	\end{claimproof}
	
	Given that the size of $S$ is bounded, the only possible obstacle to proving the lemma is that $T$ may contain many sinks and sources of degree two. 
	Let $S'$ be a set of such sinks and sources that are not incident to the allowed edges. 
	Note that $S \cap S' = \emptyset$.
	Our goal is to limit the size of $S'$ by a linear function of $k + |A|$ as well.
	Intuitively, our aim is to design a reduction rule to guarantee that, after its exhaustive application, any undirected path in $T$ (i.e., a path in the underlying graph of $T$) that avoids $S$ also contains only a bounded number of vertices from $S'$.
	Before we present such a reduction rule, let us understand the structure of inverse edges relative to $S'$. 
	
	\begin{claim}\label{claim:inverse-edges-for-paths-between-S'}
		Let $u, v \in S'$ be a pair of vertices such that there is a directed path from $u$ to $v$ in $T$, with all internal vertices outside both $S$ and $S'$.
		Let $E_{uv}$ be the set of edges on this path. 
		Then $e \in E_{uv}$ cannot share an inverse edge with $e' \in (E(T) \setminus E_{uv})$.
	\end{claim}
	\begin{claimproof}
		Since $u, v \in S'$, we know that each of them is either a sink or a source in $T$.
		Note that, given the existence of a directed path from $u$ to $v$, the vertex $u$ must be a source and $v$ must be a sink.
		Moreover, as all internal vertices of the path lie outside both $S$ and $S'$, their in-degree and out-degree in $T$ are equal to one.
		
		Consider an edge $e \in E_{uv}$ that leads from $u'$ to $v'$ and suppose that it has an inverse edge $\tilde{e} \in \Inv(e)$ that goes from $a$ to $b$.
		By \Cref{obs:basic-inverse-edge-property}, $T$ contains a directed path from $v'$ to $u'$, and this path contains $\tilde{e}$.
		Then there also must be paths in $T$ from $v'$ to $a$ and from $b$ to $u'$.
		If $v'$ is an internal vertex of the path from $u$ to $v$, then its only outgoing edge in $T$ leads to the next vertex on the path. 
		Otherwise, $v' = v$ and its out-degree is zero. 
		This means that all vertices that are reachable from $v'$ in $T$ lie on the path from $v'$ to $v$, and hence $a$ also must be on this path.
		Similarly, we can show that $b$ lies on the path from $u$ to $u'$.
		Thus, a path from $b$ to $a$ in $T$ consists only of edges from $E_{uv}$.
		Applying the \Cref{obs:basic-inverse-edge-property}, we obtain that if $\tilde{e} \in \Inv(e')$, then $e' \in E_{uv}$.
	\end{claimproof}
	
	We move on to the second reduction rule. 
	\begin{rrule}\label{rrule:contraction-between-sinks-and-sources}
		Consider four distinct vertices $v_1, v_2, v_3, v_4 \in S'$ for which $T$ contains edge-disjoint directed paths from $v_1$ to $v_2$, from $v_3$ to $v_2$, and from $v_3$ to $v_4$, all of whose internal vertices lie outside both $S$ and $S'$.
		If all edges on these paths are unprofitable, contract the first edge on the path from $v_1$ to $v_2$. 
	\end{rrule}
	
	\begin{claim}
		\Cref{rrule:contraction-between-sinks-and-sources} is safe.
	\end{claim}
	\begin{claimproof}
		Recall that by \Cref{claim:contracted-instance-properties}, we only need to show that the existence of a solution for $\mathcal{I}$ implies that there is also a solution for $\mathcal{I}' := (G/e_c, \wt', k)$, where $e_c$ is the first edge on the path from $v_1$ to $v_2$, and $\wt'$ is the weight function naturally inherited from $\wt$. 
		
		Let $E_p$ be the set of edges on the paths from $v_1$ to $v_2$, from $v_3$ to $v_2$, and from $v_3$ to $v_4$ in $T$, and let $V_p$ be the set of all vertices on these paths excluding $v_1$ and $v_4$.
		Consider a vertex $v \in V_p$. 
		Note that $v \notin S$. For internal vertices of the paths, this is ensured by the lemma statement; for $v_2$ and $v_3$, it follows from the fact that $S \cap S' = \emptyset$. 
		Consequently, $v$ is incident to no allowed edges, and its degree in $T$ is two.
		Thus, $E_p$ contains all edges that are incident to $v$ in $T$: indeed, if $v$ is an internal vertex of a path, then it has two incident edges on this path. 
		Otherwise, $v$ is either $v_2$ or $v_3$.
		In both cases, $v$ belongs to two paths, and since these paths are edge-disjoint, it also has two incident edges in $E_p$. 
		
		Suppose that $X$ is a solution for $\mathcal{I}$. Let $E_r := E_p \setminus E(X)$.
		We claim that $\tilde{X} := X + E_r - \Inv(E_r)$ is also a solution for $\mathcal{I}$, and every vertex from $V_p$ has exactly two incident edges in $\tilde{X}$, both from $E_p$.
		
		We start with the latter property. Consider a vertex $v \in V_p$.
		As argued above, $v$ is incident to no allowed edges. 
		Moreover, $v$ has two incident edges in $T$, both of which belong to $E_p$.
		Since $E_p \subseteq E(\tilde{X})$, these edges are also present in $\tilde{X}$.
		Suppose, for contradiction, that some blocked edge $e' \in E(\tilde{X})$ is incident to $v$, and let $u$ be its other endpoint.
		By \Cref{obs:basic-inverse-edge-property}, there must be a directed path $P$ between $u$ and $v$ in $T$, and for every edge $e$ on this path, we have $e' \in \Inv(e)$.
		We distinguish two cases. 
		\begin{enumerate}
			\item Assume that $P$ consists only of edges from $E_p$. If it contains at least one edge $e \in E_r$, then $e' \notin \tilde{X}$, because $\Inv(e) \cap E(\tilde{X}) = \emptyset$. 
			Otherwise, all edges on $P$ were also in $X$. 
			Since $X$ is acyclic, it follows that $e' \notin E(X)$.
			And as no new blocked edges are introduced in  $\tilde{X}$, we can conclude that $e' \notin E(\tilde{X})$.
			\item Assume that there is an edge $e_1$ from the set $(E(T) \setminus E_p)$ on the path $P$.
			Recall that all edges incident to $v$ in $T$ lie in $E_p$. 
			Consequently, $P$ must also contain an edge $e_2 \in E_p$. 
			Note that $e' \in (\Inv(e_1) \cap \Inv(e_2))$.
			By the definition of $E_p$, the edge $e_2$ lies on one of three paths: either from $v_1$ to $v_2$, from $v_3$ to $v_2$, or from $v_3$ to $v_4$. 
			In contrast, since $e_1 \notin E_p$, the edge $e_1$ is not part of any of these paths.
			Note that each of these three paths satisfies the condition of \Cref{claim:inverse-edges-for-paths-between-S'}.
			This implies that $e_1$ and $e_2$ cannot have inverse edges in common, yielding a contradiction. 
		\end{enumerate}
		
		Hence, no blocked edge is incident to $v$, and its degree it $\tilde{X}$ is two.
		
		Next, we prove that $\tilde{X}$ is acyclic.
		If there is a cycle in $\tilde{X}$, it must contain at least one edge that was not present in $X$. 
		Suppose that it contains an edge $e \in E_r \subseteq E_p$, and let $u$ and $v$ be the starting and ending vertices of $e$, respectively. 
		By \Cref{obs:basic-inverse-edge-property}, there must be a directed path from $v$ to $u$ in $\tilde{X}$.
		We aim to obtain a contradiction with this fact.
		Recall that $E_p$ consists of edges from the three paths mentioned in the lemma statement.
		We consider two cases, depending on which path contains $e$.
		\begin{enumerate}
			\item Assume that $e$ lies either on the path from $v_1$ to $v_2$ or on the path from $v_3$ to $v_2$.
			Then $v \in V_p$, so $v$ has two incident edges in $\tilde{X}$, both from $E_p$.
			If $v \neq v_2$, it is an internal vertex of the path and thus has exactly one outgoing edge in $\tilde{X}$.
			Moreover, the vertex at the opposite end of this edge also belongs to $V_p$, and we can apply the same reasoning to it.
			Traversing such edges repeatedly, we eventually reach the vertex $v_2$, while $u$ is never visited along this walk. 
			Note that $v_2$ is a sink in $\tilde{X}$: it belongs to $V_p$ and is therefore incident only to tree edges. Furthermore, $v_2 \in S'$ and is hence either a sink or a source in $T$. 
			Given the directed path from $v_1$ to $v_2$ in $T$, we infer that $v_2$ is a sink.
			Consequently, $u$ is not reachable from $v$ in $\tilde{X}$, which leads us to a contradiction. 
			
			\item The second case is analogous. 
			The only difference is that we consider the path from $v$ to $u$ from the end.
			If the edge $e$ lies on the path from $v_3$ to $v_4$, then $u \in V_p$. 
			If $u \neq v_3$, then it has exactly one incoming edge in $\tilde{X}$; let $u'$ be the starting vertex of this edge. Note that $u' \in V_p$ as well.
			By repeatedly moving to the starting vertex of the unique incoming edge, we eventually reach $v_3$. The vertex $v_3$ is a source in $\tilde{X}$, and $v$ is never visited along this backward walk.
			Consequently, a directed path from $v$ to $u$ cannot exist in $\tilde{X}$, yielding a contradiction.
		\end{enumerate}
		
		It remains to provide an argument about the weight of $\tilde{X}$.
		Since $E_r \subseteq E_p$, all edges of $E_r$ are unprofitable. 
		Thus, by \Cref{claim:trade-unprofitable-edge-set}, we have $\wt(\tilde{X}) \ge \wt(X) \ge \wt(T) + k$.
		Therefore, $\tilde{X}$ is also a solution for $\mathcal{I}$.
		
		Consider the subgraph $\tilde{X}/e_c$.
		It has weight $\wt(\tilde{X}) - \wt(e_c) \ge (\wt(T) - \wt(e_c)) + k$. 
		Consequently, we only need to show that $\tilde{X}/e_c$ is acyclic.
		Let $u$ be the ending vertex of the edge $e_c$, and let $w$ be the vertex in $\tilde{X}/e_c$ resulting from the merge of $v_1$ and $u$.
		Similarly to the proof of \Cref{claim:leaf-contraction-is-safe}, we can derive that if $\tilde{X}/e_c$ contains a cycle, $w$ must lie on it, and this implies the existence of a path from $v_1$ to $u$ in $\tilde{X}$ without the edge $e_c$. We prove that there cannot be such a path in $\tilde{X}$.
		
		Note that $u \in V_p$. Hence, its degree in $\tilde{X}$ is two, and both incident edges lie in $E_p$. 
		If $u \neq v_2$, then it is an internal vertex of the path from $v_1$ to $v_2$, and hence it has only one incoming edge in $\tilde{X}$, which is $e_c$. In this case, any path from $v_1$ to $u$ must contain the edge $e_c$.
		If $u = v_2$, then it has two incoming edges: $e_c$ and the last edge on the path from $v_3$ to $v_2$, let $u'$ be the opposite end of this edge. 
		The vertex $u'$ is also in $V_p$, and if $u' \neq v_3$, then it is an internal vertex of the path from $v_3$ to $v_2$. 
		Repeating this argument, we eventually arrive at the vertex $v_3$, which is a source in $\tilde{X}$.
		Therefore, in this case there is also no path from $v_1$ to $u$ that does not contain $e_c$.
		Therefore, the subgraph $\tilde{X}/e_c$ is acyclic, which completes the proof.
	\end{claimproof}
	
	It remains to show that if neither reduction rule is applicable, then the size of $S'$ is bounded by $\mathcal{O}(k + |A|)$. 
	Intuitively, the argument is as follows: since every vertex in $T$ whose degree differs from two belongs to the set $S$, the tree can be decomposed into undirected paths whose endpoints lie in $S$ and whose internal vertices are all outside $S$. 
	The number of such paths is $|S| - 1$, which is bounded by \Cref{claim:S-is-bounded}.
	After that, for any four consecutive vertices from $S'$ on a single path, we can show that if \Cref{rrule:contraction-between-sinks-and-sources} is not applicable, then there is a profitable edge on this subpath. 
	Moreover, \Cref{claim:inverse-edges-for-paths-between-S'} guarantees that profitable edges from different subpaths (and from different paths) cannot share inverse edges. Hence, by \Cref{obs:profits}, we can sum up their profits. 
	Therefore, if the size of $S'$ is too large, it implies that $\mathcal{I}$ is a yes-instance.
	
	\begin{claim}\label{claim:S'-is-bounded}
		If \Cref{rrule:unprofitable-leaf-contraction} and \Cref{rrule:contraction-between-sinks-and-sources} are not applicable and $|S'| \ge 12|A| + 16k$, then $\mathcal{I}$ is a yes-instance.
	\end{claim}
	\begin{claimproof}
		First, observe that the underlying graph of $T$ can be represented as a set of paths $\mathcal{P} = \{P_1, \ldots, P_t\}$ such that:
		\begin{itemize}
			\item $t = |S| - 1$;
			\item for every path $P_i$, its endpoints belong to $S$, while all its internal vertices lie outside $S$;
			\item for every pair $u, v$ of consecutive vertices on a path, there is an edge in $T$ between $u$ and $v$ in one of the directions;
			\item every edge of $T$ is covered by exactly one path from $\mathcal{P}$.
		\end{itemize}
		To achieve this, we root the underlying graph of $T$ at an arbitrary vertex $r \in S$. Then, for each $s \in S$ with $s \neq r$, we add to $\mathcal{P}$ the path from $s$ to its closest ancestor from $S$. 
		Since $r \in S$, such an ancestor always exists.
		Note that this construction guarantees that the first three required properties are satisfied. 
		For the last property, recall that every leaf of $T$ is in $S$, and hence every edge of the tree is covered at least once. 
		Conversely, suppose that two paths $P_i$ and $P_j$, starting at vertices $s_1$ and $s_2$ respectively, have a common edge. 
		If $s_1$ is an ancestor of $s_2$ or vice versa, this is impossible by our construction.
		Otherwise, let $s$ be the least common ancestor of $s_1$ and $s_2$. 
		Since $P_i$ and $P_j$ share an edge, $s$ must be an internal vertex of both paths. 
		But in this case, $s$ has degree at least three in $T$, and hence it must lie in $S$, which contradicts the structure of $P_i$ and $P_j$, because they must have ended at $s$.
		Notice that if $t \ge 4(|A| + k)$, then by \Cref{claim:S-is-bounded}, $\mathcal{I}$ is a yes-instance. 
		Thus, we may assume that $t < 4(|A| + k)$.
		
		Consider a path $P_i \in \mathcal{P}$, and let $v_1, \ldots, v_{p_i} \in P_i$ be the vertices from $S'$ listed in the order they appear on the path.
		Note that for each pair $v_j$ and $v_{j + 1}$, the subsegment of $P_i$ between them must form a directed path in $T$. Indeed, if it did not, there would be a vertex $v$ between them that has either two incident incoming edges or two incident outgoing edges on the path.
		But we know that $v \notin (S \cup S')$, and hence it must have in-degree one and out-degree one in $T$.
		We denote the directed path between $v_j$ and $v_{j + 1}$ by $P_{i, j}$. 
		
		Recall that, by the definition of $S'$, each $v_j$ is either a sink or a source in $T$. 
		The argument above guarantees that their types must alternate: i.e., if $v_j$ is a sink, then $v_{j + 1}$ must be a source, and vice versa. 
		Since \Cref{rrule:contraction-between-sinks-and-sources} is not applicable, for each $j$ with $0 \le j \le p_i - 3$ there must be a profitable edge on the subsegment of $P_i$ between $v_j$ and $v_{j + 3}$. Otherwise, the four vertices $v_j, v_{j + 1}, v_{j + 2}, v_{j + 3}$ (or the reverse order $v_{j + 3}, v_{j + 2}, v_{j + 1}, v_j$) would satisfy the conditions of this reduction rule. 
		Applying this argument to each block of four vertices of the form $v_{4x + 1}, v_{4x + 2}, v_{4x + 3}, v_{4x + 4}$, we obtain $\lfloor p_i / 4 \rfloor$ profitable edges on the path $P_i$. 
		Moreover, each pair $v_j$, $v_{j + 1}$ satisfies the conditions of \Cref{claim:inverse-edges-for-paths-between-S'}. Hence, an edge $e \in P_{i, j}$ cannot share an inverse edge with any $e' \in (E(T) \setminus P_{i, j})$.
		
		Consequently, we can construct a set of profitable tree edges $R := \{e_1, \ldots, e_q\}$ such that $q = \sum_{i = 1}^t \lfloor p_i / 4 \rfloor$ and $\Inv(e_i) \cap \Inv(e_j) = \emptyset$. 
		Note that if $q \ge k$, then $X := T - R + \Inv(R)$ is a solution for $\mathcal{I}$: indeed, it is acyclic by \Cref{lemma:inverse-edges-independence}, and by \Cref{obs:profits}, we have
		\begin{equation*}
			\wt(X) = \wt(T) + \Profit(R) = \wt(T) + \sum_{i = 1}^q \Profit(e_i) \ge \wt(T) + q \ge \wt(T) + k
		\end{equation*}
		If, on the other hand, $q < k$, then we can conclude that $|S'| = \sum_{i = 1}^t p_i \le \sum_{i = 1}^t (4\lfloor p_i / 4 \rfloor + 3) = 3t + 4q \le 3 \cdot 4(|A| + k) + 4k = 12|A| + 16k$. 
	\end{claimproof}
	
	Let $\mathcal{I}' := (G', \wt', k)$ be the equivalent instance obtained by exhaustive applications of \Cref{rrule:unprofitable-leaf-contraction} and \Cref{rrule:contraction-between-sinks-and-sources}.
	Recall that all vertices outside both $S$ and $S'$ have one incoming tree edge and one outgoing tree edge.
	Together with \Cref{claim:S-is-bounded} and \Cref{claim:S'-is-bounded}, this implies that either $\mathcal{I}'$ is a yes-instance or all but at most $\mathcal{O}(k + |A|)$ vertices in $\MSF(G', \wt')$ have in-degree one and out-degree one.
	Moreover, \Cref{claim:contracted-instance-properties}, provide a way to transform a solution for $\mathcal{I}'$ into a solution for $\mathcal{I}$. 
\end{proof}


\Cref{lemma:directed-paths-cover} does not provide any upper bound on the size of $G'$ in $\mathcal{I}'$.
Nevertheless, the MaxST of $(G',\wt')$ exhibits a simple structure, since it consists of $\mathcal{O}(k+|A|)$ pivotal vertices connected by pairwise-disjoint directed paths.

\subsection{Greedy-like approach to remaining profits}

Our third building block relies fundamentally on remaining profits and directed paths in MaxST.
Somewhat surprisingly, there is a strategy that provides a limited number of choices of removing a single edge from $E(T)$ iteratively, depending on its current remaining profit, given that $T$ is partitioned into a limited number of directed paths (guaranteed by \Cref{lemma:directed-paths-cover}).

This strategy works even if $A$ is not empty, and it highly relies on that $E(T)$ is ordered in a topological manner.
To describe it, we have to introduce the notion of \emph{proper path covers}.

\begin{definition}[Proper path cover]\label{def:path-cover}
	A \emph{proper path cover} $\mathcal{P}$ of $T = \MaxST(G,\wt)$ is a sequence $P_1, P_2, \ldots, P_t$ of directed paths of $T$, such that
	\begin{itemize}
		\item each edge $e\in E(T)$ belongs to exactly one path in $T$, and
		\item each internal vertex of each $P_i$ has in-degree and out-degree one in $T$, and
		\item if  $v\in V(G)$ is incident to an allowed edge $a\in A$, then $v$ is not an internal vertex in any $P_i$, and
		\item if there is a directed path in $T$, that starts with $s\in E(T)$ and ends with $e\in E(T)$, then $i\le j$, where $i,j\in[t]$ are such that $e\in P_i$ and $s\in P_j$.
	\end{itemize}
\end{definition}

Note that the last point in the definition above is a topological-like ordering of paths required for our strategy.
We also have to introduce the notion of edge orderings that agree with proper path covers. 
Order the edges inside the paths, so the last edge goes first, and the first edge goes last.
We call the resulting total edge ordering the \emph{$\mathcal{P}$-respecting ordering}. 
The formal definition is given below. 

\begin{definition}[$\mathcal{P}$-respecting orderings]
	Let $(G,\wt)$ be a weakly-connected digraph with positive edge weights, and let $\mathcal{P}$ be a proper path cover of $G$.
	We say that an ordering $e_1, e_2, \ldots, e_{n-1}$ of $E(T)$ is \emph{$\mathcal{P}$-respecting} if
	\begin{itemize}
		\item if there is a directed path in $T$ that starts with $e_j$ and ends with $e_i$, then $i\le j$ holds, and
		\item for each $1\le i \le j < n$, it holds that $\ell \le \ell'$, where $e_i\in P_\ell$ and $e_j\in P_{\ell'}$.
	\end{itemize}
\end{definition}

\todo{maybe put the picture illustrating proper path covers and $\mathcal{P}$-respecting orderings}

The definition of $\mathcal{P}$-respecting orderings implies that edges of a single path of $\mathcal{P}$ form consecutive segments in the ordering.
Such a segment starts with the last edge of a path and ends with a first edge of a path.
We are now ready to formulate the lemma itself. 

\begin{lemma}\label{lemma:good-solution-set}
	Let $\mathcal{I}=(G,\wt,k)$ be an instance of $\probMASMSTshort(\mathbb{Z}_{\ge 1})$, and let $\mathcal{P}$ be a proper path cover of $T$.
	Let $e_1, e_2, \ldots, e_{n-1}$ be a $\mathcal{P}$-respecting ordering of $E(T)$.
	If $\mathcal{I}$ is a yes-instance, then there exists a solution $G'$ to $\mathcal{I}$ with $S = E(T) \setminus E(G')$ such that
	 \begin{itemize}
		\item If $e_i\in S$, then $\RemainingProfit(S_{i-1},e_i)\ge -\wt(A)$, and
		\item If $e_i, e_j \in S$ for $i<j$, and $e_i,e_j$ belong to the same path in $\mathcal{P}$, then $\RemainingProfit( S_{j-1}, e_j)> 0$, and
		\item If $e_i\notin S$ and $e_j\in S$ for $i<j$, $S_{j - 1} = S_{i - 1}$ and $e_i, e_j$ belong to the same path in $\mathcal{P}$, then $\RemainingProfit(S_{i-1},e_i)<\RemainingProfit(S_{i-1},e_j)$,
	\end{itemize}
	where $S_0=\emptyset$ and $S_i=S_{i-1}\cup (S\cap \{e_i\})$ for each $i\in [n-1]$.
\end{lemma}
\begin{proof}
	Suppose that $\mathcal{I}$ is a yes-instance. 
	We start with an arbitrary solution $X$ and aim to obtain a solution $G'$ that satisfies all the required properties by applying a sequence of local modifications.
	Each modification is of one of three types, corresponding to each property in the lemma statement.
	
	Let $S$ be the set of tree edges that do not belong to the current solution. 
	Initially, $S = E(T) \setminus E(X)$. 
	Modifications of the first type help us eliminate edges in $S$ that are \textit{very unprofitable} (in terms of remaining profit). 
	We make use of the following claim:	
	\begin{claim}\label{claim:drop-all-allowed-edges-instead-of-very-unprofitable}
		If there is an edge $e_i \in S$ such that $\RemainingProfit(S_{i - 1}, e_i) \le -\wt(A)$, then $X' := X + e_i - A - (\Inv(e_i) \setminus \Inv(S_{i - 1}))$ is also a solution for $\mathcal{I}$.
	\end{claim}
	\begin{claimproof}
		Since $X$ is acyclic, we have $X \subseteq T - S + \Inv(S) + A$.
		Let $S' := S \setminus \{e_i\}$. 
		From the definition of $X'$, it follows that $X' \subseteq T - S' + \Inv(S)$. 
		Moreover, notice that if $X'$ contains an inverse edge from $\Inv(e_i)$, then this edge also belongs to $\Inv(S_{i - 1}) \subseteq \Inv(S')$. 
		Consequently, we obtain the stronger inclusion: $X' \subseteq T - S' + \Inv(S')$. 
		Then $X'$ is acyclic by \Cref{lemma:inverse-edges-independence}.
		Recall that $\RemainingProfit(S_{i - 1}, e_i)$ is defined as $\wt(\Inv(e_i) \setminus \Inv(S_{i - 1})) - \wt(e_i)$. Hence,
		\begin{equation*}
			\begin{split}
				\wt(X') &\ge \wt(X) - \wt(A) + (\wt(e_i) - \wt(\Inv(e_i) \setminus \Inv(S_{i - 1}))) \\ &=\wt(X) - \wt(A) - \RemainingProfit(S_{i - 1}, e_i) \ge \wt(X)
			\end{split}
		\end{equation*}
		Thus, $X'$ is indeed a valid solution for $\mathcal{I}$. 
	\end{claimproof}
	
	Next, we analyze the structure of $S$ inside a single path from $\mathcal{P}$. 
	We show that at most one unprofitable edge per path is needed.
	\begin{claim}\label{claim:drop-second-unprofitable-edge-on-path}
		If two edges $e_i$ and $e_j$ from $S$ with $i < j$ lie on the same path $P$ in $\mathcal{P}$ and $\RemainingProfit(S_{j - 1}, e_j) \le 0$, then $X' := X + e_j - (\Inv(e_j) \setminus \Inv(S_{j - 1}))$ is also a solution for $\mathcal{I}$. 
	\end{claim}
	\begin{claimproof}
		Consider an edge $e_j \in S$ for which there exists an edge $e_i \in S$ with $i < j$ from the same path $P$.
		Choose $e_i$ with the maximum possible index. 
		Note that in this case $S_{j - 1} = S_i$: since $e_1, \ldots, e_{n - 1}$ is a $\mathcal{P}$-respecting ordering, the edges of the path $P$ form a contiguous subsegment. 
		Hence, for every index $i'$ with  $i < i' < j$, the edge $e_{i'}$ also lies on the path $P$ and it cannot belong to $S$; otherwise, this would contradict the maximality of $i$.
		
		Suppose that $e_j$ goes from $x$ to $y$, and let $L$ be the set of vertices reachable from $y$ via \textit{tree} edges in $X'$. 
		Denote the starting vertex of $e_i$ by $z$.
		Observe that $L$ is a subpath of $P$ between $y$ and $z$ (inclusive):
		since $e_1, \ldots, e_{n - 1}$ is a $\mathcal{P}$-respecting ordering, all edges on the tree path from $y$ to $z$ have indices between $i + 1$ and $j - 1$, and are therefore present in $X'$, because $S_{j - 1} = S_i$. 
		Consequently, every vertex on this path belongs $L$.
		On the other hand, all these vertices are internal vertices of $P$ and hence each of them has exactly one outgoing edge in $T$. 
		Combining with the fact that $e_i \notin E(X')$, this implies that $L$ contains no other vertices. 
		
		We aim to show that $N_{X'}[L] = L$:
		\begin{itemize}
			\item For every \textit{tree} edge $\ell \to v$ with $\ell \in L$, the ending vertex $v$ must also belong to $L$ by the definition of $L$.
			
			\item Since $L$ consists only of internal vertices of $P$, no vertex in $L$ is incident to any \textit{allowed} edge; this follows directly from the definition of a proper path cover.
			
			\item Consider a \textit{blocked} edge $e = \ell \to v$ in $X'$ such that $\ell \in L$. 
			By \Cref{obs:basic-inverse-edge-property}, $T$ contains a directed path $P'$ from $v$ to $\ell$.
			First, we claim that $e \notin \Inv(S_{j - 1})$.
			Assume the contrary. 
			Then there exists an edge $e_t \in S_{j - 1}$ such that $e \in \Inv(e_t)$. 
			Since $S_{j - 1} = S_i$, we have $t \le i$. 
			By \Cref{obs:basic-inverse-edge-property}, $e_t$ lies on $P'$.
			Thus there is also a directed path in $T$ that starts at $e_t$ and ends in $\ell$.
			Observe that there is only one tree edge $e_s$ that leads to $\ell$ and it lies between $e_j$ and $e_i$ (including $e_j$, but not $e_i$) on the path $P$. 
			Hence, we have that $i < s \le j$.
			Consequently, there must be a path in $T$ that starts at $e_t$ and ends at $e_s$, which contradicts the definition of $\mathcal{P}$-respecting ordering, because $t \le i < s$. 
			
			Next, suppose that $v \notin L$. 
			Recall that every vertex in $L$ has exactly one incoming tree edge (as they are internal vertices of $P$).
			Moreover, $e_j$ is the only tree edge that goes from $V \setminus L$ to $L$. 
			Since $v \in (V \setminus L)$ and $\ell \in L$, $P'$ must contain $e_j$. 
			By \Cref{obs:basic-inverse-edge-property}, it follows that $e \in \Inv(e_j)$. 
			But this is impossible, because together with $e \notin \Inv(S_{j - 1})$, we would have $e \in (\Inv(e_j) \setminus S_{j - 1})$.
			But $E(X') \cap (\Inv(e_j) \setminus \Inv(S_{j - 1})) = \emptyset$, and hence the edge $e$ cannot be in $X'$. 
		\end{itemize}
		
		Using this fact, it is easy to see that $X'$ is acyclic. 
		If $X'$ contained a cycle, this cycle would have to include the edge $e_j = x \to y$, since it is the only edge in $X'$ that is not present in $X$.
		It would then follow that $x$ is reachable from $y$ in $X'$.
		But note that $x \notin L$ and $y \in L$.
		This leads us to the contradiction, because $N_{X'}[L] = L$.
		
		As for weight of $X'$, we have $\wt(X') \ge \wt(X) + \wt(e_j) - \wt(\Inv(e_j) \setminus \Inv(S_{j - 1})) = \wt(X) - \RemainingProfit(S_{j - 1}, e_j) \ge \wt(X)$. Therefore, $X'$ is indeed a valid solution for $\mathcal{I}$. 
	\end{claimproof}
	
	Finally, modifications of the third type allow us to replace an edge from $S$ by another edge outside $S$ that lies on the same path in $\mathcal{P}$, appears earlier in the ordering $e_1, \ldots, e_{n - 1}$, and has a remaining profit that is not worse.
	
	\begin{claim}\label{claim:flip-removed-edge}
		If two edges $e_i \notin S$ and $e_j \in S$ such that $i < j$ and $S_{j - 1} = S_{i - 1}$ belong to the same path in $\mathcal{P}$ and $\RemainingProfit(S_{i - 1}, e_i) \ge \RemainingProfit(S_{i - 1}, e_j)$, then $X' := X - e_i + e_j + (\Inv(e_i) \setminus \Inv(e_j) \setminus \Inv(S_{i - 1})) - (\Inv(e_j) \setminus \Inv(e_i) \setminus \Inv(S_{i - 1}))$ is also a solution for $\mathcal{I}$. 
	\end{claim}
	\begin{claimproof}
		Let $y$ be the ending vertex of $e_j$ and $z$ be the starting vertex of $e_i$. 
		Define $L$ as the set of all vertices reachable from $y$ in $X'$ via tree edges. 
		Since $S_{j - 1} = S_{i - 1}$, the edges $e_{i + 1}, \ldots, e_{j - 1}$ are present in $X'$, while $e_i \notin E(X')$. 
		Consequently, $L$ is exactly the set of vertices on the tree path from $y$ to $z$.
		We start with a part very similar to the proof of \Cref{claim:drop-second-unprofitable-edge-on-path}. Namely, we show that $N_{X'}[L] = L$. 
		
		For the \textit{tree} and \textit{allowed} edges, the situation is identical to the previous claim: tree edges cannot lead outside $L$ by the definition of this set, and $L$ is incident to no allowed edges.
		It remains to consider blocked edges.
		Suppose, for contradiction, that there exists a blocked edge $e = \ell \to v$ in $X'$ with $\ell \in L$ and $v \notin L$. 
		Recall that in this case $e \in \Inv(e_j)$ and $e \notin \Inv(S_{j - 1}) = \Inv(S_i)$. 
		Hence, $e \in (\Inv(e_j) \setminus \Inv(e_i) \setminus \Inv(S_{i - 1}))$, which contradicts the fact that $e$ belongs to $X'$. 
		
		Our next goal is to show that every edge in $E(X') \setminus E(X)$ goes between $L$ and $V \setminus L$. 
		We know that $e_j$ leads from $V \setminus L$ to $L$. 
		Therefore, it suffices to verify this for edges $e = u \to v$ from $(\Inv(e_i) \setminus \Inv(e_j) \setminus \Inv(S_{i - 1}))$.
		Recall that by \Cref{obs:basic-inverse-edge-property}, $T$ contains a directed path $P'$ from $v$ to $u$.
		
		We first show that $u \notin L$. 
		Since $e \in \Inv(e_i)$, by \Cref{obs:basic-inverse-edge-property} $e_i$ belongs to $P'$.
		Hence there also exists a path in $T$ starting at $e_i$ and ending at $u$. 
		Then $z$ is not reachable from $u$ in $T$, as otherwise $T$ would contain a cycle. 
		On the other hand, for every $\ell \in L$ there is a tree path from $\ell$ to $z$.
		Consequently, $u \notin L$. 
		It remains to prove that $v \in L$.
		Assume the contrary. 
		Recall that each vertex in $L$ has exactly one incoming edge in $T$, and $e_j$ is the only edge that goes from $V \setminus L$ to $L$.
		Since $e_i$ belongs to $P'$, its endpoint $z$ also lies on $P'$.
		Given that $v \notin L$ and $z \in L$, $P'$ must also contain the edge $e_j$, which implies $e \in \Inv(e_j)$ (by \Cref{obs:basic-inverse-edge-property}).
		But this contradicts the assumption that $e \in (\Inv(e_i) \setminus \Inv(e_j) \setminus \Inv(S_{i - 1}))$. Therefore, $v \in L$. 
		
		We now prove that $X'$ is acyclic. 
		Assume the contrary. 
		Note that every cycle must contain at least one edge that is not present in $X$. Let $a$ and $b$ be an endpoints of this edge. 
		Then there are paths from $a$ to $b$ and from $b$ to $a$ in $X'$.
		As shown earlier, $a$ and $b$ lie in different parts of the partition $(V, V \setminus L)$.
		This implies that there is a path from $L$ to $V \setminus L$ in $X'$.
		But this is impossible, because $N_{X'}[L] = L$.
		
		It remains to compare the weights of $X$ and $X'$.
		First, we claim that all edges from the set $(\Inv(e_i) \setminus \Inv(e_j) \setminus \Inv(S_{i - 1}))$ are not present in $X$.
		Assume, to the contrary, that some edge $e = u \to v$ belongs to both $E(X)$ and this set. 
		From the discussion above, we have $u \in (V \setminus L)$ and $v \in L$.
		Since $X$ is acyclic, there exists an edge $e_t \in S$ such that $e \in \Inv(e_t)$.
		Recall that $S_j = S_{i - 1} \cup e_j$. 
		Because $e \notin \Inv(e_j) \cup \Inv(S_{i-1})$, we have $e \notin \Inv(S_j)$, which forces $t > j$.
		Now, \Cref{obs:basic-inverse-edge-property} combined with $e \in \Inv(e_i)$ and $e \in \Inv(e_t)$, implies that the directed path in $T$ from $v$ to $u$ contains both $e_i$ and $e_t$.
		Since $v \in L$ and every vertex in $L$ has exactly one outgoing edge in $T$, the prefix of this path must consist of edges $e_s, s_{s - 1}, \ldots, e_i$ for some $s$ with $i \le s < j$. 
		Hence, $e_i$ goes earlier than $e_t$ on this path.
		Consequently, $e_t$ is reachable from $e_i$ in $T$, contradicting the definition of the proper path cover, because $i < j < t$. 
		Therefore, $E(X) \cap (\Inv(e_i) \setminus \Inv(e_j) \setminus \Inv(S_{i - 1})) = \emptyset$. Consequently:
		\begin{equation*}
			\begin{split}
				\wt(X') &\ge \wt(X) - \wt(e_i) + \wt(e_j) \\& + \wt(\Inv(e_i) \setminus \Inv(e_j) \setminus \Inv(S_{i - 1})) \\& -\wt(\Inv(e_j) \setminus \Inv(e_i) \setminus \Inv(S_{i - 1}))
			\end{split}
		\end{equation*} 
		Let $I := \Inv(e_i) \setminus \Inv(S_{i - 1})$ and $J := \Inv(e_j) \setminus \Inv(S_{i - 1})$. 
		From the lemma statement we have
		\begin{equation*}
			0 \le \RemainingProfit(S_{i - 1}, e_i) - \RemainingProfit(S_{i - 1}, e_j) = \left(\wt(I) - \wt(e_i)\right) - \left(\wt(J) - \wt(e_j)\right)
		\end{equation*}
		Define $W$ as $I \cap J$. 
		Since $W \subseteq I$ and $W \subseteq J$, we have $\wt(I) - \wt(J) = \wt(I \setminus W) - \wt(J \setminus W)$. 
		Substituting this into the inequality yields:
		\begin{equation*}
			(\wt(I \setminus W) - \wt(e_i)) - (\wt(J \setminus W) - e_j) \ge 0
		\end{equation*}
		Observe that $I \setminus W = \Inv(e_i) \setminus \Inv(e_j) \setminus \Inv(S_{i - 1})$ and $J \setminus W = \Inv(e_j) \setminus \Inv(e_i) \setminus \Inv(S_{i - 1})$. 
		Plugging this into the inequality for $\wt(X')$, we obtain:
		\begin{equation*}
			\wt(X') \ge \wt(X) + (\wt(I \setminus W) - \wt(e_i)) - (\wt(J \setminus W) - e_j) \ge \wt(X)
		\end{equation*}
		Therefore, $X'$ is also a valid solution for $\mathcal{I}$.
	\end{claimproof}
	
	We exhaustively apply modifications of these three types to $X$ in an arbitrary order. 
	Note that this process is finite, because each modification strictly decreases the value of $\sum_{i = 1}^{n - 1} i \cdot [e_i \in S]$.
	Indeed, modifications of the first two types remove an edge from $S$, and modifications of the third type replace an edge from $e_j \in S$ with $e_i \notin S$, where $i < j$.
	In the end, we obtain a solution for $\mathcal{I}$ that satisfies all the required properties. 
\end{proof}

Let us demonstrate the power of \Cref{lemma:good-solution-set}.
First, it follows that there is a solution to $\mathcal{I}$, where only at most $|\mathcal{P}|$ tree edges with non-positive remaining profits are removed iteratively.
Second, \Cref{obs:win-by-positive-rp} guarantees that $k$ positive remaining profits are enough to obtain a solution.
Therefore, it is enough to consider edge subsets $S\subset E(T)$ of size at most $|\mathcal{P}|+k$.

On the other hand, \Cref{lemma:good-solution-set} gives us a limited number of choices of what edge to remove next.
Indeed, when we have to remove an edge, we know that its remaining profit is in $[-\wt(A), k]$. 
The third point of \Cref{lemma:good-solution-set} guarantees that we can always choose the earliest (with respect to the ordering) edge that belongs to $P_\ell$ with the value of remaining profit equal to $p$, for small integer values of $p$ and $\ell$.

\subsection{Solution when the removed tree edges are fixed}

It comes naturally from \Cref{lemma:good-solution-set} that we have to solve the following special version of the problem.
We have a fixed edge set $S\subset E(T)$ of bounded size, and we have to find a solution $G'$ to $(G,\wt,k)$ that contains all edges in $E(T-S)$ and contains no edge of $S$.
The following lemma shows that we are able to solve this subproblem efficiently. 
%

\begin{lemma}\label{lemma-reduce-to-WFAS}
	There is an algorithm with the running time 
	$2^{(\wt(A))^{\Oh(1)}}\cdot\polyn$ that, given an instance $\mathcal{I}=(G,\wt,k)$ of $\probMASMSTshort(\mathbb{Z}_{\ge 1})$, and given a tree edge set $S\subset E(T)$,
	either
	\begin{itemize}
		\item reports that $\mathcal{I}$ is a no-instance, if $\mathcal{I}$ is a no-instance, or
		\item outputs a solution to $\mathcal{I}$, if there is a solution $G'$ to $\mathcal{I}$ such that $E(T) \setminus E(G') = S$, or
		\item outputs any of the two outcomes above, otherwise.
	\end{itemize}
\end{lemma}

When solving an instance $(G, \wt, k)$ of $\probMASMSTshort$, we can always look at the problem from a different perspective. 
Instead of finding an acyclic subgraph of $G$ of weight at least $\wt(T) + k$, we can equivalently seek an edge set of weight at most $\wt(G) - \wt(T) - k$ whose removal makes $G$ acyclic. 
This dual problem is called $\probWDFAS$.

Formally, an instance of $\probWDFASshort$ is described by a tuple $(G, \wt, k, W)$, where $G$ is a directed graph with positive integer weights on edges given by the function $\wt: E(G) \to \mathbb{Z}_{\ge 1}$, and our goal is to find a set $S \subseteq E(G)$ of cardinality at most $k$ and weight at most $W$ such that $G - S$ is acyclic. 
In their recent work on flow augmentation in directed graphs \cite{KimWFAS}, Kim, Kratsch, Pilipczuk and Wahlstr\"{o}m have shown that $\probWDFASshort$ admits an $\classFPT$ algorithm when parameterized by $k$.
\begin{proposition}[\cite{KimWFAS}, Theorem $4.3$]\label{proposition:wdfas-fpt}
	$\probWDFASshort$ can be solved in time $2^{\mathcal{O}(k^8 \log k)} n^{\mathcal{O}(1)}$.
\end{proposition}

We employ this result in our proof of \Cref{lemma-reduce-to-WFAS}. 
The core idea is that given a fixed set of removed tree edges $S$, we can reduce our problem to a $\probWDFASshort$ on the graph $T - S + \Inv(S) + A$.
Moreover, we can bound $W$ and $k$ by $\wt(A)$ as otherwise the subgraph $T - S + \Inv(S)$ already has a sufficiently large weight.
We now move on to the formal proof of the lemma.
\begin{proof}[Proof of \Cref{lemma-reduce-to-WFAS}]
	Consider a tree edge set $S$. 
	Denote $X := T - S + A + \Inv(S)$ and let $k' := \wt(X) - \wt(T) - k$.
	Note that if there exists a solution $G'$ to $\mathcal{I}$ such that $E(T) \setminus E(G') = S$, then we can remove a set of edges from $X$ of total weight and cardinality at most $k'$ to obtain an acyclic graph: since $G'$ is a solution, it is acyclic, and \Cref{obs:blocked-edges-in-acyclic-subgraph} guarantees that $G' \subseteq X$. 
	Moreover, $\wt(G') \ge \wt(T) + k$. 
	Consequently, $\wt(X) - \wt(G') \le \wt(X) - \wt(T) - k = k'$.
	Additionally, observe that since all weights are at least one, from $\wt(X) - \wt(G') \le k'$ it also follows that $|E(X \setminus G')| \le k'$.
	
	Conversely, if $X'$ is an acyclic subgraph of $X$ obtained by removing edges of total weight at most $k'$, then $X'$ is also a valid solution for $\mathcal{I}$, because $\wt(X') \ge \wt(X) - k' = \wt(X) - (\wt(X) - \wt(T) - k) = \wt(T) + k$.
	
	Therefore, we obtain the algorithm required in the lemma statement as follows: solve the \probWDFASshort problem for the instance $\mathcal{I}' := (X, \wt', k', k')$, where the weight bound and the cardinality bound are both set to $k'$, and $\wt'$ is the restriction of $\wt$ to $X$.
	If $X'$ is a solution for $\mathcal{I}'$, then it is also a valid solution for $\mathcal{I}$, and we can output $X'$.
	Otherwise, we report that $\mathcal{I}$ is a no-instance.
	
	From the discussion above, if $\mathcal{I}$ is a no-instance, we correctly report it as such, because any solution produced by the algorithm would be feasible by construction.
	Moreover, if there exists a solution $G'$ for $\mathcal{I}$ such that $E(T) \setminus E(G') = S$, then there exists a solution for $\mathcal{I'}$, and our algorithm will output it as a solution for $\mathcal{I}$.
	Thus, the described algorithm satisfies all the required properties.
	
	A potential concern is that the value of $k'$ may be very large.
	To address this, we distinguish two cases.
	If $k' \ge \wt(A)$, then we can immediately produce a solution for $\mathcal{I'}$: consider the subgraph $X' := T - S + \Inv(S)$. 
	It is acyclic by \Cref{lemma:inverse-edges-independence} and $\wt(X') \ge \wt(X) - \wt(A) \ge \wt(X) - k'$.
	Otherwise, we have $k' < \wt(A)$. 
	By \Cref{proposition:wdfas-fpt}, in this case the instance $\mathcal{I}'$ can be solved in time
	\begin{equation*}
		2^{(k')^8 \log k'} \cdot \polyn \subseteq 2^{\wt(A)^8 \log \wt(A)} \cdot \polyn
	\end{equation*}
	This completes the proof.
\end{proof}

Note that the algorithm from \Cref{lemma-reduce-to-WFAS} can report false-negatives, if the choice of $S$ is incorrect, but always outputs a solution if the choice of $S$ is correct.
Therefore, we can run the algorithm of \Cref{lemma-reduce-to-WFAS} with many choices of $S$, and we will get a solution if at least one of these choices was correct.


\subsection{Putting all together}


\begin{algorithm}[ht]
  \SetKwFunction{FMain}{extend}
\SetKwProg{Fn}{function}{}{}
\Fn{\FMain{$\mathcal{I}$, $\mathcal{P}=(P_1, P_2,\ldots, P_t)$, $e_1,e_2,\ldots, e_{n-1}$, $S$, $i$}}{
	\If{$|S|\ge k+t$}{
		\Return{\textsc{Yes}}
	}
	\If{\Cref{lemma-reduce-to-WFAS} applied to $\mathcal{I}$ and $S$ returns \textsc{Yes}}{
		\Return{\textsc{Yes}}
		}
	\If{$i\ge n-1$}{
		\Return{\textsc{No}}
	}
	\eIf{$i=0$}
	{$\ell\gets 0$}{$\ell\gets $ integer in $[t]$ such that $e_i\in P_\ell$}
	\If(\tcp*[f]{next edge is in the same path}){$i>0$}{
		\ForEach(\tcp*[f]{it is the earliest edge with rem.\ profit $p$}){$p\in[1,k]$}{
			$j\gets $ smallest integer in $[i+1,n-1]$ with $e_j\in P_\ell$ and $\RemainingProfit(S, e_j)=p$, or $n$\;
			\If{\FMain($\mathcal{I}$, $\mathcal{P}$, $e_1,e_2,\ldots, e_{n-1}$, $S\cup\{e_j\}$, $j$) is \textsc{Yes}}
			{
				\Return{\textsc{Yes}}
			}
		}
	}
	\ForEach(\tcp*[f]{next edge is in another path $P_{\ell'}$}){$\ell'\in[\ell+1,t]$ } {
		\ForEach(\tcp*[f]{ it has remaining profit $p$}){$p\in [-\wt(A), k]$}{
			$j\gets $ smallest integer in $[i+1,n-1]$ with $e_j\in P_{\ell'}$ and $\RemainingProfit(S, e_j)=p$, or $n$\;
			
			\If{\FMain($\mathcal{I}$, $\mathcal{P}$, $e_1,e_2,\ldots, e_{n-1}$, $S\cup\{e_j\}$, $j$) is \textsc{Yes}}{
				\Return{\textsc{Yes}}
			}
		}
	}
	
	\Return{\textsc{No}}\;

}
	\SetKwIF{IfOr}{OrIf}{ElseOr}{if}{or}{or}{}{}			 
	\caption{The recursive subroutine \texttt{extend} of the algorithm of \Cref{thm:main-result}.
	It takes (as an input) an instance $\mathcal{I}$, a proper path cover $\mathcal{P}$ of $T$, and a $\mathcal{P}$-respecting ordering $e_1,\ldots, e_n$ of $E(T)$.
	\texttt{extend}($\mathcal{I}$, $\mathcal{P}$, $e_1,\ldots, e_n$, $\emptyset$, $0$) returns \textsc{Yes} iff $\mathcal{I}$ is a yes-instance.}
	\label{alg:main-recursive}

\end{algorithm}

To prove our first main result, \Cref{thm:main-result}, we pipeline the lemmas into an algorithm that solves $\probMASMSTshort(\mathbb{Z}_{\ge 1})$ in $2^{k^{\Oh(1)}}\cdot\polyn$ running time.
The core internal subroutine of this algorithm works with a proper path cover and outputs a solution in \classFPT-time (see Alg.~\ref{alg:main-recursive}).
For convenience, we restate \Cref{thm:main-result} here right before the proof. 

\thmFPTint*
\begin{proof}
	Given an instance $\mathcal{I}=(G,\wt,k)$, 
	the algorithm first applies \Cref{lemma:directed-paths-cover} to $\mathcal{I}$. 
	If it reports a solution to $\mathcal{I}$, the algorithm outputs it and stops.
	Otherwise, our algorithm obtains an equivalent instance $\mathcal{I}'$.
	\Cref{lemma:directed-paths-cover} guarantees that solving $\mathcal{I}'$ is equivalent to solving $\mathcal{I}$, as the resulting solution to $\mathcal{I}'$ can be transformed into a solution to $\mathcal{I}$.
	Without loss of generality, we put $\mathcal{I}:=\mathcal{I}'$, assuming that the algorithm always transforms the solution as required.
	
	Then, our algorithm identifies $T$, $A$ and $B$ in polynomial time.
	Then, for each $e\in E(T)$ the algorithm evaluates $\Profit(e)$.
	If $\Profit(e)\ge k$, the algorithm reports $T-e+\Inv(e)$ as a correct solution (see \Cref{obs:profits}) and stops.
	Then, the algorithm evaluates $\wt(A)$.
	If $\wt(A)\ge 3k$, the algorithm constructs a solution to $\mathcal{I}$ using \Cref{lemma:allowed-edges}, outputs it and stops.
	
	The algorithm now targets towards an application of \Cref{lemma:good-solution-set}.
	Let $D$ be the set of all vertices $v\in V(T)$ such that in-degree and out-degree of $v$ in $T$ are both exactly one.
	From \Cref{lemma:directed-paths-cover} we know that $|V(T)\setminus D|$ is at most $\Oh(k + |A|) = \Oh(k)$, by $|A|\le \wt(A)< 3k$.
	Let $V_A$ be the set of vertices incident to $A$ in $G$.
	Denote $C:=(V(T)\setminus D)\cup V_A$, $|C|=\Oh(k)$.
	The algorithm evaluates $C$ in polynomial time.
	
	Then, in polynomial time, the algorithm computes a set of directed paths of $T$ with both endpoints in $C$.
	This set is defined uniquely since each vertex in $V(T)\setminus C$ has in-degree and out-degree one in $T$, and has exactly $t=|C|-1$ paths.
	Clearly, this set of paths satisfies the first three points of the \Cref{def:path-cover} of proper path covers.
	The algorithm then arranges (in polynomial time) the paths in an order $P_1, P_2, \ldots, P_t$ according to the last point in \Cref{def:path-cover}.
	The formed sequence is a proper path cover $\mathcal{P}$ of $T$ with $t=\Oh(k)$ paths.
	After that, the algorithm construct a $\mathcal{P}$-respecting ordering $e_1, \ldots, e_{n - 1}$ of $E(T)$. 
	
	We finally move on to the ``heart'' of our algorithm.
	The algorithm performs a recursive branching subroutine \texttt{extend} (see \Cref{alg:main-recursive}) to guess a correct choice of $S$ given by \Cref{lemma:good-solution-set}.
	For the sake of clarity, \texttt{extend} does not return the solution (only reports \textsc{Yes}) in the pseudo-code, but it is straightforward to change it so it returns the solution.
	The algorithm runs \texttt{extend}$(\mathcal{I}, \mathcal{P}, e_1,\ldots, e_{n - 1}, \emptyset, 0)$ as an entry-point to this recursive procedure.
	This finishes the description of the algorithm.
	
	\medskip\noindent\textbf{Correctness.} 
	Note that \texttt{extend} cannot give false-positives, since it returns \textsc{Yes} only if $|S|\ge k+t$ or \Cref{lemma-reduce-to-WFAS} returns a solution (which cannot be false-positive).
	In the first case, $|S|$ should have at least $k$ edges with positive remaining profit (which gives a solution by \Cref{obs:win-by-positive-rp}), because \texttt{extend} takes at most one edge with non-positive $\RemainingProfit$ per each path in $\mathcal{P}$.
	
	Now assume that \texttt{extend} does not return \textsc{Yes} on $\mathcal{I}$.
	By \Cref{lemma:good-solution-set}, there is a set $S$ that satisfies all three properties from the statement, and applying \Cref{lemma-reduce-to-WFAS} to $\mathcal{I}$ and $S$ gives a solution to $\mathcal{I}$.
	\texttt{extend} never returned \textsc{Yes}, that is, \texttt{extend} was never called with this value of $S$.
	But the properties of $S$ imply that, if $e_i,e_j\in S$ are consecutive (with respect to the order $e_1, \ldots, e_{n - 1}$) edges in $S$, then either
	\begin{itemize}
		\item $e_j$ belongs to the same path as $e_i$, its remaining profit is positive, and there is no $k$ with $i<k<j$ such that $\RemainingProfit(S_i,e_k)=\RemainingProfit(S_i,e_j)$, or
		\item $e_j$ belongs to path $P_{\ell'}$ that goes after the path that $e_i$ belongs to, and there is no $k$ such that $e_k\in P_{\ell'}$ and $\RemainingProfit(S_i,e_k)=\RemainingProfit(S_i,e_j)$.
	\end{itemize}
	
	In the first case, \texttt{extend} will consider the correct choice of $p\in [k]$ that leads to this choice of $j$, since there is no other edge inbetween $e_i$ and $e_j$ with this value of remaining profit.
	In the second case, \texttt{extend} will consider the correct choice of both $\ell'\in [\ell+1,t]$ and $p\in [-\wt(A),k]$.
	This will lead to the correct choice of $j$ similarly to the first case.
	Therefore, we can prove by induction that \texttt{extend} will construct $S$ correctly in one of its recursion paths if never returned \textsc{Yes} in process.
	This contradiction finishes the correctness discussion.
	
	\medskip\noindent\textbf{Running time.}
	All parts of our algorithm are polynomial, except for the recursive branching subroutine.
	It invokes \Cref{lemma-reduce-to-WFAS}, which gives $2^{k^{\Oh(1)}}\cdot\polyn$ multiplier in the running time.
	To bound the number of recursive calls, note that the depth of recursion is at most $k+t+\Oh(1)$, that is, at most $\Oh(k)$, 
	while the number of possible recursive calls produced at one call is bounded by $k + (k+\wt(A)+1)\cdot |\mathcal{P}|=\Oh(k^2)$.
	Therefore, the total number of recursive calls is $(k^2)^{\Oh(k)}=2^{k^{\Oh(1)}}$. 
	The total running time is upper bounded with $2^{k^{\Oh(1)}}\cdot\polyn$.
	
	The proof is complete.
\end{proof}

\section{\classXP-algorithm for rational edge weights}\label{sec:xp-rational-weights}
This section is dedicated to the proof of \Cref{thm:main-result-rational}: given a graph $(G, \wt)$ with rational edge weights not less than one, an acyclic subgraph of $G$ of weight at least $\wt(T) + k$ can be found in time $n^{k^{\mathcal{O}(1)}} \cdot |\mathcal{I}|^{\Oh(1)}$.

We highlight that with rational edge weights, the problem becomes non-trivial even when $k = 1$ and $A = \emptyset$.
In this case, we need to check whether a set $S \subseteq E(T)$ with $\Profit(S) \ge 1$ exists.
Every such $S$ must contain an element with positive profit. 
Hence, for integer weights this implies that $\Profit(s) \ge 1$ for some $s \in S$, and we can consider only single-element sets.
In contrast, when edge weights are rational, each $\Profit(s)$ may be arbitrarily close to zero. 
Thus, to reach $\Profit(S) \ge 1$, we may now need to combine many small profits.
Consequently, there are no longer any natural bounds on the size of $S$, which significantly increases the difficulty of the problem.
Nevertheless, an \classXP algorithm must handle the case $k = 1$ in polynomial time.

We overcome this obstacle in \Cref{sec:maximizing-profit}.
Then, in \Cref{sec:comb-with-allowed-edges}, we address the additional challenges arising from the presence of allowed edges in $G$.

\subsection{Maximizing profit under restrictions}\label{sec:maximizing-profit}

A natural way to decide whether some set reaches a target profit is to find a set of maximum profit.
It turns out that if the elements are restricted to not share inverse edges, this can be done efficiently.
This subproblem is a key building block of our $\classXP$ algorithm.

For a digraph $(G,\wt)$, let $H_{G, \wt}$ be the undirected graph on $E(T)$ where each vertex $e$ has weight $\Profit(e)$ and an edge connects $e_1$ and $e_2$ iff $\Inv(e_1) \cap \Inv(e_2) \neq \emptyset$.
Recall that if the sets $\Inv(s)$ for $s \in S$ are pairwise disjoint, then the profit is additive. 
Thus, maximizing the profit of $S$ under this restriction is equivalent to finding the maximum-weight independent set in $H_{G, \wt}$, which is $\classNP$-hard on general graphs. 
Fortunately, $H_{G, \wt}$ has a special structure. 
\begin{lemma}\label{lemma:perfect-graph}
	For each edge-weighted digraph $(G, \wt)$, $H_{G, \wt}$ is a perfect graph. 
\end{lemma}

In our proof, we rely on the Strong Perfect Graph Theorem due to Chudnovsky, Robertson, Seymour and Thomas \cite{StrongPerfectGraphTheorem}. 
Recall that a hole is an induced cycle of length at least five, and an antihole is a hole in the complement of a graph. 
A hole (or an antihole) is called odd if the corresponding cycle has an odd length.
Perfect graphs are then characterized as follows.
\begin{proposition}[\cite{StrongPerfectGraphTheorem}]\label{prop:strong-perfect-graph-theorem}
	A graph is perfect if and only if it contains neither odd holes nor odd antiholes.
\end{proposition}

We proceed to the proof of \Cref{lemma:perfect-graph}.
\begin{proof}[Proof of \Cref{lemma:perfect-graph}]
	Let $H := H_{G, \wt}$. 
	Suppose that $H$ contains an odd hole $e_1, e_2, \ldots, e_{2t + 1}$, where $t \ge 2$.
	We aim to find a triangle formed by three vertices of this hole, and thus derive a contradiction.
	Our crucial tool is the following claim about the structure of $H$. 
	\begin{claim}\label{claim:H-structure}
		If $h_1h_2 \in E(H)$, then $T$ contains a directed path between $h_1$ and $h_2$.
		Moreover, if for a tree edge $h_3 \in E(T)$, $h_1$ and $h_2$ lie in different weakly connected components of $T - h_3$, then $\{h_1, h_2, h_3\}$ is a triangle in $H$.
	\end{claim}
	\begin{claimproof}
		If $h_1h_2 \in E(H)$, then $\Inv(h_1) \cap \Inv(h_2) \neq \emptyset$. 
		Consider an edge $e$ in this intersection, and let $u$ and $v$ be the starting and ending vertices of $e$, respectively.
		By \Cref{obs:basic-inverse-edge-property}, there exists a directed path $L$ from $v$ to $u$ in $T$,
		and since $e \in \Inv(h_1) \cap \Inv(h_2)$, both $h_1$ and $h_2$ lie on $L$. 
		Hence the directed path $P$ between $h_1$ and $h_2$ can be obtained as the subpath of $L$ between them.
		
		Now assume that $h_1$ and $h_2$ lie in different weakly connected components of $T - h_3$.
		Since the removal of $h_3$ from $T$ separates $h_1$ and $h_2$, $h_3$ must belong to $P$ (and hence to $L$).
		Thus by \Cref{obs:basic-inverse-edge-property}, $e \in \Inv(h_3)$, and therefore $e \in \Inv(h_1) \cap \Inv(h_2) \cap \Inv(h_3)$. 
		Consequently, $\{h_1, h_2, h_3\}$ is a triangle in $H$.
	\end{claimproof}
	
	Let $X$ and $Y$ be the weakly connected components of $T - e_1$, and let $X' := X + e_1$. 
	Without loss of generality, assume that $e_2$ belongs to $X$.
	We claim that every $e_i$ must lie inside $X'$. 
	Assume the contrary, and let $i$ be the first index such that $e_i$ does not belong to $X'$.
	By construction, $i > 2$, and hence $e_{i - 1}, e_i \neq e_1$.
	Then $e_{i - 1}$ lies in $X$ and $e_i$ lies in $Y$.
	Applying \Cref{claim:H-structure} with $h_1 = e_{i - 1}$, $h_2 = e_i$ and $h_3 = e_1$, we obtain that $\{e_1, e_{i - 1}, e_i\}$ is a triangle in $H$, which is a contradiction.
	Therefore, every $e_i$ belongs to $X'$.
	
	Denote the endpoint of $e_1$ in $Y$ by $v$, and let $T'$ be the underlying graph of $T$, rooted at $v$.
	Consider a tree edge $e \in E(T)$ and suppose that it leads from $a$ to $b$.
	We say that it goes up if $b$ is closer to $v$ than $a$ in $T'$.
	Otherwise, $e$ goes down.
	We also say that $e$ is an ancestor of $e'$ if $e$ lies on the path between $e'$ and $v$ in $T'$.
	Note that $e_1$ is an ancestor of every $e_i$: since $v \in Y$ and each $e_i \neq e_1$ lies in $X$, the path between $e_i$ and $v$ in $T'$ must contain $e_1$.
	
	Suppose that $T$ contains a directed path $P$, and let $e$ and $e'$ be its first and last edges, respectively.
	We claim that $e$ and $e'$ go in the same direction if and only if one is an ancestor of the other in $T'$.
	Let $u$ be the vertex of $P$ that is closest to $v$ in $T'$.
	Observe that $u$ divides $P$ into two parts in the following sense: edges encountered on the path before reaching $u$ go in one direction, and the remaining edges of $P$ go in the opposite direction.
	Therefore, $e$ and $e'$ go in the same direction if and only if $u$ is either the first or the last vertex of $P$.
	Assume first that $u$ is an inner vertex of $P$.
	In this case, $e$ and $e'$ lie in subtrees rooted at different children of $u$ in $T'$, and hence neither edge is an ancestor of the other.
	Conversely, assume that $u$ is an endpoint of $P$. 
	Without loss of generality, let $u$ be the first vertex of the path. 
	Then $e$ is the only edge on $P$ incident to $u$.
	Since $u$ is closer to $v$ than any vertex of $P - e$, $e$ separates $P - e$ and $v$ in $T'$.
	Therefore, $e$ is an ancestor of $e'$.
	
	Applying the statement above to the pair $e_1$ and $e_2$, we obtain that these edges go in the same direction. 
	Indeed, $e_1$ is an ancestor of $e_2$ in $T'$, and since $e_1e_2 \in E(H)$, by \Cref{claim:H-structure} there is a directed path between them in $T$.
	Similarly, $e_1$ and $e_{2t+1}$ go in the same direction, and therefore so do $e_2$ and $e_{2t + 1}$.
	Since indices of $e_2$ and $e_{2t + 1}$ have different parities, the directions of $e_2, e_3, \ldots, e_{2t + 1}$ cannot strictly alternate. In other words, there must be some $i \in [2, 2t]$ such that $e_i$ and $e_{i + 1}$ go in the same direction.
	We know that $e_i e_{i + 1} \in E(H)$, and hence there is a directed path between them in $T$. 
	As discussed above, combined with the same direction of these edges, this implies that one of them is an ancestor of the other.
	We consider two cases.
	
	\textbf{Case $e_i$ is an ancestor of $e_{i + 1}$.}
	By definition, $e_i$ lies on the path from $e_{i + 1}$ to $v$ in $T'$. 
	Since $e_{i + 1} \in X'$, the last edge of this path is $e_1$.
	Notice that $e_1 \neq e_i, e_{i + 1}$ because $i \in [2, 2t]$.
	Consequently, $e_i$ separates $e_{i + 1}$ and $e_1$ in $T'$, and hence they lie in different weakly connected components of $T - e_i$. 
	Let $e_{2t + 2} := e_1$. 
	Then some pair $e_j, e_{j + 1}$ with $j \in [i + 1, 2t + 1]$ also lies in different weakly connected components of $T - e_i$.
	Observe that $e_j e_{j + 1} \in E(H)$.
	Hence by \Cref{claim:H-structure}, $\{e_j, e_{j + 1}, e_i\}$ is a triangle in $H$.
	
	\textbf{Case $e_{i + 1}$ is an ancestor of $e_i$.}
	Similarly to the first case, we obtain that $e_i$ and $e_1$ lie in different weakly connected components of $T - e_{i + 1}$.
	Then there exists a pair of edges $e_j, e_{j+1}$ with the same property, where $j \in [1, i - 1]$.
	Again, together with \Cref{claim:H-structure}, this yields a triangle $\{e_j, e_{j + 1}, e_{i + 1}\}$ in $H$.
	
	In both cases, we found a triangle formed by three vertices of the hole, which leads us to a contradiction.
	Therefore, $H$ does not contain an odd hole.
	
	Next, suppose that $H$ contains an odd antihole $e_1, \ldots, e_{2t + 1}$. 
	Since an antihole of size $5$ is isomorphic to a hole of size $5$, we may assume that $t \ge 3$.
	Define a partial order on tree edges: $e$ precedes $e'$ if $T$ contains a directed path from $e$ to $e'$.
	Because $T$ is acyclic, we can topologically sort the edges with respect to this order.
	For a tree edge $e \in E(T)$, by $f(e)$ we denote its position in the sorted order.
	Without loss of generality, let $f(e_1)$ be the minimum among $f(e_1), f(e_2), \ldots, f(e_{2t + 1})$.
	Also, let $e_i$ be the edge with the minimum $f$ among $e_3, \ldots, e_6$.
	Observe that for each $i \in \{3, 4, 5, 6\}$ there exist $j_1, j_2 \in \{3, 4, 5, 6\}$ such that $|i - j_1| = 1$, $|i - j_2| > 1$ and $|j_1 - j_2| > 1$:
	\begin{itemize}
		\item for $i = 3$, take $j_1 = 4$ and $j_2 = 6$;
		\item for $i = 4$, take $j_1 = 3$ and $j_2 = 6$;
		\item two remaining cases are symmetrical.
	\end{itemize}
	Notice that the given properties imply that $i$, $j_1$ and $j_2$ are distinct. 
	Since $i, j_1, j_2 \in \{3, 4, 5, 6\}$ and $t \ge 3$, we also have $|i - 1|, |j_1 - 1|, |j_2 - 1| > 1$, and $i, j_1, j_2 < 2t + 1$. 
	Recall that the antihole contains an edge between any two non-consecutive vertices.
	Consequently, $e_1 e_i, e_1 e_{j_1}, e_1 e_{j_2}, e_i e_{j_2}, e_{j_1} e_{j_2} \in E(H)$, but $e_i e_{j_1} \notin E(H)$.
	We next make use of the following claim.
	\begin{claim}\label{claim:ordered-triple-in-H}
		Let $h_1, h_2, h_3 \in E(T)$ be a triple of tree edges such that $h_1 h_2, h_2 h_3 \in E(H)$ and $f(h_1) < f(h_2) < f(h_3)$.
		Then there is a directed path in $T$ that starts at $h_1$, ends at $h_3$ and goes through $h_2$. 
	\end{claim}
	\begin{claimproof}
		By \Cref{claim:H-structure}, $T$ contains a directed path $P_1$ between $h_1$ and $h_2$. 
		Since $f(h_1) < f(h_2)$, $P_1$ must go from $h_1$ to $h_2$.
		Similarly, $T$ contains a directed path $P_2$ that goes from $h_2$ to $h_3$.
		Then the concatenation of $P_1$ and $P_2$ yields a needed path. 
	\end{claimproof}
	
	By the choice of $e_1$ and $e_i$, we have $f(e_1) < f(e_i) < f(e_{j_2})$.
	Since $e_1 e_i \in E(H)$ and $e_i e_{j_2} \in E(H)$, \Cref{claim:ordered-triple-in-H} guarantees that there is a directed path $P$ that starts at $e_1$, ends at $e_{j_2}$ and goes through $e_i$. 
	Now consider a triple $e_1$, $e_{j_1}$ and $e_{j_2}$.
	By the choice of $e_1$, we know that $f(e_1) < f(e_{j_1})$ and $f(e_1) < f(e_{j_2})$. 
	Also, $\{e_1, e_{j_1}, e_{j_2}\}$ is a triangle in $H$.
	Therefore, we can apply \Cref{claim:ordered-triple-in-H} to this triple either in the order $e_1, e_{j_1}, e_{j_2}$ or in the order $e_1, e_{j_2}, e_{j_1}$, depending on the result of comparison of $f(e_{j_1})$ and $f(e_{j_2})$.
	As a result, we obtain that $T$ contains a directed path $P'$ that goes through these three edges in one of the orders. 
	We consider two cases.
	
	\textbf{Case $f(e_{j_1}) < f(e_{j_2})$}. 
	In this case $e_{j_1}$ lies between $e_1$ and $e_{j_2}$ on $P'$.
	Since the tree paths $P$ and $P'$ both go from $e_1$ to $e_{j_2}$, we have $P = P'$.
	Recall that $e_i$ also lies on $P$.
	From $e_1 e_{j_2} \in E(H)$ it follows that there exists an edge $e$ in the intersection $\Inv(e_1) \cap \Inv(e_{j_2})$.
	Suppose that it goes from $u$ to $v$.
	By \Cref{obs:basic-inverse-edge-property}, there is a directed path $L$ from $v$ to $u$ in $T$, and both $e_1$ and $e_{j_2}$ lie on $L$. 
	Hence $P \subseteq L$.
	In particular, $e_i, e_{j_1} \in L$. 
	Then by the same observation, we have $e \in \Inv(e_i)$ and $e \in \Inv(e_{j_1})$.
	Therefore, $\Inv(e_i) \cap \Inv(e_{j_1}) \neq \emptyset$ and thus $e_i e_{j_1} \in E(H)$.
	
	\textbf{Case $f(e_{j_1}) > f(e_{j_2})$.} 
	Since $P$ goes from $e_1$ to $e_{j_2}$ and $P'$ contains $e_{j_2}$ on the path from $e_1$ to $e_{j_1}$, we have $P \subseteq P'$.
	Hence $e_i \in P'$ because it lies on $P$. 
	Recall that $e_1 e_{j_1} \in E(H)$.
	Consider an edge $e \in \Inv(e_1) \cap \Inv(e_{j_1})$.
	By the same reasoning as in the first case, we obtain that $e \in \Inv(e_i)$.
	Consequently, $e \in \Inv(e_i) \cap \Inv(e_{j_1})$ and $e_i e_{j_1} \in E(H)$ again. 
	
	In both cases, we have $e_i e_{j_1} \in E(H)$, which contradicts the definition of an antihole because $|i - j_1| = 1$. 
	Therefore, $H$ contains neither odd holes nor odd antiholes. 
	Consequently, \Cref{prop:strong-perfect-graph-theorem} guarantees that $H$ is a perfect graph. 
\end{proof}

As shown by Grötschel, Lovász, and Schrijver in \cite{Grötschel1981}, one can find a maximum-weight independent set in a perfect graph with positive integer vertex weights in polynomial time.
We remark that this result can be extended to rational weights of an arbitrary sign.
\begin{proposition}[\cite{Grötschel1981}, Section $6$]\label{proposition:IS-in-perfect-graph}
	For a perfect graph $G$ with rational vertex weights, a maximum-weight independent set can be found in time polynomial in $|V(G)| + \log C$, where $C$ is the maximum among the numerators and denominators of the weights.
\end{proposition}
\begin{proof}
	First, let $V' \subseteq V(G)$ be the set of vertices with positive weights. 
	Observe that there always exists a maximum-weight independent set that consists only of vertices from $V'$.
	Additionally, the induced subgraph $G[V']$ is also perfect.
	Therefore, we may restrict our problem to $G[V']$ and assume that all weights are positive.
	
	Scaling all weights by a constant does not change which independent set has greater total weight.
	Thus we can multiply all weights by the least common multiple of their denominators, making them integers.
	Let $W$ be the maximum resulting weight.
	Note that $W$ is bounded by $C^{|V(G)|}$ because, in the worst case, each modified weight equals its initial numerator multiplied by the product of the denominators of the other $|V(G)| - 1$ weights.
	Hence $\log W \le |V(G)| \log C$.
	Therefore, on the modified instance the algorithm for positive integer weights runs in time polynomial in $|V(G)| + \log W$, which remains polynomial in $|V(G)| + \log C$.
\end{proof}

Even when certain tree edges are forbidden from $S$ (set $F$), others are required to be in $S$ (set $R$), and a set of blocked edges $B' \subseteq \Inv(R)$ may be the inverse of multiple edges in $S$, the maximum-profit $S$ can still be found in polynomial time. 

\begin{lemma}\label{lemma:restricted-max-IS}
	There is a polynomial-time algorithm that given a digraph $(G, \wt)$ with rational edge weights, sets $F, R \subseteq E(T)$ and $B' \subseteq \Inv(R)$, finds a set with the maximum $\Profit(S)$ over all $S \subseteq E(T)$ such that:
	\begin{itemize}
		\item $F \cap S = \emptyset$ and $R \subseteq S$;
		\item for every pair of distinct elements $s_1, s_2 \in S$, it holds that $\Inv(s_1) \cap \Inv(s_2) \subseteq B'$. 
	\end{itemize}
\end{lemma}
\begin{proof}
	Consider a tree edge $e \in E(T)$ and suppose there exists $r \in R$ with $r \neq e$ such that $\Inv(e) \cap \Inv(r) \not\subseteq B'$.
	Then $e$ cannot belong to any $S$ because given $R \subseteq S$, including $e$ in $S$ would violate the last constraint.
	Therefore, we can add all such $e$ to $F$ without changing the set of feasible solutions. 
	Denote the extended set of forbidden edges by $F'$.
	If $F' \cap R \neq \emptyset$, then the first two constraints cannot be satisfied together, and hence there are no feasible sets $S$. 
	From now on, we assume that $F'$ and $R$ are disjoint.
	
	Consider the graph $G' := G - B'$, and let $\wt'$ be the restriction of $\wt$ to $E(G')$. 
	Since $B'$ consists only of \textit{blocked} edges, we know that $E(T) \subseteq E(G')$ and therefore $\MSF(G', \wt') = T$.
	Thus from the definition of inverse edges, we have $\Inv(G', \wt', e) = \Inv(G, \wt, e) \setminus B'$ for each $e$. 
	Consequently, there is an edge between $e_1$ and $e_2$ in $H_{G', \wt'}$ if and only if $\Inv(G, \wt, e_1) \cap \Inv(G, \wt, e_2) \not\subseteq B'$.
	Let $E' := E(T) \setminus (F' \cup R)$ and let $I$ be the maximum-weight independent set in $H_{G', \wt'}[E']$.  
	We aim to show that $S := I \sqcup R$ satisfies all the required properties and maximizes the profit.
	
	First, by construction of $S$, it holds that $F' \cap S = \emptyset$ and $R \subseteq S$.
	Suppose that for some distinct $s_1, s_2 \in S$, we have $\Inv(G, \wt, s_1) \cap \Inv(G, \wt, s_2) \not\subseteq B'$. 
	Then there is an edge in $H_{G', \wt'}$ between $s_1$ and $s_2$.
	Since $I$ is an independent set in $H_{G', \wt'}$, it cannot contain both $s_1$ and $s_2$.
	Without loss of generality, assume that $e_1 \notin I$ and thus $e_1 \in R$.
	By the definition of $F'$, from $e_1 \in R$ and $\Inv(G, \wt, e_1) \cap \Inv(G, \wt, e_2) \not\subseteq B'$ it follows that $e_2 \in F'$, which is a contradiction because $S \cap F' = \emptyset$.
	Therefore, $S$ satisfies all three constraints.
	
	Next, we establish a connection between $\Profit(G, \wt, S)$ and the weight of $I$ in $H_{G', \wt'}[E']$.
	By definition, we have $\Profit(G, \wt, S) = \wt(\Inv(G, \wt, S)) - \wt(S)$.
	We express $\Inv(G, \wt, S)$ as $\Inv(G, \wt, R) \sqcup (\Inv(G, \wt, I) \setminus \Inv(G, \wt, R))$ and $S$ as $R \sqcup I$.
	Then we have
	\begin{equation*}
		\Profit(G, \wt, S) = (\wt(\Inv(G, \wt, R)) - \wt(R)) + (\wt(\Inv(G, \wt, I) \setminus \Inv(G, \wt, R)) - \wt(I))
	\end{equation*}
	Since $I \cap F' = \emptyset$, $\Inv(G, \wt, e)$ and $\Inv(G, \wt, r)$ with $e \in I$ and $r \in R$ can intersect only by the elements of $B'$.
	Therefore, $\Inv(G, \wt, I) \setminus \Inv(G, \wt, R) = \Inv(G, \wt, I) \setminus B' = \Inv(G', \wt', I)$.
	Plugging this into the equation above, we obtain $\Profit(G, \wt, S) = \Profit(G, \wt, R) + \Profit(G', \wt', I)$.
	Moreover, since $I$ is an independent set in $H_{G', \wt'}$, its elements do not share inverse edges in $G'$. 
	Thus, by the third point of \Cref{obs:profits}, we have 
	\begin{equation*}
		\Profit(G, \wt, S) = \Profit(G, \wt, R) + \sum_{e \in I} \Profit(G', \wt', e)
	\end{equation*}
	Notice that sum of profits in $G'$ over all elements of $I$ is exactly the weight of $I$ in $H_{G', \wt'}[E']$.
	
	Now, assume that there is another set $S' \subseteq E(T)$ that satisfies all the required constraints, and has a larger profit in $G$ than $S$. 
	We know that $S' \cap F' = \emptyset$ and $R \subseteq S'$.
	Let $I' := S' \setminus R$. 
	Observe that $I' \subseteq E'$.
	The third constraint guarantees that for each distinct $e_1, e_2 \in I'$, we have $(\Inv(G, \wt, e_1) \cap \Inv(G, \wt, e_2)) \subseteq B'$. 
	Consequently, $I'$ is an independent set in $H_{G', \wt'}[E']$.
	Then for analogous reasons we have that 
	\begin{equation*}
		\Profit(G, \wt, S') = \Profit(G, \wt, R) + \sum_{e \in I'} \Profit(G', \wt', e)
	\end{equation*}
	Combining this with $\Profit(G, \wt, S') > \Profit(G, \wt, S)$, we obtain that 
	\begin{equation*}
		\sum_{e \in I'} \Profit(G', \wt', e) > \sum_{e \in I} \Profit(G', \wt', e)
	\end{equation*}
	In other words, the weight of $I'$ in $H_{G', \wt'}[E']$ exceeds the weight of $I$, contradicting the maximality of $I$. 
	Therefore, $S$ maximizes the profit.
	
	Since $H_{G', \wt'}$ is a perfect graph by \Cref{lemma:perfect-graph}, its induced subgraph $H_{G', \wt'}[E']$ is also perfect. Therefore, by \Cref{proposition:IS-in-perfect-graph}, the set $I$ (and hence $S$) can be found in polynomial time. 
\end{proof}

Observe that no single triple $(F, R ,B')$ allows $S$ to have an arbitrary structure. 
Intuitively, if $B'$ is small, we forbid the elements of $S$ from having strongly overlapping sets of inverse edges. 
Conversely, if $B'$ is large, then each of its elements has a corresponding edge in $R \subseteq S$, and hence any feasible $S$ must be fairly large as well. 
Therefore, this lemma alone is not sufficient to solve even instances with $A = \emptyset$.
The following result completes the picture. 
\begin{lemma}\label{lemma:bounded-weight-of-B'}
	If $\mathcal{I} = (G, \wt, k)$ is a yes-instance of $\probMASMSTshort(\mathbb{Q}_{\ge 1})$, then there exists a solution $X$ to $\mathcal{I}$ with $S = E(T) \setminus E(X)$ such that the weight of 
	$$B' = \bigcup\limits_{\substack{s_1, s_2 \in S \\ s_1 \neq s_2}} (\Inv(s_1) \cap \Inv(s_2))$$
	is at most $(\wt(A) + k)^2$.
	Additionally, if $\mathcal{I}$ admits a solution of the form $T - S + \Inv(S) + A'$ for some $S \subseteq E(T)$ and $A' \subseteq A$, then $X$ can be required to have the same form. 
\end{lemma}
\begin{proof}
	If $\mathcal{I}$ does not admit a solution of the required form, let $X$ be an arbitrary solution that minimizes $|S|$. 
	Otherwise, let $X$ be a solution of the required form with the minimum $|S|$. 
	Since $X$ is acyclic, by \Cref{obs:blocked-edges-in-acyclic-subgraph} we have that $X \subseteq T - S + \Inv(S) + A$.
	Consequently, $\wt(X) \le \wt(T) + \Profit(S) + \wt(A)$.
	Assume that $\wt(B') > (\wt(A) + k)^2$.
	We aim to derive a contradiction by constructing a solution of the form $T - S' + \Inv(S')$ with $|S'| < |S|$. 
	
	Consider a subset $C \subseteq S$ of minimum size such that $B' \subseteq \Inv(C)$. 
	Since $\wt(B') > (\wt(A) + k)^2 \ge 0$, $B'$ is not empty, and hence we have $|C| \ge 1$.
	Additionally, by definition, every element of $B'$ is inverse to at least two edges from $S$.
	Hence for each $s \in S$, we have $B' \subseteq \Inv(S \setminus \{s\})$.
	Therefore, $|C| < |S|$. 
	Let $I := S \setminus C$.
	We now consider two cases.
	
	Suppose that $|C| < \wt(A) + k$. 
	Since $B' \subseteq \Inv(C)$, we have $B' = \bigcup_{c \in C} (\Inv(c) \cap B')$ which implies that 
	\begin{equation*}
		\sum_{c \in C} \wt(\Inv(c) \cap B') \ge \wt(B') > (\wt(A) + k)^2
	\end{equation*}
	Let $c \in C$ be an element with the maximum value of $\wt(\Inv(c) \cap B')$. 
	This weight is at least 
	\begin{equation*}
		\frac{\sum_{c \in C} \wt(\Inv(c) \cap B')}{|C|} \ge \frac{(\wt(A) + k)^2}{\wt(A) + k} = \wt(A) + k
	\end{equation*}
	Let $S' := S \setminus \{c\}$.
	As shown above, $B' \subseteq \Inv(S')$.
	Consequently, $\Inv(c) \cap B' \subseteq \Inv(c) \cap \Inv(S')$. 
	Since $B'$ is not empty, $|S| \ge 2$, and thus the subgraph $T - c + \Inv(c)$ is not a feasible solution of $\mathcal{I}$. 
	Because \Cref{lemma:inverse-edges-independence} guarantees its acyclicity, we have that $\Profit(c) < k$.
	It follows that
	\begin{equation*}
		\begin{split}
			\RemainingProfit(S', c) &= \wt(\Inv(c) \setminus \Inv(S')) - \wt(c) \\ &= (\wt(\Inv(c)) - \wt(c)) - \wt(\Inv(c) \cap \Inv(S')) \\ &\le (\wt(\Inv(c)) - \wt(c)) - \wt(\Inv(c) \cap B') \\ &\le \Profit(c) - (\wt(A) + k) \\ &\le -\wt(A)
		\end{split}
	\end{equation*}
	Then $X' := T - S' + \Inv(S')$ is also a solution of $\mathcal{I}$.
	Indeed, by \Cref{lemma:inverse-edges-independence} it is acyclic. 
	Moreover, 
	\begin{equation*}
		\begin{split}
			\wt(X') &= \wt(T) + \Profit(S') = \wt(T) + (\Profit(S') + \RemainingProfit(S', s)) - \RemainingProfit(S', s) \\ &= \wt(T) + \Profit(S) - \RemainingProfit(S', s) \\ &= (\wt(T) + \Profit(S) + \wt(A)) - (\RemainingProfit(S', s) + \wt(A)) \\ &\ge \wt(X) - (-\wt(A) + \wt(A)) \\ &= \wt(X) \ge \wt(T) + k
		\end{split}
	\end{equation*}
	But this contradicts the choice of $X$ because $|S'| < |S|$.
	
	Now suppose that $|C| \ge \wt(A) + k$.
	We first show that $|\Inv(I) \cap \Inv(C)| \ge |C|$.
	Assume the contrary.
	We order the elements of $C$ arbitrarily, and then for each $e \in \Inv(I) \cap \Inv(C)$ we mark the first element $c \in C$ in this order such that $e \in \Inv(c)$.
	Since $|\Inv(I) \cap \Inv(C)| < |C|$, at least one $c \in C$ remains unmarked. 
	Define $C'$ as $C \setminus \{c\}$.
	Because $|C'| < |C|$, from the minimality of $C$ it follows that there exists $b \in B'$ such that $b \notin \Inv(C')$. 
	On the other hand, $b \in \Inv(C)$.
	Hence $c$ is the only element in $C$ with $b \in \Inv(c)$.
	Recall that every element of $B'$ is inverse to at least two edges in $S$.
	Consequently, there must be $i \in I$ with $b \in \Inv(i)$, and therefore $b \in \Inv(I) \cap \Inv(C)$.
	But then $c$ would have been marked when processing $b$, because no other element of $C$ contains $b$ in its inverse set. 
	This contradicts the choice of $c$.
	Therefore, $|\Inv(I) \cap \Inv(C)| \ge |C|$. 
	
	Since $X$ is a feasible solution of $\mathcal{I}$, we have $\wt(X) \ge \wt(T) + k$. 
	Additionally, recall that $\wt(X) \le \wt(T) + \Profit(S) + \wt(A)$, and therefore $\Profit(S) \ge k - \wt(A)$. 
	On the other hand, observe that
	\begin{equation*}
		\begin{split}
			\Profit(S) &= \wt(\Inv(S)) - \wt(S) \\ &= (\wt(\Inv(C)) - \wt(C)) + (\wt(\Inv(I) \setminus \Inv(C)) - \wt(I)) \\ &= \Profit(C) + (\wt(\Inv(I)) - \wt(I)) - \wt(\Inv(I) \cap \Inv(C)) \\ &=
			\Profit(C) + \Profit(I) - \wt(\Inv(I) \cap \Inv(C))
		\end{split}
	\end{equation*}
	Since $|\Inv(I) \cap \Inv(C)| \ge |C| \ge \wt(A) + k$ and all weights are at least one, we have 
	\begin{equation*}
		\Profit(C) + \Profit(I) = \Profit(S) + \wt(\Inv(I) \cap \Inv(C)) \ge (k - \wt(A)) + (\wt(A) + k) = 2k	
	\end{equation*}
	This implies that at least one subgraph of $T - C + \Inv(C)$ and $T - I + \Inv(I)$ is a solution of $\mathcal{I}$. 
	Indeed, \Cref{lemma:inverse-edges-independence} guarantees that both of them are acyclic. 
	Moreover, $\max\{\Profit(C), \Profit(I)\} \ge 2k / 2 = k$. 
	This again contradicts the choice of $X$, since $1 \le |C| < |S|$ implies that both $C$ and $I$ are strictly smaller than $S$. 
	
	In both cases, we derive a contradiction. 
	Consequently, $\wt(B') \le (\wt(A) + k)^2$.
	This completes the proof. 
\end{proof}

In combination, \Cref{lemma:restricted-max-IS} and \Cref{lemma:bounded-weight-of-B'} supply us with a powerful tool. 
For example, an instance $(G, \wt, k)$ with $A = \emptyset$ can now be solved in time $n^{\Oh(k^2)}$ as follows.
Enumerate all $B' \subseteq B$ of weight at most $(\wt(A) + k)^2 = k^2$. 
For each $B'$, consider all sets $R$ formed by picking, for every $b \in B'$, one tree edge that breaks the cycle in $T + b$.
For every such pair $(B', R)$, apply the algorithm from \Cref{lemma:restricted-max-IS} with the given $B'$, $R$, and $F = \emptyset$, and obtain a maximum-profit set $S$ under these restrictions. 
If a solution exists, some such $S$ will yield $\Profit(S) \ge k$.

\subsection{Dealing with allowed edges}\label{sec:comb-with-allowed-edges}

When allowed edges come into play, we encounter a further complication: the structure of acyclic subgraphs no longer admits the same clean characterization as before. 
Given a set of removed tree edges $S$, the blocked edges of any acyclic subgraph are still confined to $\Inv(S)$. 
However, once the subgraph includes allowed edges, we lose the guarantee that the whole $\Inv(S)$ can be taken without producing cycles. 
In this section, we unravel this difficulty. 

Our approach, however, will differ significantly from the one used in the integral case, so let us first compare two versions of the problem (integral- and rational-weighted) in the context of allowed edges and discuss why the methods we developed for our \classFPT-algorithm for integral weights cannot be directly applied to rational weights.

Recall that, at a high level, our \classFPT-algorithm consisted of two parts: (i) finding a small family of tree edge sets that are candidates for removal, and (ii) for a fixed such set, reducing the problem to \probWDFAS.
In this scheme, allowed edges posed no major issue: we simply retained all of them in the reduced instance and let the black-box algorithm for \probWDFASshort deal with them.

Although the second part of the scheme above can be easily modified to work with rational weights, the first part's building blocks (\Cref{lemma:directed-paths-cover} and \Cref{lemma:good-solution-set}) rely heavily on the integrality of profits.
Since we could not come up with any analogues of these lemmas for rational weights, we had to abandon the entire scheme and deal with allowed edges by other means. 

On the positive side, we note that by \Cref{lemma:allowed-edges}, the total weight of allowed edges is still bounded by $3k$. 

\subsubsection{$\Inv$-respecting subgraphs}

As a starting point, we aim to mimic the setting without allowed edges. 
Although subgraphs of the form $T - S + \Inv(S) + A'$ with $A' \subseteq A$ are not necessarily acyclic, it remains tempting to search for solutions of this kind, because their weight equals $\wt(T) + \Profit(S) + \wt(A')$. 
This would enable us to reuse the ideas tied to edge profits. 
We call such subgraphs \emph{$\Inv$-respecting}, and extend this terminology to instances that admit a solution of this structure.

Fortunately, we can always seek a solution that is almost $\Inv$-respecting, in the sense that the total weight of the inverse edges it discards is at most $\wt(A)$. 
\begin{claim}{claimAlmostInvResp}\label{claim:inv-respecting-sol}
	If $\mathcal{I} = (G, \wt, k)$ is a yes-instance of $\probMASMSTshort(\mathbb{Q}_{\ge 1})$, there exists a solution $X \subseteq G$ to $\mathcal{I}$ with $S = E(T) \setminus E(X)$ such that $\wt(\Inv(S) \setminus E(X)) \le \wt(A)$.
\end{claim}
\begin{claimproof}
	Consider an arbitrary solution $X'$, and let $S' := E(T) \setminus E(X')$.
	Define $B'$ as $\Inv(S') \setminus E(X')$.
	If $\wt(B') \le \wt(A)$, we are done.
	Otherwise, consider the subgraph $X := X' + B' - A$.
	Note that $X$ contains the same set of tree edges as $X'$.
	Moreover, since $X'$ is acyclic, \Cref{obs:blocked-edges-in-acyclic-subgraph} guarantees that $X' \subseteq T - S' + \Inv(S') + A$.
	Consequently, $X \subseteq T - S' + \Inv(S')$.
	From the definition of $B'$ it follows that $\Inv(S') \subseteq X$.
	Therefore, $X = T - S' + \Inv(S')$.
	
	Since $B' \cap E(X') = \emptyset$, we have $\wt(X) \ge \wt(X') + \wt(B') - \wt(A)$.
	Moreover, as we know that $X = T - S' + \Inv(S')$, 
	by \Cref{lemma:inverse-edges-independence} it is acyclic.
	Therefore, $X$ is also a feasible solution of $\mathcal{I}$.
	
	Additionally, $X$ satisfies the needed property.
	Indeed, let $S := E(T) \setminus E(X)$.
	As shown above, we have $S = S'$.
	Thus from $\Inv(S') \subseteq X$ it follows that $\wt(\Inv(S) \setminus E(X)) = 0 \le \wt(A)$.
\end{claimproof}

This allows us to reduce an arbitrary instance to an $\Inv$-respecting one in time $m^{\wt(A)}$: we simply guess which inverse edges are discarded by a solution and remove them from the graph.
If the guess is correct, the resulting instance admits an $\Inv$-respecting solution. 

\subsubsection{Compressed representation of acyclic subgraphs}

Note that an $\Inv$-respecting solution is defined by two edge sets: $S$ and $A'$. 
Moreover, due to \Cref{lemma:allowed-edges}, the number of allowed edges is at most $3k$.
Therefore, we can enumerate all $A' \subseteq A$, and for each choice, search for a suitable set $S$. 
To ensure that $T - S + \Inv(S) + A'$ is acyclic for a fixed $A'$, we rely on the following alternative characterization. 

Let $V'$ be the set of vertices incident to $A'$.
Since $X = T - S + \Inv(S)$ is acyclic, any cycle in $X + A'$ must intersect $V'$.
Furthermore, such a cycle can be decomposed into directed paths in $X$ and edges from $A'$, where the endpoints of each fragment lie in $V'$.
Consequently, instead of working with $X + A'$ directly, we may analyze the compressed graph on $V'$ whose edges represent directed paths in $X$ and the edges of $A'$.
We formalize this idea as follows. 
\begin{claim}\label{claim:forbid-paths-for-acyclicity}
	Let $X \subseteq T + B$ be an acyclic subgraph. For $A' \subseteq A$, denote by $V'$ the set of vertices incident to $A'$. 
	Then $X + A'$ is acyclic if and only if there exists a set $\overline{P} \subseteq V' \times V'$ such that:
	\begin{itemize}
		\item for each $(u, v) \in \overline{P}$, there is no directed path from $u$ to $v$ in $X$;
		\item let $G'$ be the directed graph on $V'$ where for two distinct $u, v \in V'$, we have $(u, v) \in E(G')$ if $(u, v) \in A'$ or $(u, v) \notin \overline{P}$. 
		Then $G'$ is acyclic. 
	\end{itemize}
\end{claim}
\begin{claimproof}
	Suppose that $\overline{P}$ satisfies the required properties, but $X + A'$ has a cycle $C$. 
	Since $X$ is acyclic, $C$ must contain at least one edge from the set $A'$ and hence at least one vertex from $V'$. 
	Let $v_1, \ldots, v_t$ be a sequence of vertices from $V'$ that lie on $C$ in the order $C$ visits them, where $v_t = v_1$.
	For each pair of consecutive vertices $v_i$ and $v_{i + 1}$, $C$ contains a path $P$ from $v_i$ to $v_{i + 1}$. 
	By the choice of $v_1, \ldots, v_t$, none of the internal vertices of $P$ belong to $V'$.
	Observe that if the length of $P$ is more than one, then every edge on $P$ has at least one endpoint outside $V'$.
	In this case, $P$ cannot contain edges from $A'$ because all endpoints of $A'$ are in $V'$. 
	Thus, $P$ consists either of a single edge from $A'$ or entirely of edges from $X$.
	The latter case implies that $X$ contains a path from $v_i$ to $v_{i + 1}$, and hence $(v_i, v_{i + 1}) \notin \overline{P}$. 
	Since $X$ is acyclic and $P$ has length at least one, we also have $v_i \neq v_{i + 1}$. 
	Consequently, in both cases we obtain that $(v_i, v_{i + 1}) \in E(G')$.
	Then $v_1, \ldots, v_t$ is a cycle in $G'$, which contradicts the choice of $\overline{P}$. 
	
	Conversely, suppose that $X + A'$ is acyclic. 
	Consider the set $\overline{P}$ consisting of all pairs $(u, v) \in V' \times V'$ for which $X$ does not contain a directed path from $u$ to $v$. 
	Clearly, we need to verify only the second property. 
	Suppose, for a contradiction, that $G'$ has a cycle $v_1, \ldots, v_t$, where $v_t = v_1$. 
	For each pair of consecutive vertices $v_i, v_{i + 1}$, we know that $(v_i, v_{i + 1}) \in E(G')$. 
	Hence either $(v_i, v_{i + 1}) \in A'$, or $v_i \neq v_{i + 1}$ and $(v_i, v_{i + 1}) \notin \overline{P}$. 
	We claim that in both cases, $X + A'$ contains a directed path of length at least one from $v_i$ to $v_{i + 1}$. 
	In the former case, this path consists of the single edge $(v_i, v_{i + 1})$.
	In the latter case, the choice of $\overline{P}$ guarantees that there is a directed path from $v_i$ to $v_{i + 1}$ in $X$.
	Since $v_i \neq v_{i + 1}$, this path contains at least one edge. 
	Then $X + A'$ contains a directed cycle composed of these paths, contradicting the acyclicity of $X + A'$. 
\end{claimproof}

Therefore, we can ensure the acyclicity of $T - S + \Inv(S) + A'$ by choosing $\overline{P}$ that makes $G'$ acyclic and then prohibiting directed paths in $T - S + \Inv(S)$ for every $(u, v) \in \overline{P}$.

\subsubsection{Forbidding paths via constraint sets}
Surprisingly, subgraphs of the form $T - S + \Inv(S)$ that do not contain a given directed path can be described by a small collection of the same restrictions as used in \Cref{lemma:restricted-max-IS}. 
We call a triple $(F, R, B')$ from the statement of that lemma a \emph{restriction triple}, and any set of such triples a \emph{constraint set}.
A subgraph $X$, with $S = E(T) \setminus E(X)$, satisfies $(F, R, B')$ if the lemma conditions hold for $S$, and a constraint set if it satisfies at least one of its triples.

Recall that each blocked edge corresponds to a path in $T$ (see \Cref{obs:basic-inverse-edge-property}).
In our approach, for every edge from $B'$, we will need to identify the first and the last edge of the corresponding path that are absent from the subgraph.
We call a restriction triple \emph{strict} if its sets $F$ and $R$ uniquely determine these edges.
Formally, $(F, R, B')$ is called strict if for each $(u, v) \in B'$ there exist two edges $e_1, e_2 \in R$ (not necessarily distinct) such that:
\begin{itemize}
	\item $(u, v) \in \Inv(e_1) \cap \Inv(e_2)$;
	\item every edge of the path in $T$ from $v$ to the start of $e_1$ belongs to $F$;
	\item every edge of the path in $T$ from the end of $e_2$ to $u$ belongs to $F$.
\end{itemize}

If a subgraph $X = T - S + \Inv(S)$ satisfies such a strict triple, then $e_1, e_2 \in R \subseteq S$, so these two edges are absent from $X$.
Conversely, edges of the path before $e_1$ and after $e_2$ must belong to $X$ because they are present in $F$ and $S \cap F = \emptyset$. 

Given a strict restriction triple, we can filter the family of subgraphs $T - S + \Inv(S)$ satisfying it: a constraint set, constructed in the next lemma, keeps exactly those subgraphs that avoid the forbidden path.
This is by far the most difficult part of our \classXP-algorithm. 
\begin{lemma}\label{lemma:ban-single-path}
	Let $(G, \wt)$ be a digraph with edge weights. 
	For every $u, v \in V(G)$ and every strict restriction triple $(F_0, R_0, B_0')$, there exists a constraint set $\mathcal{C}$ such that:
	\begin{itemize}
		\item For every $S \subseteq E(T)$, let $X = T - S + \Inv(S)$. 
		If $X$ satisfies $\mathcal{C}$, then $X$ contains no directed path from $u$ to $v$. 
		Conversely, if $X$ satisfies $(F_0, R_0, B_0')$ and contains no directed path from $u$ to $v$, then $X$ satisfies $\mathcal{C}$.
		\item For each $(F, R, B') \in \mathcal{C}$, it holds that $B' = B_0'$, $F_0 \subseteq F$, $R_0 \subseteq R$, and each of the sets $F \setminus F_0$ and $R \setminus R_0$ can be represented as a union of at most $|B_0'| + 1$ directed paths in $T$.
	\end{itemize}
\end{lemma}
\begin{proof}
	We first show how to construct a constraint set with the needed properties if $T$ contains a directed path $P$ from $v$ to $u$.
	This case is essential for the general construction. 
	For two vertices on $P$, we say that the higher of them is the one that appears later on $P$ (i.e., closer to $u$).
	The second vertex is then lower. 
	
	Let $x \neq v$ be a vertex lying on $P$. 
	Our approach is based on the fact that if a subgraph of the form $T - S + \Inv(S)$ contains a path from $x$ to $v$, then there always exists such a path of the specific structure: it has a prefix that leads to some vertex $x'$ on $P$ that is lower than $x$ and consists of tree edges and \textit{exactly one blocked edge}.
	Moreover, given a blocked edge on this prefix, we can identify the corresponding $x'$. 
	
	For a blocked edge $e = (a, b)$, let $x'(e)$ be the first vertex on the path in $T$ from $b$ to $a$ that belongs to $P$.
	If this path does not contain vertices from $P$, we leave $x'(e)$ undefined. 
	We say that $e$ is $x$-active if $a$ is reachable from $x$ in $T$, $x'(e)$ is defined and lower than $x$.
	Intuitively, an $x$-active edge is a candidate for the unique blocked edge in the path prefix described above.
	Note that if $e$ is $x$-active, then the path in $T$ from $b$ to $a$ consists of paths from $b$ to $x'(e)$, from $x'(e)$ to $x$ (such a path exists because $x'(e)$ is lower than $x$ on $P$), and from $x$ to $a$. 
	
	When we consider a specific subgraph $T - S + \Inv(S)$, not every $x$-active edge can be used in the path prefix: indeed, some $x$-active edges may be not present in the subgraph. Even if an edge is present, some needed tree edges may be outside $T - S$. 
	We say that $e$ is $(x, S)$-active if $e$ is $x$-active, $e \in \Inv(S)$, and $T - S + e$ contains a directed path from $x$ to $x'(e)$. 
	
	We then formalize the discussion above as follows.  
	\begin{claim}\label{claim:descend-P-using-single-blocked-edge}
		Consider a vertex $x \neq v$ lying on $P$ and a subgraph $X = T - S + \Inv(S)$ with $S \subseteq E(T)$. 
		$X$ contains a directed path from $x$ to $v$ if and only if there exists $e \in B$ such that:
		\begin{itemize}
			\item $e$ is $(x, S)$-active;
			\item $X$ contains a directed path from $x'(e)$ to $v$.
		\end{itemize}
	\end{claim}
	\begin{claimproof}
		If $e$ is $(x, S)$-active, then $T - S + e \subseteq X$ because $e \in \Inv(S)$.
		By the definition of the $(x, S)$-active edge, $T - S + e$ contains a path from $x$ to $x'(e)$, and thus $X$ contains this path as well. 
		Together with a path from $x'(e)$ to $v$ in $X$, this yields a path from $x$ to $v$. 
		
		We now show the converse implication. 
		Consider an arbitrary path $L$ in $X$ from $x$ to $v$.
		Note that $L$ must contain at least one blocked edge. 
		Otherwise, $v$ would be reachable from $x$ in $T - S$, and together with the subpath of $P$ from $v$ to $x$, this would yield a cycle in $T$, which is a contradiction. 
		Let $e = (a, b)$ be the first blocked edge on $L$. 
		The prefix of $L$ before $e$ leads from $x$ to $a$ and contains no blocked edges.
		Hence, it lies entirely in $T - S$.  
		Consequently, $T$ contains the path $P'$ from $v$ to $a$ formed by concatenating the subpath of $P$ from $v$ to $x$ with the path from $x$ to $a$.
		
		Since $e$ is blocked, \Cref{obs:basic-inverse-edge-property} guarantees that $T$ contains a path from $b$ to $a$.
		Let $x''$ be the first vertex on it that lies on $P'$. 
		Since $a$ lies on $P'$, $x''$ is always defined. 
		Suppose that $x'' \neq b$, and pick any edge $e'$ on the path in $T$ from $b$ to $x''$.
		Let $R$ and $Q$ be the weakly connected components of $T - e'$ labeled so that $b \in R$ and $x'' \in Q$, and let $q$ be the endpoint of $e'$ that lies in $R$. 
		By the choice of $x''$, we know that $q$ does not belong to $P'$. 
		Since $T$ is a tree, this implies that $P'$ lies entirely within $Q$.
		Indeed, suppose some vertex $r$ from $P'$ was in $R$. 
		Then the unique path in $T$ between $r \in R$ and $x'' \in Q$ (which is a subpath of $P'$) must contain the edge $e'$, and hence the vertex $q$. 
		Thus $q$ would lie on $P'$, a contradiction. 
		Therefore, all vertices of $P'$ are in $Q$, and hence so is the vertex $v$. 
		By \Cref{claim:banned-edges-structure}, the only edge in $T + B$ that leads from $R$ to $Q$ is $e'$.
		Note that the suffix of $L$ after the edge $e$ leads from $b \in R$ to $v \in Q$, and hence it contains $e'$. 
		Then $e' \in E(X)$, and since it is a tree edge, we have $e' \in E(T - S)$.  
		Therefore, the entire path from $b$ to $x''$ lies in $T - S$.
		Combined with a path from $x$ to $a$ in $T - S$ and the edge $e$, it forms a path from $x$ to $x''$ in $T - S + e$.  
		
		We now consider two cases based on the position of $x''$ on $P'$. 
		First, assume that $x''$ lies on the subpath of $P'$ from $x$ to $a$ (which lies entirely in $T - S$).
		Then there is also a path from $x''$ to $a$ in this subgraph. 
		Combined with the path from $b$ to $x''$ in $T - S$, this implies that $T - S$ contains a path from $b$ to $a$. 
		By definition, $a \neq b$. 
		Thus, this path together with the edge $e$ yields a cycle in $T - S + e \subseteq X$.
		But \Cref{lemma:inverse-edges-independence} states that $X$ is acyclic, a contradiction. 
		Therefore, $x''$ cannot lie on this part of $P'$.
		
		It follows that $x''$ lies on $P$ and is lower than $x$. 
		We claim that $x'(e) = x''$. 
		Assume the contrary. 
		Then $x'(e)$ appears earlier than $x''$ on the path from $b$ to $a$. 
		By the choice of $x''$, this would mean that $x'(e)$ does not lie on $P'$.
		Since $P'$ contains the subpath of $P$ from $v$ to $x$, we obtain that $x'(e)$ is higher than $x$. 
		But in this case, $T$ contains a path from $x'(e)$ to $x''$ (a part of the path from $b$ to $a$ between them) and a path from $x''$ to $x'(e)$ (a subpath of $P$ between them), which contradicts its acyclicity. 
		
		We conclude by verifying that $e$ satisfies all the required properties. 
		As shown above, $x'(e) = x''$ is lower than $x$ and $T - S$ contains a path from $x$ to $a$.
		Hence, $e$ is $x$-active. 
		Moreover, $e \in E(X) = E(T - S + \Inv(S))$, and $T - S + e$ contains a path from $x$ to $x'' = x'(e)$.
		Therefore, $e$ is $(x, S)$-active. 
		Finally, we need to show that $X$ contains a path from $x''$ to $v$. 
		If $x'' = b$, then the suffix of $L$ after the edge $e$ goes from $x''$ to $v$. 
		Otherwise, the path from $b$ to $x''$ in $T - S$ is non-empty. 
		Consider the last edge on this path. 
		As discussed above, $L$ must contain this edge, and hence $x''$ belongs to $L$ as well. 
		Then we again have a suffix of $L$ which leads from $x''$ to $v$. 
		Since $L \subseteq X$, this completes the proof. 
	\end{claimproof}
	
	This allows us to forbid a path from $x$ to $v$ in $T - S + \Inv(S)$ in a recursive manner: for each $x$-active edge, we either impose some restrictions on $S$ to ensure that this edge is not $(x, S)$-active, or forbid a path from the corresponding $x'$ to $v$. 
	Fortunately, if the subgraph satisfies $(F_0, R_0, B_0')$, this can be done efficiently. 
	
	We first show that $(x, S)$-activeness of edges from $B_0'$ does not actually depend on $S$ and is determined by the set $R_0$. 
	\begin{claim}\label{claim:(x-S)-activeness-for-B0}
		Consider a vertex $x$ lying on $P$.
		For each $x$-active edge $e \in B_0'$ and each subgraph $X = T - S + \Inv(S)$ satisfying $(F_0, R_0, B_0')$, $e$ is $(x, S)$-active if and only if $T - R_0 + e$ contains a path from $x$ to $x'(e)$. 
	\end{claim}
	\begin{claimproof}
		Since $X$ satisfies $(F_0, R_0, B_0')$, we have $e \in B_0' \subseteq \Inv(R_0)$ and $R_0 \subseteq S$, which implies that $e \in \Inv(S)$. 
		Therefore, it remains to check that there is a path from $x$ to $x'(e)$ in $T - S + e$ if and only if the same path is present in $T - R_0 + e$. 
		One implication is trivial: if such a path exists in $T - S + e$, then $T - R_0 + e$ also contains this path because $R_0 \subseteq S$ and hence $T - S + e \subseteq T - R_0 + e$. 
		We now show the converse implication. 
		
		Suppose that there is a path from $x$ to $x'(e)$ in $T - R_0 + e$.
		Since $e$ is $x$-active, $x'(e)$ is lower than $x$ on $P$, and hence $x$ is reachable from $x'(e)$ in $T$. 
		Then the path from $x$ to $x'(e)$ must use the edge $e$. 
		Let $a$ and $b$ be the starting and ending vertices of $e$, respectively.
		It follows that $T - R_0$ contains paths from $x$ to $a$ and from $b$ to $x'(e)$. 
		
		Since $(F_0, R_0, B_0')$ is a strict restriction triple and $e \in B_0'$, the path in $T$ from $b$ to $a$ contains edges $e_1, e_2 \in R_0$. 
		Moreover, it is guaranteed that all edges on the paths from $b$ to the start of $e_1$ and from the end of $e_2$ to $a$ belong to $F_0$. 
		
		Since $e$ is $x$-active, $x$ lies on the path from $b$ to $a$. 
		Then $x$ appears later than $e_2$ on this path because there is a path from $x$ to $a$ in $T - R_0$. 
		Consequently, all edges of the path from $x$ to $a$ are in $F_0$. 
		Since $X$ satisfies $(F_0, R_0, B_0')$, $S \cap F_0 = \emptyset$, and therefore the path from $x$ to $a$ is present in $T - S$.
		Symmetrically, we can show that the path from $b$ to $x'(e)$ is also present in $T - S$. 
		Then these paths together with $e$ form a path from $x$ to $x'(e)$ in $T - S + e$. 
	\end{claimproof}
	
	It follows that for each $x$-active edge $e \in B_0'$ such that $T - R_0 + e$ contains a path from $x$ to $x'(e)$, we must forbid a path from $x'(e)$ to $v$. 
	
	For the edges from $B \setminus B_0'$, we use a different idea. 
	We first characterize their $(x, S)$-activeness using only the interaction between $P$ and $S$. 
	\begin{claim}\label{claim:(x-S)-activeness-outside-B0'}
		Let $X = T - S + \Inv(S)$ with $S \subseteq E(T)$ be a subgraph that satisfies $(F_0, R_0, B_0')$. 
		Consider a vertex $x$ lying on $P$. 
		An $x$-active edge $e \in B \setminus B_0'$ is $(x, S)$-active if and only if the subpath of $P$ from $x'(e)$ to $x$ contains an edge from $S$. 
	\end{claim}
	\begin{claimproof}
		Let $a$ and $b$ be the starting and ending vertices of $e$, respectively.
		Since $e$ is $x$-active, the path from $b$ to $a$ in $T$ consists of paths from $b$ to $x'(e)$, from $x'(e)$ to $x$, and from $x$ to $a$. 
		
		Suppose that $e$ is $(x, S)$-active. 
		Then $e \in \Inv(s)$ for some $s \in S$.
		By \Cref{obs:basic-inverse-edge-property}, all such $s$ lie on the path from $b$ to $a$ in $T$. 
		On the other hand, there must be a path in $T - S + e$ from $x$ to $x'(e)$.
		Since $x'(e)$ is lower than $x$, this path must use the edge $e$.
		Consequently, the paths in $T$ from $x$ to $a$ and from $b$ to $x'(e)$ do not contain edges from $S$. 
		Then $s$ must lie on the path from $x'(e)$ to $x$. 
		
		Conversely, assume that there is an edge $s \in S$ on the path from $x'(e)$ to $x$. 
		Then this edge also lies on the path from $b$ to $a$ in $T$. 
		By \Cref{obs:basic-inverse-edge-property}, it follows that $e \in \Inv(s)$.
		Since $X$ satisfies $(F_0, R_0, B_0')$ and $e \in B \setminus B_0'$, $e$ cannot be an inverse edge to multiple elements of $S$, and hence $s$ is the only edge on the path from $b$ to $a$ that lies in $S$. 
		Then $T - S$ contains the paths from $x$ to $a$ and from $b$ to $x'(e)$.
		This implies that $T - S + e$ contains a path from $x$ to $x'(e)$. 
		Therefore, $e$ is $(x, S)$-active. 
	\end{claimproof}
	
	As we show below, this implies that if some edge from $B \setminus B_0'$ can be used as a unique blocked edge in the prefix of the path from $x$ to $v$, then the edge $e \in B \setminus B_0'$ with the lowest $x'(e)$ can also be used on such a prefix. 
	
	\begin{claim}\label{claim:can-descend-through-lowest-outside-B0}
		Let $X = T - S + \Inv(S)$ with $S \subseteq E(T)$ be a subgraph that satisfies $(F_0, R_0, B_0')$. 
		Consider a vertex $x$ lying on $P$. 
		If there exists $e' \in B \setminus B_0'$ such that $e'$ is $(x, S)$-active and $X$ contains a path from $x'(e')$ to $v$, then $e \in B \setminus B_0'$ with the lowest $x'(e)$ is also $(x, S)$-active and $X$ contains a path from $x'(e)$ to $v$. 
	\end{claim}
	\begin{claimproof}
		Since $e'$ is $(x, S)$-active, \Cref{claim:(x-S)-activeness-outside-B0'} guarantees that the subpath of $P$ from $x'(e')$ to $x$ contains an edge $s \in S$. 
		By the choice of $e$, the vertex $x'(e)$ is not higher than $x'(e')$ on $P$, and hence $s$ also belongs to the subpath of $P$ from $x'(e)$ to $x$.
		Consequently, by the same claim, $e$ is $(x, S)$-active.
		
		We now aim to show that a subpath of $P$ from $x'(e)$ to $x'(e')$ does not contain edges from $S$. 
		Assume the contrary, and let $s'$ be such an edge. 
		Note that $s \neq s'$. 
		It follows that the subpath of $P$ from $x'(e)$ to $x$ contains two edges from $S$. 
		Since $e$ is $x$-active, these edges also lie on the path from $b$ to $a$ in $T$.
		By \Cref{obs:basic-inverse-edge-property}, we have $e \in \Inv(s) \cap \Inv(s')$. 
		However, this contradicts the fact that $X$ satisfies $(F_0, R_0, B_0')$ because $e \notin B_0'$.
		Therefore, $T - S \subseteq X$ contains the path from $x'(e)$ to $x'(e')$. 
		Consequently, a path from $x'(e)$ to $v$ in $X$ can be obtained by concatenating the path from $x'(e)$ to $x'(e')$ with the path from $x'(e')$ to $v$. 
		
	\end{claimproof}
	
	We now construct a constraint set that forbids a path from $u$ to $v$. 
	To obtain such a set, we use a recursive algorithm $\mathcal{A}$ that takes as an input $U$ and $F \supseteq F_0$, where $U$ is a set of vertices lying on $P$ and $F \subseteq E(T)$. 
	The constraint set returned by this algorithm will have the following properties: for every $S \subseteq E(T)$, define $X := T - S + \Inv(S)$.
	If $X$ satisfies $\mathcal{A}(U, F)$, then $X$ contains no directed path from $u'$ to $v$ for any $u' \in U$.
	Conversely, if $X$ satisfies $(F, R_0, B_0')$ and contains no directed path from $u'$ to $v$ for any $u' \in U$, then $X$ satisfies $\mathcal{A}(U, F)$.
	Moreover, for each $(F', R', B') \in \mathcal{A}(U, F)$, we will have $F \subseteq F', R' = R_0, B' = B_0'$, and $F' \setminus F$ will be a union of at most $|B_0'| + |U|$ directed paths in $T$. 
	
	We output the result of the call $\mathcal{A}(\{u\}, F_0)$.
	Note that if $\mathcal{A}$ is correct, this matches the lemma statement. 
	
	Let us start with two trivial cases.
	If $v \in U$, then $\mathcal{A}(U, F) = \emptyset$.
	This output is correct because the antecedents of both implications are false: no subgraph satisfies the empty constraint set by definition, and every subgraph contains a path from $v \in U$ to $v$. 
	If $U = \emptyset$, then there are no paths to forbid, so we simply set $\mathcal{A}(U, F) = \{(F, R_0, B_0')\}$. 
	From now on, we assume that $U \neq \emptyset$ and $v \notin U$.
	
	Let $x$ be the highest vertex in $U$. 
	Note that $x \neq v$.
	We first find all edges $e \in B_0'$ such that $e$ is $x$-active and $T - R_0 + e$ contains a path from $x$ to $x'(e)$. 
	Let $U'$ be the set of vertices $x'(e)$ for all such edges. 
	We define $U_1$ as $(U \cup U') \setminus \{x\}$.
	
	If $B \setminus B_0'$ does not contain $x$-active edges, we define $U_2 := U_1$ and $F' := \emptyset$. 
	Otherwise, let $e \in B \setminus B_0'$ be an $x$-active edge with the lowest $x'(e)$.
	We define $U_2$ as $U_1 \cup \{x'(e)\}$. 
	Since $e$ is $x$-active, $x'(e)$ is lower on $P$ than $x$. 
	Let $F'$ be the set of edges on the subpath of $P$ from $x'(e)$ to $x$. 
	Finally, we set $\mathcal{A}(U, F) = \mathcal{A}(U_1, F \cup F') \cup \mathcal{A}(U_2, F)$. 
	
	Note that the highest vertices in $U_1$ and $U_2$ are strictly lower than $x$. 
	Indeed, by construction, $U_1 \setminus U$ and $U_2 \setminus U$ contain only vertices of the form $x'(e)$ for $x$-active edges $e$. 
	By the definition of $x$-activeness, this implies that $x'(e)$ is lower than $x$.
	Additionally, $U_1$ and $U_2$ do not contain $x$. 
	
	Therefore, the computation graph is acyclic. 
	We now prove the correctness of $\mathcal{A}$. 
	Let us start with the structure of a fixed triple $(\tilde{F}, \tilde{R}, \tilde{B}') \in \mathcal{A}(U, F)$. 
	By construction, we have $\tilde{R} = R_0$ and $\tilde{B}' = B_0'$. 
	Moreover, in both recursive calls the set $F$ can only expand. At a leaf, $\mathcal{A}$ generates a restriction triple with the given set $F$. 
	Hence, $F \subseteq \tilde{F}$. 
	We prove the remaining structural property in a separate claim. 
	\begin{claim}
		For each $(\tilde{F}, \tilde{R}, \tilde{B}') \in \mathcal{A}(U, F)$, the set $\tilde{F} \setminus F$ is a union of at most $|B_0'| + |U|$ directed paths in $T$. 
	\end{claim}
	\begin{claimproof}
		For a fixed triple $(\tilde{F}, \tilde{R}, \tilde{B}') \in \mathcal{A}(U, F)$, consider the branch of the recursion that produced it. 
		Let the sequence of arguments of $\mathcal{A}$ along this branch be $(U_1', F_1'), \ldots, (U_t', F_t')$, where $(U_1', F_1') = (U, F)$ and $F_t' = \tilde{F}$ (the last equation follows from the form of produced constraint sets in the leaves of $\mathcal{A}$).
		The structure of recursive calls of $\mathcal{A}$ implies that $F'_i \subseteq F'_{i + 1}$ for each $i$. 
		We denote the highest vertex in $U_i'$ by $x_i$. 
		As discussed above, $x_{i + 1}$ is always lower than $x_i$. 
		Making a recursive call, $\mathcal{A}$ adds to $U$ only the elements of the form $x'(e)$ for some $e \in B$. 
		Moreover, if $x'(e) \in U'_{i + 1} \setminus U'_i$, then the edge $e$ is always $x_i$-active.  
		
		For a fixed $e \in B$, let $I(e) \subseteq \{1, \ldots, t\}$ be a set of indices $i$ with $x'(e) \in U_i'$. 
		We claim that $I(e)$ forms a contiguous subsegment. 
		Assume the contrary. 
		Then there exist indices $\ell, r$ with $r > \ell + 1$ such that $\ell, r \in I(e)$, but $\{\ell + 1, \ldots, r - 1\} \cap I(e) = \emptyset$. 
		Hence, $x'(e) \in U'_{\ell} \setminus U'_{\ell + 1}$. 
		Recall that in each recursive call, $\mathcal{A}$ removes only the highest vertex of the set. 
		Therefore, $x'(e) = x_{\ell}$. 
		On the other hand, since $x'(e) \in U'_{r} \setminus U'_{r - 1}$, the edge $e$ must be $(x_{r - 1})$-active. 
		By the definition of $x$-activeness, this implies that $x_{\ell} = x'(e)$ is lower than $x_{r - 1}$, which leads us to a contradiction because $\ell < r - 1$. 
		
		For an edge $e$, let $\ell(e)$ be the minimum element of $I(e)$. If $I(e) = \emptyset$, we set $\ell(e) = \infty$. 
		Also, let $\ell_i$ be the number of $e \in B_0'$ with $\ell(e) = i$, and let $p_i := \sum_{j \le i} \ell_j$. 
		To bound the growth of $F_i'$, we define the potential function $\varphi(i) := |U_i'| + (|B_0'| - p_i)$. 
		Note that $\varphi(1) = |U| + (|B_0'| - p_1) \le |B_0'| + |U|$. 
		Also, since $|U_t'| \ge 0$ and $p_t \le |B_0'|$, we have $\varphi(t) \ge 0$. 
		
		We now analyze how $\varphi$ changes during a single transition. First, note that 
		\begin{equation*}
			\varphi(i + 1) - \varphi(i) = (|U_{i + 1}'| - |U_i'|) + (p_i - p_{i + 1}) = |U_{i + 1}' \setminus U_i'| - |U_i' \setminus U_{i + 1}'| - \ell_{i + 1}
		\end{equation*}
		Recall that $U_{i + 1}' \setminus U_i'$ consists of elements of the form $x'(e)$, where $e \in B$. 
		Consider $x'(e) \in U_{i + 1}' \setminus U_i'$ such that $e \in B_0'$. 
		Since $I(e)$ is a contiguous subsegment, $i \notin I(e)$, and $i + 1 \in I(e)$, we have $\ell(e) = i + 1$. 
		Consequently, we can associate every such $x'(e)$ with the edge $e$ that increases $\ell_{i + 1}$. 
		Let $U'' := U_{i + 1}' \setminus U_i' \setminus \{x'(e) \mid e \in B_0'\}$. 
		We obtain that
		\begin{equation*}
			\varphi(i + 1) - \varphi(i) \le |U''| - |U_i' \setminus U_{i + 1}'|
		\end{equation*}
		Furthermore, since in each recursive call $\mathcal{A}$ removes exactly one element of the set, we have $|U_i' \setminus U_{i + 1}'| = 1$, and hence the difference above is at most $|U''| - 1$. 
		
		We now distinguish two cases based on the type of transition between $(U_i', F_i')$ and $(U_{i + 1}', F_{i + 1}')$. 
		Recall that $\mathcal{A}$ makes two recursive calls. 
		In the first of them (with parameters $U_1$ and $F \cup F'$), it adds to $U$ only elements $x'(e)$ with $e \in B_0'$.
		Hence, $|U''| = 0$, and therefore $\varphi(i + 1) - \varphi(i) \le -1$. 
		Moreover, the new set $F$ in this recursive call is obtained by adding one subpath of $P$ (and therefore one directed path in $T$). 
		
		In the second recursive call, $\mathcal{A}$ adds to $U$ at most one element $x'(e)$ with $e \notin B_0'$.
		Consequently, $|U''| \le 1$ and $\varphi(i + 1) - \varphi(i) \le 0$. 
		Additionally, in this recursive call the set $F$ remains unchanged. 
		
		Therefore, the potential is never increases and each growth of $F'$ by a single directed path in $T$ is associated with a decrease of the potential by at least one. 
		Combining this with the fact that $\varphi(1) - \varphi(t) \le |B_0'| + |U|$, we derive that $\tilde{F} \setminus F$ is a union of at most $|B_0'| + |U|$ directed paths in $T$. 
	\end{claimproof}
	
	For the remaining properties of $\mathcal{A}$, we apply an induction. 
	The base cases ($U = \emptyset$ and $v \in U$) are considered above. 
	The following claim shows the correctness of a single transition. 
	\begin{claim}\label{claim:recursive-constraint-correctness}
		If $\mathcal{C}_1 := \mathcal{A}(U_1, F \cup F')$ and $\mathcal{C}_2 := \mathcal{A}(U_2, F)$ are correct, then so is $\mathcal{C} := \mathcal{A}(U, F)$. 
	\end{claim}
	\begin{claimproof}
		Assume, for a contradiction, that $X = T - S + \Inv(S)$ satisfies $\mathcal{C}$ but contains a path from $x$ to $v$ for some $x \in U$.	
		Since $\mathcal{C} = \mathcal{C}_1 \cup \mathcal{C}_2$, it follows that $X$ satisfies either $\mathcal{C}_1$ or $\mathcal{C}_2$. 
		Because $U_1 \subseteq U_2$, the inductive hypothesis implies in both cases that $X$ contains no path from any $u' \in U_1$ to $v$.
		By the definition of $U_1$, $U \setminus U_1$ consists of a single element, which is the highest vertex in $U$. 
		Therefore, $x$ must be exactly this vertex.  
		
		Since $v \notin U$, we know that $x \neq v$. 
		\Cref{claim:descend-P-using-single-blocked-edge} then provides an edge $e \in B$ such that $e$ is $(x, S)$-active and $X$ contains a path from $x'(e)$ to $v$.
		Also, as discussed above, each triple $(\tilde{F}, \tilde{R}, \tilde{B}') \in \mathcal{C}$ obeys $\tilde{R} = R_0$, $\tilde{B}' = B_0'$, and $F \subseteq \tilde{F}$. 
		Recall that $\mathcal{A}$ works only with sets $F$ such that $F_0 \subseteq F$. 
		Thus $F_0 \subseteq \tilde{F}$. 
		By our assumption, $X$ must satisfy some restriction triple from $\mathcal{C}$.
		Then it also satisfies $(F_0, R_0, B_0')$.
		We now distinguish two cases based on whether $e$ belongs to $B_0'$.
		
		\textbf{Case $e \in B_0'$.}
		According to \Cref{claim:(x-S)-activeness-for-B0}, $T - R_0 + e$ contains a path from $x$ to $x'(e)$. 
		By construction, $x'(e) \in U_1 \subseteq U_2$.
		But since $X$ contains a path from $x'(e)$ to $v$, this contradicts the correctness of either $\mathcal{C}_1$ or $\mathcal{C}_2$.
		
		\textbf{Case $e \in B \setminus B_0'$.} 
		Recall that $e$ is $(x, S)$-active and $X$ contains a path from $x'(e)$ to $v$.  
		\Cref{claim:can-descend-through-lowest-outside-B0} then tells us that the edge $e' \in B \setminus B_0'$ with the lowest $x'(e')$ enjoys the same two properties.
		By definition, $x'(e') \in U_2$.
		Therefore, $X$ cannot satisfy $\mathcal{C}_2$ because there is a path in $X$ from $x'(e')$ to $v$. 
		Given that $\mathcal{C} = \mathcal{C}_1 \cup \mathcal{C}_2$, $X$ must then satisfy $\mathcal{C}_1 = \mathcal{A}(U_1, F \cup F')$. 
		Together with the structure of restriction triples that $\mathcal{A}$ can produce, this implies that for each $(\tilde{F}, \tilde{R}, \tilde{B}) \in \mathcal{C}_1$, it holds that $F \cup F' \subseteq \tilde{F}$. 
		Hence, $S \cap F' = \emptyset$ (as $S$ is disjoint from $\tilde{X}$ of the satisfied triple). 
		Because we have an $x$-active edge from $B \setminus B_0'$, the set $F'$ is defined as a subpath of $P$ from $x'(e')$ to $x$.
		However, then \Cref{claim:(x-S)-activeness-outside-B0'} states that $e'$ cannot be $(x, S)$-active, which yields a contradiction again. 
		
		In both cases, we obtain a contradiction. 
		Therefore, if $X$ satisfies $\mathcal{C}$, it contains no path from any $x \in U$ to $v$. 
		
		We now establish the converse direction. 
		Assume that $X$ satisfies $(F, R_0, B_0')$ (and therefore also $(F_0, R_0, B_0')$) and that $X$ contains no path from any vertex of $U$ to $v$.
		Let $x$ be the highest vertex in $U$ (as before, $x \neq v$).  
		In particular, $X$ does not contain a path from $x$ to $v$. 
		\Cref{claim:descend-P-using-single-blocked-edge} then guarantees that for every edge $e \in B$, either $e$ is not $(x, S)$-active, or $X$ contains no path from $x'(e)$ to $v$. 
		
		We first examine edges from $B_0'$.
		By definition, $U_1 \setminus U$ consists of such $x'(e)$ that $e \in B_0'$, $e$ is $x$-active, and $T_0 - R_0 + e$ contains a path from $x$ to $x'(e)$. 
		By \Cref{claim:(x-S)-activeness-for-B0}, all such edges $e$ are $(x, S)$-active.
		This implies that there is no path in $X$ from $x'(e)$ to $v$.  
		Consequently, $X$ contains no path from any vertex of $U_1$ to $v$. 
		We now consider two cases. 
		
		\textbf{Case $F' \cap S = \emptyset$.}
		Then $X$ satisfies $(F \cup F', R_0, B_0')$.  
		By inductive hypothesis, $X$ satisfies $\mathcal{C}_1$, and therefore $\mathcal{C}$ as well. 
		
		\textbf{Case $F' \cap S \neq \emptyset$.}
		It follows that $F'$ is nonempty. 
		Recall that $\mathcal{A}$ defines $F'$ this way only if there exists an $x$-active edge from $B \setminus B_0'$. 
		In this case, $F'$ consists of edges lying on the subpath of $P$ from $x'(e)$ to $x$, where $e \in B \setminus B_0'$ is an $x$-active edge with the lowest $x'(e)$. 
		\Cref{claim:(x-S)-activeness-outside-B0'} then tells us that $e$ is $(x, S)$-active.
		Returning to the earlier dichotomy, we conclude that $X$ contains no path from $x'(e)$ to $v$. 
		Since $U_2 = U_1 \cup \{x'(e)\}$ and we have already shown that no path exists from any vertex of $U_1$ to $v$, the inductive hypothesis for $U_2$ applies.  
		Thus $X$ satisfies $\mathcal{C}_2$, and consequently $\mathcal{C}$.
		
		In both cases, $X$ satisfies $\mathcal{C}$, which completes the converse direction.
	\end{claimproof}
	
	Therefore, the constraint set $\mathcal{C} := \mathcal{A}(\{u\}, F_0)$ satisfies all the properties required by the lemma. 
	
	Finally, we generalize this construction of $\mathcal{C}$ to arbitrary vertices $u$ and $v$.
	Consider a path from $u$ to $v$ in the underlying graph of $T$. 
	Let $L := v_1, \ldots, v_t$ be the sequence of vertices on this path, where $v_1 = u$ and $v_t = v$.  
	For each pair of consecutive vertices $v_i$, $v_{i + 1}$ such that the edge in $T$ between them is directed from $v_i$ to $v_{i + 1}$, we mark $v_i$ and $v_{i + 1}$. 
	Additionally, we mark $v_1$ and $v_t$. 
	
	We first show that if $X$ is of the form $T - S + \Inv(S)$ for some $S \subseteq T$, then any path from $u$ to $v$ in $X$ must visit all marked vertices in the order they appear in $L$.
	Obviously, any such path contains $v_1$.  
	Now, suppose that some path from $u$ to $v$ contains all marked vertices up to $v_\ell$ in the required order and does not contain $v_r$, where $v_\ell$ and $v_r$ are marked and all vertices between them are unmarked. 
	Clearly, $v_r \neq v_1$ and $v_r \neq v_t$. 
	Then $T$ contains either an edge $(v_{r - 1}, v_r)$ or an edge $(v_r, v_{r + 1})$. 
	Let $e$ be any of these edges that is present in $T$, and let $R$ and $Q$ be the weakly connected components of $T - e$ labeled so that $v_\ell \in R$ and $v_r \in Q$.
	Note that $v_t$ is then also in $Q$. 
	\Cref{claim:banned-edges-structure} states that no blocked edge goes from $R$ to $Q$. 
	Additionally, the only tree edge that goes between these components is $e$. 
	Since $X \subseteq T + B$, the only edge in $X$ leading from $R$ to $Q$ is $e$.
	Consequently, any path from $v_\ell \in R$ to $v_t \in Q$ must contain the edge $e$, and therefore the vertex $v_r$, which contradicts the definition of $v_r$.  
	
	For each pair $v_\ell$ and $v_r$ of consecutive marked vertices in $L$ (such that $v_{\ell + 1}, \ldots, v_{r - 1}$ are unmarked), we construct a separate constraint set $\mathcal{C}_{\ell, r}$, and then define $\mathcal{C}$ as their union. 
	For a given pair $v_\ell$ and $v_r$, we distinguish two cases. 
	\begin{enumerate}
		\item 
		If $\ell + 1 = r$ and $T$ contains an edge $(v_\ell, v_r)$, we define $\mathcal{C}_{\ell, r}$ as $\{(F_0, R, B_0')\}$, where $R = R_0 \cup \{(v_\ell, v_r)\}$. 
		Note that this restriction triple obeys the second property of the lemma statement. 
		
		\item 
		Otherwise, $T$ contains a path from $v_r$ to $v_\ell$. 
		Indeed, if $\ell + 1 = r$ and the previous case does not apply, then $(v_r, v_\ell) \in E(T)$. 
		If $\ell + 1 < r$, then all vertices $v_i$ with $\ell < i < r$ are unmarked, and hence the tree $T$ contains the edges $(v_i, v_{i - 1})$ and $(v_{i + 1}, v_i)$. 
		Concatenating these edges yields a path from $v_r$ to $v_\ell$.
		
		We now aim to forbid a path from $v_\ell$ to $v_r$ in $X$, having the opposite path in $T$.
		The algorithm $\mathcal{A}$ provides the corresponding constraint set. 
	\end{enumerate}
	
	It remains to verify that $\mathcal{C}$ has the first property stated by the lemma. 
	Suppose that $X = T - S + \Inv(S)$ satisfies $\mathcal{C}$. 
	According to the definition of $\mathcal{C}$, it follows that $X$ then satisfies one of the constraint sets $\mathcal{C}_{\ell, r}$.
	We claim that $X$ does not contain a path from $v_\ell$ to $v_r$.
	Given that they are consecutive marked vertices in $L$, this would imply that there is also no path from $u$ to $v$ in $X$.  
	
	If $\mathcal{C}_{\ell, r}$ comes from the second case considered above, then the absence of the path from $v_\ell$ to $v_r$ is guaranteed by the \Cref{claim:recursive-constraint-correctness}. 
	Otherwise, $X$ satisfies the only restriction triple $(F_0, R, B_0')$. 
	Consequently, $R \subseteq S$, and hence $(v_\ell, v_r) \in S$, implying that this edge is not present in $X$. 
	We denote this edge by $e$. 
	Similarly to the discussion above, $e$ is the only edge in $T + B$ that leads from the weakly connected component of $v_\ell$ to the weakly connected component of $v_r$ in $T - e$. 
	Since $e \notin E(X)$ and $X \subseteq T + B$, $X$ then does not contain a path from $v_\ell$ to $v_r$.
	
	We move on to the converse direction. 
	Suppose that $X$ satisfies $(F_0, R_0, B_0')$ and contains no path from $u$ to $v$. 
	Then for some pair of consecutive marked vertices $v_\ell$ and $v_r$, the path from $v_\ell$ to $v_r$ is also not present in $X$ (as otherwise the path from $u$ to $v$ can be assembled from these paths). 
	If $\mathcal{C}_{\ell, r}$ is obtained using $\mathcal{A}$, \Cref{claim:recursive-constraint-correctness} implies that $X$ will satisfy $\mathcal{C}_{\ell, r}$, and therefore $\mathcal{C}$. 
	Otherwise, since $X$ contains no path from $v_\ell$ to $v_r$, it cannot contain the edge $(v_\ell, v_r)$.
	Hence, $(v_\ell, v_r) \in S$. 
	Given that $X$ satisfies $(F_0, R_0, B_0')$, it follows that $X$ satisfies $(F_0, R, B_0')$ as well, which belongs to $\mathcal{C}$. 
	
	The proof is completed. 
\end{proof}

Finally, we combine this section's results into the main tool for handling allowed edges. 

\begin{lemma}\label{lemma:constraint-set}
	If $\mathcal{I} = (G, \wt, k)$ is an $\Inv$-respecting instance of $\probMASMSTshort(\mathbb{Q}_{\ge 1})$, then there exists a family of constraint sets $\{\mathcal{C}_{A'} \mid A' \subseteq A\}$ with the following properties:
	\begin{itemize}
		\item for each $S \subseteq E(T)$ and $A' \subseteq A$, if $T - S + \Inv(S) + A'$ satisfies $\mathcal{C}_{A'}$, then it is acyclic; 
		\item there exist $S \subseteq E(T)$ and $A' \subseteq A$ such that $T - S + \Inv(S) + A'$ is a solution to $\mathcal{I}$ that satisfies $\mathcal{C}_{A'}$; 
		\item for each triple $(F, R, B') \in \mathcal{C}_{A'}$, it holds that $\wt(B') \le (\wt(A) + k)^2$, and each of the sets $F$ and $R$ can be represented as a union of at most $4(\wt(A) + k + 1)^4$ directed paths in $T$.
	\end{itemize}
\end{lemma}
\begin{proof}
	For each $A'$, we construct $\mathcal{C}_{A'}$ as follows.
	Let $V'$ be the set of vertices incident to $A'$. 
	We first fix a set $B' \subseteq B$ of weight at most $(\wt(A) + k)^2$ and a set $\overline{P}$ such that the graph $G'$ from \Cref{claim:forbid-paths-for-acyclicity} constructed by $A'$ and $\overline{P}$ is acyclic. 
	Additionally, for each $(a, b) \in B'$, we fix two edges on the directed path in $T$ from $b$ to $a$ (such a path exists by \Cref{obs:basic-inverse-edge-property}).
	
	For a fixed choice of $\overline{P}$, $B'$, and pairs of edges, we start by constructing sets $F_0$ and $R_0$. 
	For each $(a, b) \in B'$, let $e_1$ and $e_2$ be the chosen edges on the path from $b$ to $a$ in $T$. 
	Assume that $e_1$ appears earlier on this path than $e_2$.
	Then we add $e_1$ and $e_2$ to $R_0$, and we add to $F_0$ all edges lying on the paths in $T$ from $b$ to the start of $e_1$ and from the end of $e_2$ to $a$.
	By construction, $(F_0, R_0, B')$ is a strict restriction triple. 
	Note that the edge sets added to $R_0$ and $F_0$ for a fixed element of $B'$ are the unions of two directed paths in $T$.
	Then $R_0$ and $F_0$ can each be represented as a union of at most $2|B'|$ directed paths in $T$.
	
	For each $(u, v) \in \overline{P}$, we apply \Cref{lemma:ban-single-path} to the vertices $u, v$ and the triple $\mathcal{R} := (F_0, R_0, B')$ and obtain a constraint set $\mathcal{C}_{\mathcal{R}, u, v}$. 
	We note that the lemma guarantees that for each $(F, R, B'') \in \mathcal{C}_{\mathcal{R}, u, v}$, we have $B'' = B'$.
	For each choice of triples $(F_{u, v}, R_{u, v}, B') \in \mathcal{C}_{\mathcal{R}, u, v}$ for all $(u, v) \in \overline{P}$, we add the triple $\left(\bigcup F_{u, v}, \bigcup R_{u, v}, B'\right)$ to $\mathcal{C}_{A'}$. 
	This completes the construction. 
	
	We now verify the required properties one by one. 
	First, suppose that a subgraph $X := T - S + \Inv(S) + A'$ with $S \subseteq E(T)$ and $A' \subseteq A$ satisfies $\mathcal{C}_{A'}$. 
	Then $X' := X - A'$ also satisfies this constraint set.
	Consider a satisfied restriction triple $\left(\bigcup F_{u, v}, \bigcup R_{u, v}, B'\right) \in \mathcal{C}_{A'}$ and the corresponding sets $\overline{P}$, $F_0$ and $R_0$.
	It follows that for each $(u, v) \in \overline{P}$, $X'$ satisfies $(F_{u, v}, R_{u, v}, B')$. 
	Let $\mathcal{R} := (F_0, R_0, B')$.
	Since $(F_{u, v}, R_{u, v}, B') \in \mathcal{C}_{\mathcal{R}, u, v}$, $X'$ also satisfies this constraint set. 
	Consequently, by \Cref{lemma:ban-single-path}, $X'$ does not contain a directed path from $u$ to $v$ for every $(u, v) \in \overline{P}$. 
	We are now in a position to apply \Cref{claim:forbid-paths-for-acyclicity}. 
	Indeed, by \Cref{lemma:inverse-edges-independence}, $X'$ is acyclic. 
	Moreover, by the choice of $\overline{P}$, the graph $G'$ constructed by $\overline{P}$ and $A'$ is also acyclic.
	Then we obtain that $X = X' + A$ is acyclic too. 
	
	We proceed to the second point of the lemma. 
	By \Cref{lemma:bounded-weight-of-B'}, there exists a solution $X$ of $\mathcal{I}$ with $S := E(T) \setminus E(X)$ such that the weight of
	\begin{equation*}
		B' := \bigcup\limits_{\substack{s_1, s_2 \in S \\ s_1 \neq s_2}} (\Inv(s_1) \cap \Inv(s_2))
	\end{equation*}
	is at most $(\wt(A) + k)^2$. 
	Moreover, since $\mathcal{I}$ is $\Inv$-respecting, we can assume that $X$ is also $\Inv$-respecting.
	Hence, $X = T - S + \Inv(S) + A'$ for some $A' \subseteq A$.
	Since $X$ is a solution of $\mathcal{I}$, it is acyclic (and hence so is $X' := X - A'$).
	Then \Cref{claim:forbid-paths-for-acyclicity} guarantees that there exists a set $\overline{P} \subseteq V' \times V'$ such that: 
	\begin{itemize}
		\item the graph $G'$ constructed by $\overline{P}$ and $A'$ is acyclic;
		\item for each $(u, v) \in \overline{P}$, $X'$ has no directed path from $u$ to $v$.
	\end{itemize}
	
	By the definition of $B'$, for each $(a, b) \in B'$, we have at least two edges $s \in S$ such that $(a, b) \in \Inv(s)$.
	Moreover, by \Cref{obs:basic-inverse-edge-property}, all such $s$ lie on the path from $b$ to $a$ in $T$.
	For each $(a, b) \in B'$, take the first and the last edge on this path that belong to $S$.
	Note that during the construction of $\mathcal{C}_{A'}$, we have considered these $\overline{P}$, $B'$, and these pairs of edges for each $(a, b) \in B'$.
	Let $F_0$ and $R_0$ be the sets constructed by $B'$ and these pairs of edges. 
	Then $X'$ satisfies $\mathcal{R} := (F_0, R_0, B')$.
	Indeed, by construction, each edge from $F_0$ does not belong to $S$, while each edge from $R_0$ is in $S$.
	The remaining property follows directly from the definition of $B'$.
	Additionally, recall that for each $(u, v) \in \overline{P}$, $X'$ does not contain a path from $u$ to $v$. 
	Then by \Cref{lemma:ban-single-path}, $X'$ satisfies $\mathcal{C}_{\mathcal{R}, u, v}$. 
	Let $(F_{u, v}, R_{u, v}, B'') \in \mathcal{C}_{\mathcal{R}, u, v}$ be a satisfied triple.
	Again, this lemma guarantees that $B'' = B'$.
	Consequently, $X'$ (and therefore $X$) satisfies $\left( \bigcup F_{u, v}, \bigcup R_{u, v}, B' \right)$, which belongs to $\mathcal{C}_{A'}$.
	
	Finally, we prove the last point of the lemma. 
	Consider a restriction triple $\left(\bigcup F_{u, v}, \bigcup R_{u, v}, B'\right) \in \mathcal{C}_{A'}$ and the corresponding sets $\overline{P}$, $F_0$ and $R_0$.
	By the choice of $B'$, we have $\wt(B') \le (\wt(A) + k)^2$.
	We show the needed property for the set $F := \bigcup F_{u, v}$, the argument for $\bigcup R_{u, v}$ is identical. 
	Let $\mathcal{R} := (F_0, R_0, B')$.
	Since for each $(u, v) \in \overline{P}$, we know that $(F_{u, v}, R_{u, v}, B') \in \mathcal{C}_{\mathcal{R}, u, v}$, where $\mathcal{C}_{\mathcal{R}, u, v}$ is constructed using \Cref{lemma:ban-single-path}, it follows that $F_0 \subseteq F_{u, v}$. 
	If $\overline{P} = \emptyset$, then $F = \emptyset$ as well. 
	Otherwise, we have 
	\begin{equation*}
		F = \left(\bigcup (F_{u, v} \setminus F_0)\right) \cup F_0
	\end{equation*}
	By \Cref{lemma:ban-single-path}, each $F_{u, v} \setminus F_0$ can be represented as a union of at most $|B'| + 1$ directed paths in $T$. 
	Since $\overline{P} \subseteq V' \times V'$, it holds that $|\overline{P}| \le |V'|^2 \le (2|A'|)^2$ (last inequality follows from the definition of $V'$).
	Recall that $F_0$ is a union of at most $2|B'|$ directed paths in $T$.
	This implies that $F$ is a union of at most
	\begin{equation*}
		4|A'|^2 \cdot (|B'| + 1) + 2|B'|
	\end{equation*}
	directed paths in $T$. 
	Because all edge weights in $G$ are at least one, we have $|B'| \le \wt(B') \le (\wt(A) + k)^2$ and $|A'| \le |A| \le \wt(A)$. 
	Consequently, we obtain the following bound:
	\begin{equation*}
		\begin{split}
			4|A'|^2 \cdot (|B'| + 1) + 2|B'| & \le 4\wt(A)^2 \cdot ((\wt(A) + k)^2 + 1) + 2(\wt(A) + k)^2 \\ & \le 4(\wt(A) + k)^2 \cdot ((\wt(A) + k)^2 + 1) + 4(\wt(A) + k)^2 \\ & \le 4((\wt(A) + k)^2 + 1)^2 \\ & \le 4(\wt(A) + k + 1)^4
		\end{split}
	\end{equation*}
	The proof is completed. 
\end{proof}

Although the proofs of the two preceding lemmas are constructive, the final algorithm does not need to build the corresponding constraint sets.
Instead, because the sets $B'$, $F$ and $R$ enjoy the special structure, we can enumerate all triples with the required properties. In particular, this enumeration covers every element of the constraint set.

\subsection{Summing up}

We are ready to prove our second main result, \Cref{thm:rational-xp}. For convenience, we restate the theorem below. 

\thmXPRat*
\begin{proof}
	Let $\mathcal{I} = (G, \wt, k)$ be an instance of $\probMASMSTshort_{\mathbb{Q}_{\ge 1}}$. 
	If $\wt(A) \ge 3k$, then the algorithm from  \Cref{lemma:allowed-edges} constructs a solution to $\mathcal{I}$ in polynomial time. 
	From now on, we assume that $\wt(A) < 3k$.
	
	We iterate over all possible choices of a set $D \subseteq B$ of weight at most $\wt(A)$, a set $B' \subseteq B \setminus D$ of weight at most $(\wt(A) + k)^2$, and sets $P_F, P_R \subseteq V \times V$, each of size at most $4(\wt(A) + k + 1)^4$. 
	Additionally, we require that for every $(u, v)$ from $P_F \cup P_R$, $T$ contains a directed path from $u$ to $v$. 
	For a fixed choice of $D$, $B'$, $P_F$ and $P_R$, we begin by constructing the sets 
	$$
	\begin{aligned}
		& F = \bigcup_{(u, v) \in P_F} E_{u, v},\\ 
		& R = \bigcup_{(u, v) \in P_R} E_{u, v},
	\end{aligned}
	$$
	where $E_{u, v}$ denotes the set of edges on the path from $u$ to $v$ in $T$. 
	
	Let $G' = G - D$, and let $\wt'$ be the restriction of $\wt$ to $E(G')$.
	Since $D$ consists only of blocked edges, we know that $T \subseteq G'$, and hence $\MSF(G', \wt') = T$.
	Therefore, $F, R \subseteq E(\MSF(G', \wt'))$ by the definition of these sets. 
	
	We check whether $B' \subseteq \Inv(G', \wt', R)$, and if so, apply the algorithm from \Cref{lemma:restricted-max-IS} to $(G', \wt')$ and the triple $(F, R, B')$, which returns a set $S \subseteq E(T)$.
	For each $A' \subseteq A$, we check whether the subgraph $T - S + \Inv(S) + A'$ is a solution to $\mathcal{I}$.
	
	If no check succeeds over all choices of $D$, $B'$, $P_F$, $P_R$, and $A'$, we report that $\mathcal{I}$ is a no-instance.
	
	\medskip\noindent\textbf{Running time.}
	We note that for a fixed choice of the sets $D$, $B'$, $P_F$, $P_R$, and $A'$, our algorithm works in polynomial time. 
	Indeed, we first construct $F$, $R$, and $G'$, then we check that $B' \subseteq \Inv(G', \wt', R)$, apply the algorithm from \Cref{lemma:restricted-max-IS}, and finally verify that $T - S + \Inv(S) + A'$ is a feasible solution to $\mathcal{I}$.
	All these parts are polynomial. 
	Therefore, it remains to bound the number of choices of the sets $D$, $B'$, $P_F$, $P_R$ and $A'$. 
	
	Recall that all edge weights are at least one.
	Hence, if the weight of a set is bounded by some value, then the same bound applies for its size.
	Since $|B| \le m \le n^2$, there are at most $n^{2\wt(A)}$ ways to choose $D$ and at most $n^{2(\wt(A) + k)^2}$ possible sets $B'$. 
	Similarly, because $|E(T)| = n - 1$, we have at most $n^{4(\wt(A) + k + 1)^4}$ ways to choose each of the sets $P_F$ and $P_R$. 
	Finally, $|A| \le \wt(A)$, and hence there are at most $2^{\wt(A)}$ possible sets $A'$. 
	Combining these bounds with $\wt(A) < 3k$, we obtain that there are $n^{\mathcal{O}\left(k^4\right)}$ ways to choose these sets. 
	
	\medskip\noindent\textbf{Correctness.}
	By construction, our algorithm outputs only feasible solutions. 
	Therefore, it remains to check that it does not report false negatives. 
	Suppose that $\mathcal{I}$ is a yes-instance.
	In this case, by \Cref{claim:inv-respecting-sol}, it has a solution $X$ with $S = E(T) \setminus E(X)$ and $D = \Inv(G, \wt, S) \setminus E(X)$ such that $\wt(D) \le \wt(A)$. 
	Let $G' = G - D$, and let $\wt'$ be the restriction of $\wt$ to $E(G')$. 
	By the definition of $D$, it holds that $E(X) \cap D = \emptyset$, and hence $X$ is a subgraph of $G'$. 
	
	Since $D \subseteq B$, we have $\MSF(G', \wt') = T$. 
	Because allowed, blocked, and inverse edges are defined based on the maximum spanning tree, it follows that $A(G', \wt') = A(G, \wt)$, $B(G', \wt') = B(G, \wt) \setminus D$, and $\Inv(G', \wt', e) = \Inv(G, \wt, e) \setminus D$ for each $e \in E(T)$. 
	We claim that $X$ is an $\Inv$-respecting subgraph in $G'$. 
	Since $X$ is a solution to $\mathcal{I}$, it is acyclic. 
	By applying \Cref{obs:blocked-edges-in-acyclic-subgraph} to the graph $G'$, we know that $X$ can contain blocked edges only from the set $\Inv(G', \wt', S) = \Inv(G, \wt, S) \setminus D$.
	Additionally, by the definition of $D$, we have $\Inv(G, \wt, S) \setminus E(X) = D$, which implies $\Inv(G', \wt', S) \subseteq E(X)$.
	Therefore, $X = T - S + \Inv(G', \wt', S) + A'$ for some $A' \subseteq A$. 
	Note that $X$ is also a solution to $\mathcal{I}' = (G', \wt', k)$. 
	Consequently, $\mathcal{I}'$ is an $\Inv$-respecting instance. 
	
	\Cref{lemma:constraint-set} provides a family of constraint sets $\{\mathcal{C}_{A'} \mid A' \subseteq A\}$ for the instance $\mathcal{I}'$. 
	Moreover, there exist $S' \subseteq E(T)$ and $A' \subseteq A$ such that $X' = T - S' + \Inv(G', \wt', S') + A'$ is a solution to $\mathcal{I}'$ that satisfies the constraint set $\mathcal{C}_{A'}$. 
	Let $(F, R, B') \in \mathcal{C}_{A'}$ be the satisfied constraint triple. 
	By definition, we have $B' \subseteq B(G', \wt') = B(G, \wt) \setminus D$.
	Additionally, the lemma guarantees that $\wt(B') \le (\wt(A) + k)^2$, and each of the sets $F$ and $R$ equals the union of at most $4(\wt(A) + k + 1)^4$ directed paths in $T$. 
	Let $P_F$ and $P_R$ be the sets of endpoints of such paths for $F$ and $R$, respectively. 
	
	From the definition of $D$, $B'$, $P_F$, and $P_R$, it follows that our algorithm must have considered these sets. 
	Moreover, given $P_F$ and $P_R$, it constructed exactly the sets $F$ and $R$.
	Since $(F, R, B')$ is a restriction triple, we have $B' \subseteq \Inv(G', \wt', R)$.
	Hence, our algorithm applied \Cref{lemma:restricted-max-IS} to $(G', \wt')$ and constraints $(F, R, B')$ and obtained a set $S''$. 
	Since $S'$ is also feasible under these constraints and $S''$ has a maximum profit, we have $\Profit(G', \wt', S'') \ge \Profit(G', \wt', S')$.
	
	Consider the subgraph $X'' = T - S'' + \Inv(G', \wt', S'') + A'$. 
	Since $S''$ is feasible under the constraints given by $(F, R, B')$, $X''$ satisfies this restriction triple, and therefore it also satisfies $\mathcal{C}_{A'}$.
	Hence, \Cref{lemma:constraint-set} guarantees that $X''$ is acyclic. 
	Note that 
	$$
	\wt(X'') = \wt(T) + \Profit(G', \wt', S'') + \wt(A') \ge \wt(T) + \Profit(G', \wt', S') + \wt(A') = \wt(X')
	$$
	Since $X'$ is a solution to $\mathcal{I}'$, we have $\wt(X') \ge \wt(T) + k$, and therefore $\wt(X'') \ge \wt(T) + k$. 
	Consequently, $X''$ is a solution to $\mathcal{I}$, which is found by our algorithm. 
\end{proof}

\section{Conclusion}\label{sec:concl}

We conclude with several open questions and directions for further research.
First of all, the parameterized complexity of \probMASMSTshort remains unsettled, when we deal with rational weights.

\begin{problem}
	Is $\probMASMSTshort(\mathbb{Q}_{\ge 1})$ \classW{1}-hard or fixed-parameter tractable?
\end{problem}

Another standalone question concerns \Cref{lemma:allowed-edges}, and it might be of independent interest.
\Cref{lemma:allowed-edges} in fact shows that a solution of weight at least $\wt(T)+ \lceil\wt(A)/3 \rceil$ is guaranteed for integer weights and can be found in polynomial time, using known constructions of linear orderings for oriented trees.
Can we improve the denominator $3$ in this guarantee using more specific constructions?
If yes, can we find the required subset of allowed edges in polynomial time?
If no, is there a no-instance with integer weights and $\wt(A)=3k-3$?

\begin{problem}
		Can the constant $3$ in \Cref{lemma:allowed-edges} be improved?
\end{problem}

Our positive results for the $\MaxST$ guarantee suggest that \probMAS above the Poljak-Turz\'{i}k guarantee could also admit efficient algorithms.
We re-formulate the corresponding open question, which was already posed in  \cite{etscheid2018linear} for integral weights and for a broader class of problems.

\begin{problem}
	Can we find an acyclic subgraph of weight at least $$\wt(G)/2+\wt(\MinST(G,\wt))/4+k$$ in a directed graph $G$ with positive integer edge weights $\wt: E(G)\to \mathbb{Z}_{\ge 1}$ in \classFPT-time?
\end{problem}

We note that for the \textsc{Maximum Cut} problem, the Poljak--Turz\'{i}k guarantee also applies, and in \cite{MaxCutPT} Lill, Petrova and Weber have shown that a cut of weight at least $\wt(E(G))/2+\wt(\MinST(G,\wt))/4+k$ can be found in $2^{\Oh(k)}\cdot |G|$ time, if $G$ has integer edge weights.
Similarly, the weight of the maximum spanning tree is also a viable lower bound for the maximum cut weight in $G$.
This gives rise to the following natural open problem.

\begin{problem}
	Can we find a cut of weight at least $\wt(\MaxST(G,\wt))+k$ in an undirected graph $G$ with positive integer edge weights $\wt: E(G)\to \mathbb{Z}_{\ge 1}$ in \classFPT-time?
\end{problem}

For the last two open questions, the weights could as well be replaced with rationals not less than one.

One can also study the existence of efficient kernelization algorithms.
The main result of \cite{etscheid2018linear} demonstrates that many problems parameterized above the Poljak--Turz\'{i}k bound admit linear kernels, in the unweighted (oriented graph) case.
In \cite{MaxCutPT}, the authors show that \textsc{Maximum Cut} parameterized above Poljak--Turz\'{i}k is in \classFPT in the case of integral weights.
Can we improve this result and prove that it admits polynomial kernels?
Can we do the same for \Cref{thm:integral-fpt} and show that $\probMASMSTshort(\mathbb{Z}_{\ge 1})$ admits polynomial kernels?
We summarize this direction in the following question.

\begin{problem}
	Do polynomial kernels exist for any relevant weighted problem parameterized above the Poljak--Turz\'{i}k or the maximum spanning tree guarantees?
\end{problem}

Obtaining a negative answer to this question for some weighted parameterized problem, that, on the other hand, admits an \classFPT-algorithm, would also be very interesting.

\bibliographystyle{alpha}
\bibliography{ref}

\end{document}